\documentclass[12pt]{article}
\usepackage{graphicx} 
\usepackage[T1]{fontenc}
\usepackage[utf8]{inputenc}
\usepackage{amsthm}
\usepackage{amsmath}
\usepackage{amssymb,bm}
\usepackage{array}
\usepackage{natbib}
\usepackage{xcolor}
\usepackage{float}
\usepackage{makecell}
\usepackage{subcaption}
\usepackage{multirow}
\usepackage{booktabs}
\graphicspath{{Fig/}}

\def \aa {\hbox{\boldmath$a$}}

\def \uu {\hbox{\boldmath$u$}}

\def \xx {\hbox{\boldmath$x$}}
\def \yy {\hbox{\boldmath$y$}}
\def \zz {\hbox{\boldmath$z$}}
\def \ww {\hbox{\boldmath$w$}}

\def \bbeta {\hbox{\boldmath$\beta$}}

\def \G {{\mathcal G}}

\newcommand{\U}{{\mathcal U}}

\def \T {{\mathcal T}}

\theoremstyle{thmstyleone}%
\newtheorem{theorem}{Theorem}
\theoremstyle{thmstyletwo}%
\theoremstyle{thmstylethree}%

\newcommand\blfootnote[1]{%
  \begingroup
  \renewcommand\thefootnote{}\footnote{#1}%
  \addtocounter{footnote}{-1}%
  \endgroup
}

\begin{document}

\def\spacingset#1{\renewcommand{\baselinestretch}%
{#1}\small\normalsize} \spacingset{1}

\date{$^{1}$Center of Operations Research, Miguel Hern\'{a}ndez University of Elche, Spain\\
			$^{2}$Department of Economics and Management, University of Pisa, Italy}

  \title{\bf Temporal M-quantile models and robust bias-corrected small area predictors}
  \author{ 
    Mar\'{\i}a Bugallo$^{1}$
    Domingo Morales$^{1}$
    \\
    Nicola Salvati$^{2}$
    Francesco Schirripa Spagnolo$^{2}$
    \blfootnote{Supported by the projects PID2022-136878NB-I00 (Spain), PROMETEO-2021-063 (Valencia, Spain); PRIN-Quantification in the Context of Dataset Shift - QuaDaSh (Grant P2022TB5JF, Italy); MAPPE, Bando a Cascata Programme “Growing Resilient, Inclusive and Sustainable (GRINS)” PE0000018.}
    }
  \maketitle

\begin{abstract}
In small area estimation, it is a smart strategy to rely on data measured over time. However, linear mixed models struggle to properly capture time dependencies when the number of lags is large. Given the lack of published studies addressing robust prediction in small areas using time-dependent data, this research seeks to extend M-quantile models to this field. Indeed, our methodology successfully addresses this challenge and offers flexibility to the widely imposed assumption of unit-level independence. Under the new model, robust bias-corrected predictors for small area linear indicators are derived. Additionally, the optimal selection of the robustness parameter for bias correction is explored, contributing theoretically to the field and enhancing outlier detection. For the estimation of the mean squared error (MSE), a first-order approximation and analytical estimators are obtained under general conditions. Several simulation experiments are conducted to evaluate the performance of the fitting algorithm, the new predictors, and the resulting MSE estimators, as well as the optimal selection of the robustness parameter. Finally, an application to the Spanish Living Conditions Survey data illustrates the usefulness of the proposed predictors. 
\end{abstract}

\noindent%


\spacingset{1.5} 
   \addtolength{\textheight}{.2in}
   
\section{Introduction}
\label{sec:intro}
Estimates of finite population parameters for subgroups, such as geographic areas or socio-economic groups, are increasingly required for better planning and evaluation of government programs. However, National Statistical Offices (NSOs) typically design surveys to obtain accurate estimates at a certain level of aggregation. When the goal is to ensure valid inference for specific subgroups of the population, where the portion of available data is small, direct estimation 
may be unreliable. This is regarded as a small area estimation (SAE) problem, which can be addressed by using statistical models that link the sample data with auxiliary covariate information (e.g., administrative records or census data) to improve the survey estimators in the areas of interest. SAE models can be grouped into two broad categories: area-level models, which relate design-based direct estimates to area-specific covariates, and unit-level models, which use individual survey responses as the target variables rather than direct estimates.

In this paper, we focus on unit-level models. Under this approach, the Empirical Best Linear Unbiased Predictors (EBLUP) based on linear mixed models (LMM) are widely used in SAE. These models typically incorporate area-specific random effects to account for the between-area variability that is not explained by the covariates. For a review of these SAE models, see \cite{rao2015} and \cite{morales2021}. Most of these models use cross-sectional data, but many surveys are repeated over time. Consequently, it is possible to improve the predictive performance and reduce the variability of the estimates by considering historical data, i.e., by taking into account potential time correlations.

The effective use of past information and modeling of temporal dependencies is an appealing method for borrowing strength in  SAE.
In this regard, several authors have developed theoretical  advances in SAE models that account for temporal correlation. Temporal adaptations of LMMs have been proposed by \cite{hobza2018},  \cite{morales2019} and \cite{guadarrama2021}. However, these models often rely on strong distributional assumptions and require the formal specification of the dependence structure of the random effects. 

The M-quantile (MQ) approach \citep{breckChamb88, Cha06} is considered a valuable alternative in SAE for relaxing some conventional modeling assumptions and obtaining predictors that are robust against outlying observations. Nonetheless, standard MQ models require the assumption of unit-level independence. Therefore, the first aim of this paper is to extend MQ models to time-dependent target variables. Specifically, we propose a Time-Weighted MQ (TWMQ) regression, inspired by the MQ Geographically Weighted Regression (GWR) proposed by \cite{salvati2012}. As an inherent property of MQ models, our approach avoids distributional assumptions and allows for characterizing differences between areas, as well as time dependencies, through data-driven estimation of the regression coefficients. Consequently, the new models feature time-varying parameters and remain distribution-free for both areas and time.

The second aim of this paper is to use the TWMQ regression to develop small area estimators. As the traditional MQ approach, our proposal can be considered a robust-projective approach based on plug-in robust prediction, i.e. the optimal, but outlier-sensitive, parameter estimates are replaced by outlier robust versions. Unfortunately, even though these methods usually lead to a low prediction variance, they may also involve an unacceptable prediction bias \citep{chambers1986, Chambers14, Jio13}. To solve this issue, we propose robust bias-corrected predictors based on the TWMQ regression. As a result, the introduction of a bias correction term may increase the variability of the predictors. Generally speaking, the robust bias-corrected predictor is obtained by incorporating a second influence function that depends on a tuning constant, called robustness parameter. Its choice is crucial because it allows a compromise between bias and variance. To the best of our knowledge, no published studies have proposed optimal criteria for accurately selecting selecting the robustness parameter in the MQ approaches.
Therefore, we have defined an optimal criterion aimed at reducing bias without unbalancing the MSE.
The proposed idea can be applied all the MQ model-based predictors, not only to the predictors derived from our model. Furthermore, its potential role in the detection of outliers has been explored.

The effectiveness of the proposed model and the performance of the candidate predictors have been examined through model-based simulation experiments. We have also investigated the automatic selection of the robustness parameters and the performance of the MSE estimators. Empirical results presented in this paper demonstrate significant improvements. 
Finally, we applied the proposed TWMQ approach to data from the 2013-2022 Spanish Living Conditions Survey to assess changes in the average income levels across the provinces of Empty Spain.

The paper is organized as follows. Section \ref{sec.MQfun} introduces the MQ functions. Section \ref{MQ} reviews the basic theory of MQ models for SAE, focusing on the most relevant aspects in order to subsequently define the temporal model.
Section \ref{TWMQ} adapts the three-level MQ linear regression to time-dependent target variables, it derives robust small area predictors and analytical MSE estimators, and presents a data-driven criterion for an optimal selection of the robustness parameter for bias correction. Section \ref{sec.sim} includes several model-based simulations to investigate the performance of the novel statistical methodology. Section \ref{sec.aplic} applies the proposed approach to estimate the average level of income in Empty Spain. Section \ref{sec.conc} summarises the main findings and lines of future research. The paper includes Supplementary Material organized in four sections.
Section A derives a first-order approximation of the MSE of the robust bias-corrected MQ predictor based on the TWMQ models and proposes analytical estimators. Section B presents an optimal criterion for the selection of the robustness parameter and proves the existence and uniqueness of the solution. Section C presents further results from the model-based simulation experiments. Section D contains additional results from the application to real data.
%
%
\section{M-Quantile functions}\label{sec.MQfun}
Let $Y$ be a random variable with cumulative distribution function (c.d.f.) $F_Y(y)=P(Y\leq y)$, $y\in \mathbb{R}$, and standard deviation $\sigma_Y>0$. Let
us define the $(q,\sigma_q,\psi)$-check function
\begin{equation*}
\rho_{q}(u,\sigma_q)=2\sigma_q\big|{q}-I_{(-\infty,0)}(u)\big|\rho(\sigma_q^{-1}u),\quad  u\in \mathbb{R},\,\, \sigma_q>0, \quad 0<q<1,
\end{equation*}
where $\rho(u)$ is a continuously differentiable loss function and $\psi(u)=\rho(u)/du$.
The partial derivative of $\rho_{q}(u,\sigma_q)$, with respect to $u$, is
\begin{equation}\label{psiqu}
\psi_{q}(u,\sigma_q)=\frac{\partial\rho_{q}(u,\sigma_q)}{\partial u}
=2\Big\{{q} I_{(0,\infty)}(u)+(1-{q}) I_{(-\infty,0]}(u)\Big\}\,\psi(\sigma_q^{-1}u).
\end{equation}
The MQ function of order $q$ of $Y$ \citep{breckChamb88} is calculated
as
$$
Q_q(Y;\sigma_q,\psi)=\underset{\xi_{q}\in \mathbb{R}}{\mbox{solution}}\,\bigg\{ \int_{\mathbb{R}}\psi_{q}(y-\xi_{q},\sigma_q)\, dF_Y(y)=0\bigg\},\quad 0<q<1.
$$
To complete the definition of $\psi_{q}(u,\sigma_q)$ in (\ref{psiqu}), we take $\sigma_q=\sigma_Y$ and use the Huber function
\begin{eqnarray}\label{psiHuber}
\psi(u)&=&u I_{(-c_\psi,c_\psi)}(u)+c_\psi\, \mbox{sgn}(u) I_{(-\infty,-c_\psi]\cup[c_\psi,\infty)}(u), \ u\in\mathbb R;\quad c_\psi>0.
\end{eqnarray}
\section{Three-level M-quantile linear models}\label{MQ}
Let $U$ be a finite population of size $N$ hierarchically partitioned in domains $U_d$ and subdomains $U_{dt}$ of sizes $N_{d}$ and $N_{dt}$, respectively, $d=1,\ldots,D$, $t=1,\ldots,T$. Let $s$, $s_d$ and $s_{dt}$ be the corresponding sampled subsets of sizes $n$, $n_d$ and $n_{dt}$, respectively. Throughout the theoretical part, we assume that a vector of $p\geq 1$ unit-level auxiliary variables ${\xx}_{dtj}^{\prime}=(x_{dtj1},\ldots,x_{dtjp})$ is known for all individuals in $U_{dt}$ and the variable of interest, $y_{dtj}$, is only observed for individuals in $s_{dt}$. 
For $0<q<1$, the three-level MQ linear regression (MQ3) model \citep{Marchetti2018} is
\begin{equation}\label{MQ3}
    y_{dtj}={\xx}_{dtj}^{\prime}\bbeta_\psi(q)+ e_{\psi,dtj}(q),\quad d=1,\ldots,D,\, t=1,\ldots,T,\, j=1,\ldots,N_{dt},
\end{equation}
where $\bbeta_\psi(q)=(\beta_{\psi1}(q),\ldots,\beta_{\psi p}(q))^\prime$ is the vector of regression coefficients at $q$ and
$e_{\psi,dtj}(q)$ are independent model errors with unknown c.d.f. $F_{q}(u)=P(e_{\psi, dtj}(q)\leq u|{\xx}_{dtj})$, $u\in\mathbb R$. 

For the MQ function, we assume that
$Q_{q}(e_{\psi,dtj}(q);\sigma_q,\psi|{\xx}_{dtj})=0$, 
and, although no explicit parametric assumptions are being made, the homoscedasticity assumption $\sigma_q=\mbox{var}^{1/2}(e_{\psi,dtj}(q))=\sigma_{\psi}(\bbeta_\psi(q))$ is imposed.
It is worth noting that $\bbeta_\psi(q)$ only varies with the order of the quantile.
Indeed, models \eqref{MQ3} can be expressed with two subscripts, $d$ and $i$, where $i$ runs through all combinations of values of the indexes $t$ and $j$. However, the adopted notation is necessary from Section \ref{sec.res} onwards.
%
For the Huber function, it is customary to set $c_\psi = 1.345$ because it guarantees 95\% efficiency when the errors are standard normal, and still offers protection against outliers \citep{Holland1977}.

The regression vector $\bbeta_\psi(q)$ can be estimated solving the following equations
\begin{equation}\label{hatbeta3Mquantile}
    \hat\bbeta_\psi(q)=\underset{\bbeta_\psi(q)\in \mathbb{R}}{\mbox{solution}}\,
    \bigg\{\sum_{d=1}^D\sum_{t=1}^T\sum_{j=1}^{n_{dt}}\psi_{q}\big(y_{dtj}-{\xx}_{dtj}^{\prime}\bbeta_\psi(q),\hat\sigma_q\big) x_{dtjk}=0,\, k=1,\ldots,p \bigg\}.
\end{equation}
In practice, the system \eqref{hatbeta3Mquantile} is solved using the iterative re-weighted least squares (IRLS) algorithm, which ensures convergence to a unique solution \citep{Bianchi2015}.
%
%
%
\subsection{MQ approach for inter-area variability}\label{MQ3.to.SAE}
The MQ modelling approach to SAE is described in \cite{Cha06} and just an overview of the adaptation to three-level nested data is presented here.
The idea is to non-parametrically capture the variability of the population, beyond what is explained by the auxiliary variables, using the so-called MQ coefficients \citep{dawber:2019}. For $j=1,\ldots,N_{dt}$, the unit-level MQ coefficients of models \eqref{MQ3} are
\begin{align}\label{qdtj0}
    {q}_{dtj}=\underset{0<q<1}{\mbox{solution}}\,\bigg\{ Q_{q}(y_{dtj};\sigma_q,\psi|{\xx}_{dtj})=y_{dtj}\bigg\},\
    Q_{q}(y_{dtj};\sigma_q,\psi|{\xx}_{dtj})={\xx}_{dtj}^{\prime}\bbeta_\psi(q).
\end{align}
It holds that  $y_{dtj}=Q_{{q}_{dtj}}(y_{dtj};\sigma_{{q}_{dj}},\psi|\xx_{dtj})=\xx_{dtj}^\prime\bbeta_\psi({q}_{dtj})$,  $\sigma_{{q}_{dtj}}={\sigma}_{\psi}(\bbeta_\psi({q}_{dtj}))$.
The unit-level MQ coefficient $q_{dtj}$ is the ``most likely'' quantile of unit $j$.
That is, of all the MQ3 models that vary by $0<q<1$, the model with $q=q_{dtj}$ would predict $y_{dtj}$ without error
if $\bbeta_\psi(q)$ was known. 
Using the sample observations, an estimator of ${q}_{dtj}$ is
$$
\hat{q}_{dtj}=\underset{0<q<1}{\mbox{solution}}\,\bigg\{\hat{Q}_{q}(y_{dtj};\hat\sigma_q,\psi|{\xx}_{dtj})=y_{dtj}\bigg\},\ \
\hat{Q}_{q}(y_{dtj};\hat\sigma_q,\psi|{\xx}_{dtj})={\xx}_{dtj}^{\prime}\hat\bbeta_\psi(q).
$$
The domain population and sample means of unit-level MQ coefficients are
\begin{align}\label{qd..}
    \theta_d\triangleq \bar{q}_{d..}=\frac{1}{N_{d}}\sum_{t=1}^T\sum_{j=1}^{N_{dt}}q_{dtj}, \quad \hat\theta_d\triangleq \hat{\bar{{q}}}_{d..}=\frac{1}{n_{d}}\sum_{t=1}^T\sum_{j=1}^{n_{dt}}\hat{q}_{dtj}, \quad d=1,\ldots,D,
\end{align}
respectively. We expect that the MQ3 model with $q=\hat\theta_d$ will be the one that provides the best predictions in the domain $d$. Then, in order to predict $y_{dtj}$ in the unobserved part of domain $d$,  we will choose the MQ3 model with $q=\hat\theta_d$.
\subsection{Robust small area predictors for MQ3 models}\label{sec.mq}
The MQ3 models can be used to predict subdomain quantities including but not limited to, population means.
Let $r_{dt}=U_{dt}-s_{dt}$ be the non sampled subset of $U_{dt}$.
If we assume that the sum of the residuals $e_{\psi,dtj}(\theta_d)$ in $r_{dt}$ is close to zero, then 
\begin{align}
    \label{Y.dt}
    \overline{Y}&_{dt}=\frac{1}{N_{dt}}\sum_{j=1}^{N_{dt}}y_{dtj}=
    \frac{1}{N_{dt}}\bigg\{\sum_{j\in s_{dt}}y_{dtj}+\sum_{j\in r_{dt}}\xx_{dtj}^\prime\bbeta_\psi(\theta_d)
    +\sum_{j\in r_{dt}}e_{\psi,dtj}(\theta_d)\bigg\}
    \\
    \nonumber
    &\approx\frac{1}{N_{dt}}\bigg\{\sum_{j\in s_{dt}}y_{dtj}+\sum_{j\in r_{dt}}\xx_{dtj}^\prime\bbeta_\psi\big(\hat\theta_d\big)\bigg\}
    +\frac{1}{N_{dt}}\sum_{j\in r_{dt}}\xx_{dtj}^\prime\frac{\partial\bbeta_\psi(q)}{\partial q}\Big|_{q=\hat\theta_d}(\theta_d-\hat\theta_d).
\end{align}
Typically, the second summand of the last expression is much smaller than the first one.
Therefore, we define the MQ predictor of $\overline{Y}_{dt}$ as
\begin{align}\label{mq}
    \hat{\overline{Y}}_{dt}^{mq}=\frac{1}{N_{dt}}
    \bigg\{\sum_{j\in s_{dt}}y_{dtj}+\sum_{j\in r_{dt}}\xx_{dtj}^\prime\hat\bbeta_\psi\big(\hat\theta_d\big)\bigg\}.
\end{align}
Given that $Q_{\theta_d}(e_{\psi,dtj}(\theta_d);\sigma_{\theta_d },\psi|{\xx}_{dtj})=0$, it follows that
$\sum_{j\in r_{dt}}e_{\psi,dtj}(\theta_d)\approx 0$ if $\theta_d\in(1/2-\varepsilon, 1/2+\varepsilon)$ for some small $\varepsilon>0$,
but not otherwise.
This bias can be estimated as
$$\hat B\big(\hat{\overline{Y}}_{dt}^{mq}\big)= \hat E\big[\hat{\overline{Y}}_{dt}^{mq}-\overline{Y}_{dt}\big]=-\frac{1}{n_{dt}}\Big(1-\frac{n_{dt}}{N_{dt}}\Big)
\sum_{j\in s_{dt}}\hat e_{\psi,dtj}(\hat\theta_d).
$$
Based on a robustification of $\hat e_{\psi,dtj}(\hat\theta_d)$, the robust bias-corrected (BMQ) predictor  is
\begin{eqnarray}\label{bmq}
    \hat{\overline{Y}}_{dt}^{bmq}=\hat{\overline{Y}}_{dt}^{mq}+\frac{1}{n_{dt}}\Big(1-\frac{n_{dt}}{N_{dt}}\Big)\sum_{j\in s_{dt}}\sigma_{\theta_d}
    \phi\big(\sigma_{\theta_d}\hat e_{\psi,dtj}(\hat\theta_d)\big),
\end{eqnarray}
where $\phi$ is an influence function with robustness parameter $c_\phi\geq 0$.
The last summand on the right-hand side of (\ref{bmq}) has the role of controlling the potential bias. The characterization of $\phi$ is worthy of comment. By setting the value of $c_\phi$, it is possible to trade bias for variance in robust bias-corrected predictors.
If $c_{\phi}=0$, the BMQ predictor reduces to the MQ predictor.
As $c_{\phi}$ increases, larger residuals have larger weights, biased by the MQ predictor. Consequently, larger values of $c_{\phi}$ lead to less bias, but also less robustness and more variability.
It stands to reason that it is crucial to propose optimality criteria for a proper selection of $c_{\phi}$. For brevity, this is done in Section \ref{sec.c}  for predictors derived from the TWMQ models but an analogous argument can be applied to the BMQ predictor derived from the MQ3 models.
\subsection{Residual analysis and inter-period weights}\label{sec.res}
Let define the subdomain-level residuals $r_{\psi,dt}=\frac{1}{n_{dt}}\sum_{j\in s_{dt}}\hat{e}_{\psi,dtj}(\hat{\theta}_d)$, $d=1,\ldots,D$,  $t=1,\ldots,T$, and assume that there exists
some unknown subdomain-level temporal dependency between the target variables $y_{dt_1j_1}$ and $y_{dt_2j_2}$, $t_1,t_2=1,\ldots,T$, $j_1=1,\ldots,N_{dt_1}$, $j_2=1,\ldots,N_{dt_2}$. In such case, the subdomain-level temporal dependency will remain in the subdomain-level residuals
$r_{\psi,dt_1}$ and $r_{\psi,dt_2}$. As already pointed out, the MQ3 models are able to non-parametrically capture area-level variability by fitting different regression surfaces $\hat{Q}_{\hat\theta_d}(y_{dtj};\hat\sigma_q,\psi|{\xx}_{dtj})={\xx}_{dtj}'\bbeta_{\psi}(\hat\theta_d)$,
but they are not able to model temporal dependencies.
This motivates us to fit a seasonal autoregressive $AR(P)_D$ model with period $D$ and order $0\leq P\leq T$ to $\{r_{\psi,dt}:  d=1,\ldots,D,\ t=1,\ldots,T\}$.
Here the domains play the role of seasons and $P$ measures how the auto-correlation decays over time.
The extreme cases are $P=0$, where there is no autoregressive dependence structure, and $P=T$ which includes all possible lags. Our idea is to define inter-period weights $w_{t_1t_2}$ that measure the dependency between $r_{\psi,dt_1}$ and $r_{\psi,dt_2}$, 
based on the estimated coefficients of the fitted seasonal autoregressive $AR(P)_D$ model. We differentiate between two cases:
\begin{enumerate}
    \item[(1)]
    If $1\leq t\leq P$, there is past information from the first period up to time $t$, so the weights will be distributed over $\{1,\ldots,t\}$.
    Accordingly, if $t=1$ there is no past information, i.e., only data from the first period is available.
    \item[(2)]
    If $P<t\leq T$, there is past information from the first period up to time $t$, but we only assign positive weights to the last $P$ delays, i.e.
    $\{t-P,\ldots,t\}$.
\end{enumerate}
Then, the equation of an $AR(P)_D$ process $\{z_i\}_{i\in \mathbb Z}$ is $z_i=\phi_0+\phi_1 z_{i-D}+\phi_2 z_{i-2D}+\dots+\phi_P z_{i-PD}+a_{i}$, where $\{a_i\}_{i\in \mathbb Z}$ is a collection of normally distributed independent random variables,
of zero mean and finite variance $\sigma^2_a$, and
$\phi_0,\phi_1,\dots,\phi_P\in\mathbb R$ fulfil that $\phi_P\neq 0$ and
$1-\phi_1u-\phi_2u^2-\cdots-\phi_Pu^P\neq0$, $\forall u\in\mathbb{C}, \ |u|\leq 1$ (stationarity condition). Let be $S_t=\underset{p=1}{\overset{t}{\sum}}|\phi_p|$, $1\leq t\leq P$. The set of past time periods that produce a dependency in the distributions of the target variables at time period $t$ are
\begin{align*}
    \T_t=\{\tilde{t},\ldots,t\},\quad
    \tilde{t}=\tilde{t}(t,P)=1, \text{ if } 1\leq t\leq P; \ t-P, \text{ if } t>P,
\end{align*}
so the vector of inter-period weights is $\ww_t=(w_{t1},\dots, w_{tT})$, where
\begin{equation}\label{wti}
    w_{ti}=\frac{|\phi_{t+1-i}|}{S_t} \,\text{ if }\, \tilde{t}\leq i \leq  t,\, \text{ and  } w_{ti}=0 \text{ if } i>t \text{ or } i<\tilde{t}, \quad t=1,\dots, T.
\end{equation}
The population and samples sizes are $N_{.t}=\sum_{d=1}^DN_{dt}$ and $n_{.t}=\sum_{d=1}^Dn_{dt}$, respectively, and the relevant subsets at time period $t$ and corresponding sizes, are
\begin{align}\label{tsets}
    \nonumber
    s_{d(t)}&=\bigcup_{i\in\T_t} s_{di},\
    U_{d(t)}=\bigcup_{i\in\T_t} U_{di};\ \
    n_{d(t)}=\sum_{i\in\T_t}n_{di},\
    N_{d(t)}=\sum_{i\in\T_t}N_{di},
    \\
    s_{(t)}&=\bigcup_{d=1}^D s_{d(t)},\
    U_{(t)}=\bigcup_{d=1}^D U_{d(t)};\ \
    n_{(t)}=\sum_{d=1}^D n_{d(t)},\
    N_{(t)}=\sum_{d=1}^DN_{d(t)}.
\end{align}
\section{Time-weighted M-quantile models}\label{TWMQ}
%
%
%
For $0<q<1$, $t=1,\dots,T$, the time-weighted MQ linear regression (TWMQ) models are
\begin{equation}\label{TMQ3}
    y_{dij}={\xx}_{dij}^{\prime}\bbeta_\psi(q,\ww_t)+ e_{\psi,dij}(q,\ww_t),\quad d=1,\ldots,D,\, i=1,\ldots,T,\, j=1,\ldots,N_{di},
\end{equation}
where $\bbeta_\psi(q,\ww_t)=(\beta_{\psi1}(q,\ww_t),\ldots,\beta_{\psi p}(q,\ww_t))^\prime$ is the vector of regression coefficients, $\ww_t=(w_{t1},\ldots,w_{tT})^\prime$ is the vector of known non-negative inter-period weights
and
$e_{\psi,dij}(q,\ww_t)$ are independent model errors with unknown c.d.f. $F_{qt}(u)=P(e_{\psi,dij}(q,\ww_t)\leq u|{\xx}_{dij})$, $u\in\mathbb R$. It is satisfied by definition that
$Q_{q}(e_{\psi,dij}(q,\ww_t);\sigma_{qt},\psi|{\xx}_{dij})=0$
and, although no explicit parametric assumptions are being made, the homoscedasticity assumption $\sigma_{qt}=\mbox{var}^{1/2}(e_{\psi,dij}(q,\ww_t))=\sigma_{\psi}(\bbeta_\psi(q,\ww_t))$ is imposed. Our choice is to set the weights $\ww_t$ according to (\ref{wti}), so the time-dependent subsets are \eqref{tsets}.

For $0<q<1$, $t=1,\dots, T$, the regression vector $\bbeta_\psi(q,\ww_t)$ is the solution of
\begin{equation}\label{hatbetaTMquantile}
\sum_{d=1}^D\sum_{i\in\T_t}w_{ti}\sum_{j=1}^{n_{di}}\psi_{q}\big(y_{dij}-{\xx}_{dij}^{\prime}\bbeta_\psi(q,\ww_t),\hat\sigma_{qt}\big) x_{dijk}=0,\, k=1,\ldots,p, 
\end{equation}
where $\psi$ is the Huber function \eqref{psiHuber}
and $\hat\sigma_{qt}=\hat\sigma_{\psi}(\bbeta_\psi(q,\ww_t))=\mbox{mad}_{\psi,n}(q,\ww_t)/0.6745$ is a robust estimator of $\sigma_{qt}$,
for $\bbeta_\psi(q,\ww_t)$ known, and
$\mbox{mad}_{\psi,n}(q,\ww_t)$ is the median absolute deviation (MAD) of the errors $e_{\psi,dij}(q,\ww_t)=y_{dij}-{\xx}_{dij}^{\prime}\bbeta_\psi(q,\ww_t)$. As $\bbeta_\psi(q,\ww_t)$ is unknown, we rely on an iterative procedure to solve the system (\ref{hatbetaTMquantile}) of $p$ non-linear equations.
\subsection{Iterative re-weighted least squares algorithm}\label{sec.IRLS}
The system \eqref{hatbetaTMquantile} is solved using the IRLS algorithm.

We define the weights $w_{\psi dij}(q,\ww_t)=\psi_q(e_{\psi,dij}(q,\ww_t),\hat\sigma_{qt})/e_{\psi,dij}(q,\ww_t)$ and
the $t$-relevant vectors $\yy_{s(t)}=\underset{1\leq d\leq D}{\mbox{col}}\big(\underset{i\in\T_t}{\mbox{col}}(\underset{1\leq j\leq n_{di}}{\mbox{col}}(y_{dij}))\big)$ and matrices $X_{s(t)}=\underset{1\leq d\leq D}{\mbox{col}}\big(\underset{i\in\T_t}{\mbox{col}}(\underset{1\leq j\leq n_{di}}{\mbox{col}}(\xx_{dij}^\prime))\big)$ and $W_{s(t)}(q,\ww_t)=\underset{1\leq d\leq D}{\mbox{diag}}\big(\underset{i\in\T_t}{\mbox{diag}}(\underset{1\leq j\leq n_{di}}{\mbox{diag}}(
w_{ti}w_{\psi dij}(q,\ww_t)))\big)$.
We write the equations in (\ref{hatbetaTMquantile}) as
\begin{equation}\label{hatbetaTMqweight}
    \sum_{d=1}^D\sum_{i\in\T_t}w_{ti}\sum_{j=1}^{n_{di}} w_{\psi dij}(q,\ww_t) \big(y_{dij}-{\xx}_{dij}^{\prime}\bbeta_\psi(q,\ww_t)\big)x_{dijk}=0,\quad k=1,\ldots,p.
\end{equation}
or in the matrix form $X_{s(t)}^\prime W_{s(t)}(q,\ww_t)\yy_{s(t)}-X_{s(t)}^\prime W_{s(t)}(q,\ww_t)X_{s(t)}\bbeta_\psi(q,\ww_t)=0$.

If $X_{s(t)}^\prime W_{s(t)}(q,\ww_t)X_{s(t)}$ is invertible, we can write (\ref{hatbetaTMqweight}) in explicit form, i.e. 
\begin{align}\label{beta.up}
    &\bbeta_\psi(q,\ww_t)=\big(X_{s(t)}^\prime W_{s(t)}(q,\ww_t)X_{s(t)}\big)^{-1}X_{s(t)}^\prime W_{s(t)}(q,\ww_t) \yy_{s(t)}
    \\ \nonumber
    &=\bigg(\sum_{d=1}^D\sum_{i\in\T_t}\sum_{j=1}^{n_{di}}w_{ti}w_{\psi dij}(q,\ww_t)\xx_{dij}\xx_{dij}^\prime\bigg)^{-1}
    \sum_{d=1}^D\sum_{i\in\T_t}\sum_{j=1}^{n_{di}}w_{ti}w_{\psi dij}(q,\ww_t)\xx_{dij}y_{dij}.
\end{align}
This yields to the following IRLS algorithm to calculate $\hat{\bbeta}_\psi(q,\ww_t)$.

\begin{enumerate}
    \item[1.]
    Set the initial values $\hat{\bbeta}_\psi^{(0)}(q,\ww_t)$ using e.g. the weighted least squares estimator
    \begin{equation*}
        \hat\bbeta_\psi^{(0)}(q,\ww_t)=\bigg(\sum_{d=1}^D\sum_{i\in\T_t}w_{ti}\sum_{j=1}^{n_{di}}\xx_{dij}\xx_{dij}^\prime\bigg)^{-1}\sum_{d=1}^D\sum_{i\in\T_t}w_{ti}\sum_{j=1}^{n_{di}}\xx_{dij}y_{dij}.
    \end{equation*}
    \item[2.]
    For each iteration ${l}=1,2,\ldots$, do

    \noindent 2.1.
    Calculate
    $   \hat{e}_{\psi,dij}^{({l}-1)}(q,\ww_t)=y_{dij}-{\xx}_{dij}^{\prime}\hat{\bbeta}_\psi^{({l}-1)}(q,\ww_t)$,
    $\hat\sigma_{q t}^{({l}-1)}=\hat{\sigma}_{\psi}(\hat\bbeta_\psi^{({l}-1)}(q,\ww_t))$,
    $w_{\psi dij}^{({l}-1)}(q,\ww_t)= \psi_q\big(\hat{e}_{\psi,dij}^{({l}-1)}(q,\ww_t),\hat\sigma_{q t}^{({l}-1)}\big)/\hat{e}_{\psi,dij}^{({l}-1)}(q,\ww_t)$ and
    
    $W_{s(t)}^{({l}-1)}\big(q,\ww_t)=\underset{1\leq d\leq D}{\mbox{diag}}\big(\underset{i\in\T_t}{\mbox{diag}}(\underset{1\leq j\leq n_{di}}{\mbox{diag}}(w_{ti}w_{\psi dij}^{({l}-1)}(q,\ww_t)))\big)$.

    \noindent 2.2.
    Update the estimator of $\bbeta_\psi(q,\ww_t)$ according to \eqref{beta.up}. 
    
    \item[3.]
    Repeat Step 2 until convergence.
\end{enumerate}

Here it is important to note the clear differences of the TWMQ models \eqref{TMQ3} compared to the MQ GWR \citep{salvati2012}, where the weights are symmetric, i.e. where something like $w_{ti}=w_{it}$, $i,t=1,\dots, T$ should be imposed. For spatial dependencies, adding new locations does not necessarily imply that those already considered lose relevance. In contrast, in the case of temporal dependencies, the first lags lose relevance as more recent information becomes available.
One of the great advantages of the TWMQ models is that recent data can be easily incorporated into the fitting process. In fact, it manages to attribute only positive weights to the closest temporal data through the sets $\T_t$, $t=1,\dots, T$. The computational cost of fitting the TWMQ models is not scalable if we include more recent observations. That is, the regression coefficients $\hat{\bbeta}_\psi(q,\ww_t)$ vary over time, using only the nearest data to estimate them. 
In light of the above, the TWMQ models define the fitted regression surfaces
\[\hat Q_{q}(y_{dtj};\sigma_{qt},\psi|{\xx}_{dtj})={\xx}_{dtj}^{\prime}\hat\bbeta_\psi(q,\ww_t),\ d=1,\dots,D, \ t=1,\dots,T, \ j=1,\dots, N_{dt},\]
that uses information across quantiles and over time.
\subsection{Robust small area predictors based on TWMQ models}\label{sec.tmq}
For the purpose of using the TWMQ models in SAE, we use $\hat{\bbeta}_\psi(\theta_d,\ww_t)$, where the estimation of $\theta_d$ is obtained by fitting the MQ3 models \eqref{MQ3}. We focus on the prediction of the population means  $\overline{Y}_{dt}$, which have been defined in \eqref{Y.dt}. In the following, we include some mathematical notation and developments to derive the small area predictors. For the quantile $q=\theta_d$, we write $\sigma_{\theta_dt}=\sigma_{\psi}(\bbeta_\psi(\theta_d,\ww_t))$ and introduce a simplified notation for model errors and standardized errors, i.e.
\begin{equation}\label{TMQ3resid1}
    e_{\psi,dtj} =y_{dtj}-{\xx}_{dtj}^{\prime}\bbeta_\psi(\theta_d,\ww_t),\
    u_{\psi,dtj}=\sigma_{\theta_dt}^{-1}e_{\psi,dtj}(\theta_d,\ww_t), \ j=1,\ldots,N_{dt}.
\end{equation}
For $j=1,\ldots,n_{di}$, $i=1,\dots, T$, the model residuals are 
\begin{align*}
    \hat e_{\psi,dij}(q,\ww_t)=y_{dij}-\xx_{dij}'\hat\bbeta_{\psi}(q,\ww_t).
\end{align*}
For $j=1,\ldots,N_{dt}$, we define the pseudo-residuals and standardized pseudo-residuals
\begin{eqnarray}\label{TMQ3psresid}
    \nonumber
    \tilde{e}_{\psi,dtj}(q)&=&\xx_{dtj}^\prime\big(\bbeta_\psi(q,\ww_t)-\hat\bbeta_\psi(\theta_d,\ww_t)\big),\quad
    \tilde{u}_{\psi,dtj}(q)=\sigma_{\theta_dt}^{-1}\,\tilde{e}_{\psi,dtj}(q),
    \\
    \nonumber
    \hat{e}_{\psi,dtj}(q)&=&\xx_{dtj}^\prime\big(\bbeta_\psi(q,\ww_t)-\hat\bbeta_\psi(\hat\theta_d,\ww_t)\big),\quad
    \hat{u}_{\psi,dtj}(q)= \sigma_{\theta_dt}^{-1}\hat{e}_{\psi,dtj}(q),
    \\
    \check{e}_{\psi,dtj}(q)&=&\xx_{dtj}^\prime\big(\bbeta_\psi(q,\ww_t)-\bbeta_\psi(\hat\theta_d,\ww_t)\big),\quad
    \check{u}_{\psi,dtj}(q)= \sigma_{\theta_dt}^{-1}\check{e}_{\psi,dtj}(q).
\end{eqnarray}
For $j=1,\ldots,N_{dt}$, the unit-level MQ coefficients of models \eqref{TMQ3} are
\begin{equation}\label{qdtjwt}
    {q}_{dtj}=\underset{0<q<1}{\mbox{solution}}\,\bigg\{ Q_{q}(y_{dtj};\sigma_{qt},\psi|{\xx}_{dtj})=y_{dtj}\bigg\}, \ Q_{q}(y_{dtj};\sigma_{qt},\psi|{\xx}_{dtj})={\xx}_{dtj}^{\prime}\bbeta_\psi(q,\ww_t),
\end{equation}
For the quantile $q=q_{dtj}$, we just write $\tilde{e}_{\psi,dtj}\triangleq \tilde{e}_{\psi,dtj}(q_{dtj})$, $\tilde{u}_{\psi,dtj}=\sigma_{\theta_dt}^{-1}\,\tilde{e}_{\psi,dtj}(q_{dtj})$,
\begin{align}\label{TMQ3resid}
\hat{e}_{\psi,dtj}\triangleq\hat{e}_{\psi,dtj}(q_{dtj}),\
    \hat{u}_{\psi,dtj}= \sigma_{\theta_dt}^{-1}\hat{e}_{\psi,dtj}(q_{dtj}),\ \check{e}_{\psi,dtj}\triangleq\check{e}_{\psi,dtj}(q_{dtj}), \ \check{u}_{\psi,dtj}=\sigma_{\theta_dt}^{-1}\,\check{e}_{\psi,dtj}(q_{dtj}).
\end{align}
Following the same idea as in Section \ref{sec.mq}, the TMQ predictor of $\overline{Y}_{dt}$ is
\begin{equation}\label{tmq}
    \hat{\overline{Y}}_{dt}^{tmq}=\frac{1}{N_{dt}}
    \bigg\{\sum_{j\in s_{dt}}y_{dtj}+\sum_{j\in r_{dt}}\xx_{dtj}^\prime\hat\bbeta_\psi\big(\hat\theta_d,\ww_t\big)\bigg\}
\end{equation}
and the robust bias-corrected (BTMQ) predictor of \eqref{tmq} is
\begin{eqnarray}\label{btmq}
    \hat{\overline{Y}}_{dt}^{btmq}=\hat{\overline{Y}}_{dt}^{tmq}+\frac{1}{n_{dt}}\Big(1-\frac{n_{dt}}{N_{dt}}\Big)\hat{B}_{dt}^{btmq},\quad \hat{B}_{dt}^{btmq}=
    \sum_{j\in s_{dt}}\sigma_{\theta_d t}
    \phi\big(\hat u_{\psi,dtj}\big)
\end{eqnarray}
where $\phi$ is an influence function with robustness parameter $c_\phi\geq 0$. For the first time in the literature, in Section \ref{sec.c} we propose an optimality criterion for a proper selection of $c_{\phi}$, by minimising the MSE, and prove its existence and uniqueness.

\subsection{MSE estimation for TMQ predictors}\label{MSE.TMQ}
Let $d=1,\dots, D$, $t=1,\dots, T$. Let us assume that $q=\hat\theta_d$ is known. Define the $n(t)\times 1$ vector indicating the units of $s(t)$ that belongs to domain $d$ and time period $t$, i.e. $\varepsilon_{dt}=\underset{1\leq g\leq D}{\mbox{col}}\big(\delta_{g d}\underset{i\in\T_t}{\mbox{col}}(\delta_{it}\underset{1\leq j\leq n_{g i}}{\mbox{col}}(1))\big)
=\underset{1\leq g\leq D}{\mbox{col}}\big(\underset{i\in\T_t}{\mbox{col}}(\underset{1\leq j\leq n_{gi}}{\mbox{col}}(\varepsilon_{dt,gij}))\big)$ where $\delta_{ab}=1$ if and only if $a=b$.
Then we can write the TMQ predictor in the form
\begin{align*}
    \hat{\overline{Y}}_{dt}^{tmq}&=
    \frac{1}{N_{dt}}\,\aa_{dt}^\prime \yy_{s(t)}=\frac{1}{N_{dt}}\sum_{g=1}^D\sum_{i\in\T_t}\sum_{j=1}^{n_{g i}}a_{dt,gij}y_{gij},
\end{align*}
where $\aa_{dt}^\prime=\varepsilon_{dt}^\prime+\zz_{dt}^\prime=\underset{1\leq g\leq D}{\mbox{col}}\big(\underset{i\in\T_t}{\mbox{col}}(\underset{1\leq j\leq n_{gi}}{\mbox{col}}(a_{dt,gij}))\big)$, $a_{dt,gij}=\varepsilon_{dt,gij}+z_{dt,gij}$ and $\zz_{dt}^\prime=\underset{1\leq g\leq D}{\mbox{col}}\big(\underset{i\in\T_t}{\mbox{col}}(\underset{1\leq j\leq n_{gi}}{\mbox{col}}(z_{dt,g ij}))\big)$,
$
\zz_{dt}^\prime
=\bigg(\sum_{j\in r_{dt}}\xx_{dtj}^\prime\bigg)\big(X_{s(t)}^\prime W_{s(t)}(\hat\theta_d,\ww_t)X_{s(t)}\big)^{-1}X_{s(t)}^\prime W_{s(t)}(\hat\theta_d,\ww_t)
$.
Let us define the $N(t)\times 1$ vector indicating the units of $U(t)$ that belongs to domain $d$ and time period $t$ and the vectors $1_{dt}=\underset{1\leq j\leq N_{dt}}{\mbox{col}}(1)$, $\yy_{dt}=\underset{1\leq j\leq N_{dt}}{\mbox{col}}(y_{dtj})$, $1_{r_{dt}}=\underset{j\in r_{dt}}{\mbox{col}}(1)$, $\yy_{r_{dt}}=\underset{j\in r_{dt}}{\mbox{col}}(y_{dtj})$. The prediction error and MSE of  $\hat{\overline{Y}}_{dt}^{tmq}$
are, respectively,
\begin{eqnarray*}
    \hat{\overline{Y}}_{dt}^{tmq}-\overline{Y}_{dt}
    &=&N_{dt}^{-1}\big(\aa_{dt}^\prime \yy_{s(t)}-1_{dt}^\prime y_{dt}\big)
    =N_{dt}^{-1}\big(\zz_{dt}^\prime \yy_{s(t)}-1_{r_{dt}}^\prime \yy_{r_{dt}}\big),\\
    MSE(\hat{\overline{Y}}_{dt}^{tmq})&=&E\Big[\big(\hat{\overline{Y}}_{dt}^{tmq}-\overline{Y}_{dt}\big)^2\Big]
    =
    \text{var}\big(\hat{\overline{Y}}_{dt}^{tmq}-\overline{Y}_{dt}\big)
    +B^2\big(\hat{\overline{Y}}_{dt}^{tmq}\big),
\end{eqnarray*}
where
$\text{var}\big(\hat{\overline{Y}}_{dt}^{tmq}-\overline{Y}_{dt}\big)=\frac{1}{N_{dt}^2}\text{var}\big( \zz_{dt}^\prime \yy_{s(t)}-1_{r_{dt}}^\prime \yy_{r_{dt}}\big)$ and
$B\big(\hat{\overline{Y}}_{dt}^{tmq}\big)=\frac{1}{N_{dt}}E\big[\zz_{dt}^\prime \yy_{s(t)}-1_{r_{dt}}^\prime \yy_{r_{dt}}\big].
$

Following \cite{Cha06}, we derive a first-order approximation \citep{Chambers2011} of $MSE(\hat{\overline{Y}}_{dt}^{tmq})$ and propose several estimators. In general, the procedure described below could be applied when the quantity to be predicted can be expressed as a linear combination of the values taken by the target variable in all units of the sample. It is important to note that these  pseudo-linear MSE estimators assume that the weights $\aa_{dt}$ are fixed quantities, and thus ignore their contribution to the MSE, derived from the estimation of the $\theta_d$ coefficients using the MQ3 models. In practice, the latter should not be a major problem, as this variability is expected to be rather small \citep{schirripa2021}. Firstly, the variance can be approximated as
\begin{align}\label{vardif1}
    \text{var}\big(\hat{\overline{Y}}_{dt}^{tmq}-&\overline{Y}_{dt}\big)\approx\frac{1}{N_{dt}^2} \bigg(\sum_{g=1}^D\sum_{i\in\T_t}\sum_{j=1}^{n_{g i}}z_{dt,gij}^2\text{var}(y_{gij})
    +\sum_{j \in r_{dt}}\text{var}(y_{dtj})\bigg)
    \\
    &=
    \frac{1}{N_{dt}^2} \bigg(\sum_{g=1}^D\sum_{i\in\T_t}\sum_{j=1}^{n_{g i}}z_{dt,gij}^2\text{var}(y_{gij})
    +\sum_{g=1}^D\sum_{i\in\T_t}\sum_{j=1}^{N_{gi}}\varepsilon_{rdt,g ij}\text{var}(y_{gij})\bigg),
    \nonumber
\end{align}
where
$
\varepsilon_{rdt}=\underset{1\leq g\leq D}{\mbox{col}}\big(\delta_{g d}\underset{i\in\T_t}{\mbox{col}}(\delta_{it}
\mbox{col}(\underset{j\in s_{gi}}{\mbox{col}}(0), \underset{j\in r_{g i}}{\mbox{col}}(1)))\big)
=\underset{1\leq g\leq D}{\mbox{col}}\big(\underset{i\in\T_t}{\mbox{col}}(\underset{1\leq j\leq N_{gi}}{\mbox{col}}(\varepsilon_{rdt,gij}))\big).
$

To estimate $\text{var}(y_{gij})$ and $\text{var}(y_{dtj})$, we consider two approaches.
(1) {\it Median approach}:
the variance estimators are based on the median model.
Accordingly, we use the median estimator
$\widehat{\text{var}}(y_{gij})=(y_{gij}-\xx_{g ij}^\prime\hat\bbeta_\psi(0.5,\ww_t))^2=\hat e^2_{\psi,gij}(0.5,\ww_t)$ for the variance of the sample observations in $s(t)$ and estimate
the second summand in (\ref{vardif1}) using
$$
\sum_{j \in r_{dt}}\widehat{\text{var}}(y_{dtj})
=\sum_{g=1}^D\sum_{i\in\T_t}\sum_{j=1}^{n_{gi}}\varepsilon_{dt,g ij}\frac{N_{dt}-n_{dt}}{n_{d(t)}-1}
\ \hat e^2_{\psi,gij}(0.5,\ww_t).
$$
Therefore, the median estimator of $\text{var}\big(\hat{\overline{Y}}_{dt}^{tmq}-\overline{Y}_{dt}\big)$ is
$$
\hat{V}_{11,dt}^{tmq}=\frac{1}{N_{dt}^2}\sum_{g=1}^D\sum_{i\in\T_t}\sum_{j=1}^{n_{gi}}\lambda_{dt,gij}\ \hat e^2_{\psi,gij}(0.5,\ww_t), \ \lambda_{dt,gij}=z_{dt,g ij}^2+\frac{N_{dt}-n_{dt}}{n_{d(t)}-1}\,\varepsilon_{dt,gij}.
$$

\noindent
(2) {\it Area quantile coefficient approach}:
the variance is estimated according to the representative quantile of the area $g$ from which the observation is drawn.
We use the area quantile coefficient estimator
$\widehat{\text{var}}(y_{gij})=(y_{g ij}-\xx_{g ij}^\prime\hat\bbeta_\psi(\hat\theta_g,\ww_t))^2=\hat e^2_{\psi,g ij}(\hat\theta_g,\ww_t)$ for the variance of the sample observations in $s(t)$ and estimate
the second summand in (\ref{vardif1}) using
$$
\sum_{j \in r_{dt}}\widehat{\text{var}}(y_{dtj})
=\sum_{g=1}^D\sum_{i\in\T_t}\sum_{j=1}^{n_{gi}}\varepsilon_{dt,g ij}\frac{N_{dt}-n_{dt}}{n_{d(t)}-1}\ \hat e^2_{\psi,g ij}(\hat\theta_g,\ww_t).
$$
Therefore, the area quantile coefficient estimator of $\text{var}\big(\hat{\overline{Y}}_{dt}^{tmq}-\overline{Y}_{dt}\big)$ is
$$
\hat{V}_{12,dt}^{tmq}=\frac{1}{N_{dt}^2}\sum_{g=1}^D\sum_{i\in\T_t}\sum_{j=1}^{n_{gi}}\lambda_{dt,gij}\ \hat e^2_{\psi,gij}(\hat\theta_g,\ww_t).
$$
By using a different formula for the prediction error of $\hat{\overline{Y}}_{dt}^{tmq}$, we  give two alternative estimators of the variance. For this sake, we define the scalars and vectors
$$
\bar{y}_{sdt}=\frac{1}{n_{dt}}\sum_{j\in s_{dt}}y_{dtj},\,\,\,
\bar{y}_{rdt}=\frac{1}{N_{dt}-n_{dt}}\sum_{j\in r_{dt}}y_{dtj},\,\,\,
\bar{\xx}_{rdt}^\prime=\frac{1}{N_{dt}-n_{dt}}\sum_{j\in r_{dt}}\xx_{dtj}^\prime.
$$
From equation (\ref{tmq}), we can write the TMQ predictor and the population mean as
\begin{align*}
  & \hat{\overline{Y}}_{dt}^{tmq}=\frac{1}{N_{dt}}
\Big\{n_{dt}\bar{y}_{sdt}+(N_{dt}-n_{dt})\bar{\xx}_{rdt}^\prime\hat\bbeta_\psi\big(\hat\theta_d,\ww_t\big)\Big\},\ 
\overline{Y}_{dt}=\frac{1}{N_{dt}} \Big\{n_{dt}\bar{y}_{sdt}+(N_{dt}-n_{dt})\bar{y}_{rdt}\Big\},
\end{align*}
so the prediction error is $\hat{\overline{Y}}_{dt}^{tmq}-\overline{Y}_{dt}
=\Big(1-\frac{n_{dt}}{N_{dt}}\Big)\Big\{\bar{\xx}_{rdt}^\prime\hat\bbeta_\psi\big(\hat\theta_d,\ww_t\big)-\bar{y}_{rdt}\Big\}$.

An estimator of the prediction error variance is
\begin{equation}\label{vardif2}
    \widehat{\text{var}}\big(\hat{\overline{Y}}_{dt}^{tmq}-\overline{Y}_{dt}\big)=\Big(1-\frac{n_{dt}}{N_{dt}}\Big)^2\bar{\xx}_{rdt}^\prime\hat{V}_\beta\bar{\xx}_{rdt}
    +\Big(1-\frac{n_{dt}}{N_{dt}}\Big)^2\widehat{\text{var}}(\bar{y}_{rdt}), 
\end{equation}
where $\hat{V}_\beta=\widehat{\text{var}}(\hat\bbeta_\psi(\hat\theta_d,\ww_t))$.
Based on the sandwich approach 
to estimate the asymptotic variance of the vector of coefficients in MQ linear models \citep{Bianchi2015},
we derive an estimator of ${V}_\beta$.
Under assumptions (A1)-(A8), listed in Section A.1 of the Supplementary Material, an estimator of ${V}_\beta$ is
\begin{align}\label{var.beta}
    \hat{V}_\beta=\frac{n_{(t)}^2\sigma_{\theta_d t}^2}{n_{(t)}-p}\frac{\overset{D}{\underset{g=1}{\sum}}
        \ \underset{i\in\T_t}{\sum}\ \underset{j\in s_{gi}}{\sum}
        \psi_{\hat\theta_d t}^2\big(\hat  e_{\psi,gij}(\hat\theta_d,\ww_t), \sigma_{\theta_d t}\big)}
    {\left(\overset{D}{\underset{g=1}{\sum}}
        \ \underset{i\in\T_t}{\sum}\ \underset{j\in s_{gi}}{\sum}
        \dot{\psi}_{\hat\theta_d t}\big(\hat e_{\psi,g ij}(\hat\theta_d,\ww_t),\sigma_{\theta_d t}\big)\right) ^2}
    \left(\sum_{g=1}^D\sum_{i\in\T_t}\sum_{j\in s_{gi}}\xx_{gij}^\prime \xx_{g ij}\right)^{-1},
\end{align}
where $\dot{\psi}_{\hat\theta_d,t}$ is the partial derivative of ${\psi}_{\hat\theta_d,t}$ with respect to the first argument.

In order to estimate  $\text{var}\big(\bar{y}_{rdt}\big)$ in \eqref{vardif2}, we can use
\begin{align}\label{var.opt}
    \widehat{\text{var}}_1\big(\bar{y}_{rdt}\big)=\frac{\sum_{j\in s_{dt}}\hat e^2_{\psi,dtj}}{(N_{dt}-n_{dt})(n_{d(t)}-1)}  \text{ or }
    \widehat{\text{var}}_2\big(\bar{y}_{rdt}\big)=\frac{\sum_{g=1}^D\sum_{j \in s_{gt}}\hat e^2_{\psi,gtj}}{(N_{dt}-n_{dt})(n-D)}.
\end{align}
By substituting $\hat{V}_\beta$ and $\widehat{\text{var}}_1\big(\bar{y}_{rdt}\big)$ or $\widehat{\text{var}}_2\big(\bar{y}_{rdt}\big)$ in (\ref{vardif2}),
we obtain the estimators $\hat{V}_{21,dt}^{tmq}$ and $\hat{V}_{22,dt}^{tmq}$, respectively. Secondly, the bias $B\big(\hat{\overline{Y}}_{dt}^{tmq}\big)$ can be estimated by
\begin{align}\label{ebiastmq}
    \hat B_{dt}=\frac{1}{N_{dt}}\bigg(\sum_{g=1}^D\sum_{i\in\T_t}\sum_{j \in s_{gi}}a_{dt,gij}\xx_{gij}^\prime\hat\bbeta_{\psi}(\hat\theta_g,\ww_t)
    -\sum_{j \in U_{dt}} \xx_{dtj}^\prime\hat\bbeta_{\psi}(\hat\theta_d,\ww_t)\bigg),
\end{align}
where the terms $a_{dt,gij}$'s are the components of $\aa_{dt}$ appearing in the definition of the prediction error of $\hat{\overline{Y}}_{dt}^{tmq}$.
All in all, we have four estimators of $MSE(\hat{\overline{Y}}_{dt}^{tmq})$. They are
$\ mse_{11,dt}^{tmq}=\hat V_{11,dt}^{tmq}+\hat B_{dt}^2$,
$\ mse_{12,dt}^{tmq}=\hat V_{12,dt}^{tmq}+\hat B_{dt}^2$,
$\ mse_{21,dt}^{tmq}=\hat V_{21,dt}^{tmq}+\hat B_{dt}^2\ $
and $\ mse_{22,dt}^{tmq}=\hat V_{22,dt}^{tmq}+\hat B_{dt}^2$.
By taking squared roots of the above estimators, we define $rmse_{11,dt}^{tmq}$, $rmse_{12,dt}^{tmq}$, $rmse_{21,dt}^{tmq}$ and $rmse_{22,dt}^{tmq}$, respectively.
They are estimators of $RMSE(\hat{\overline{Y}}_{dt}^{tmq})=\big(MSE(\hat{\overline{Y}}_{dt}^{tmq})\big)^{1/2}$.

{\subsection{ MSE estimation for BTMQ predictors}\label{MSE.BTMQ}

As a follow-up to Section \ref{MSE.TMQ}, the BTMQ predictor of $\overline{Y}_{dt}$ can be written as
\begin{eqnarray*}
    \hat{\overline{Y}}_{dt}^{btmq}
    &=&\frac{1}{N_{dt}}
    \bigg\{\sum_{j\in s_{dt}}y_{dtj}+\sum_{j\in r_{dt}}\xx_{dtj}^\prime\hat\bbeta_\psi\big(\hat\theta_d,\ww_t\big)\bigg\}
    +\frac{1}{n_{dt}}\Big(1-\frac{n_{dt}}{N_{dt}}\Big)\hat
    B_{dt}^{btmq},
\end{eqnarray*}
where $\hat
B_{dt}^{btmq}$ has been defined in \eqref{btmq}. 
We now define
\begin{eqnarray*}
    \overline{Y}_{dt}^{btmq}&=&\frac{1}{N_{dt}}
    \bigg\{\sum_{j\in s_{dt}}y_{dtj}+\sum_{j\in r_{dt}}\xx_{dtj}^\prime\bbeta_\psi(\theta_d,\ww_t)\bigg\}
    +\frac{1}{n_{dt}}\Big(1-\frac{n_{dt}}{N_{dt}}\Big){B}_{dt}^{btmq},\,
    \\
    \tilde{\overline{Y}}_{dt}^{btmq}&=&\frac{1}{N_{dt}}
    \bigg\{\sum_{j\in s_{dt}}y_{dtj}+\sum_{j\in r_{dt}}\xx_{dtj}^\prime\hat\bbeta_\psi(\theta_d,\ww_t)\bigg\}
    +\frac{1}{n_{dt}}\Big(1-\frac{n_{dt}}{N_{dt}}\Big)\tilde{B}_{dt}^{btmq},
    \\
    {B}_{dt}^{btmq}&=&\sum_{j\in s_{dt}}\sigma_{\theta_dt}\phi(u_{\psi,dtj}),\quad
    \tilde{B}_{dt}^{btmq}=\sum_{j\in s_{dt}}\sigma_{\theta_dt}\phi(\tilde{u}_{\psi,dtj}),
\end{eqnarray*}
where $u_{\psi,dtj}$ and $\tilde{u}_{\psi,dtj}$ have been defined in
\eqref{TMQ3resid1} and \eqref{TMQ3resid}, respectively.
It holds that:
\[MSE\big(\hat{\overline{Y}}_{dt}^{btmq}\big)=
E\big[\big(\hat{\overline{Y}}_{dt}^{btmq}-\overline{Y}_{dt}\big)^2\big]=\mbox{var}\big(\hat{\overline{Y}}_{dt}^{btmq}-\overline{Y}_{dt}\big)+(E\big[\hat{\overline{Y}}_{dt}^{btmq}-\overline{Y}_{dt}\big])^2.\]

Under the assumptions listed in Section A.1,  the prediction error of $\hat{\overline{Y}}_{dt}^{btmq}$ is:
\begin{equation}\label{pred.errbtmq}
   \hat{\overline{Y}}_{dt}^{btmq}-\overline{Y}_{dt}=
\big(\hat{\overline{Y}}_{dt}^{btmq}-\tilde{\overline{Y}}_{dt}^{btmq}\big)
+\big(\tilde{\overline{Y}}_{dt}^{btmq}-\overline{Y}_{dt}^{btmq}\big)
+\big(\overline{Y}_{dt}^{btmq}-\overline{Y}_{dt}\big)=\overline{Y}_{dt}^{(3)}+\overline{Y}_{dt}^{(2)}+\overline{Y}_{dt}^{(1)}, 
\end{equation}

To simplify the notation we define the variance and the expected prediction difference of $\overline{Y}_{dt}^{(k)}$, for $k=1,2,3$, as $V_{dt}^{(k)}$, and $E_{dt}^{(k)}$, respectively. Then we write
\begin{align*}
&\mbox{var}\big(\hat{\overline{Y}}_{dt}^{btmq}-\overline{Y}_{dt}\big)=
    V_{dt}^{(1)}+V_{dt}^{(2)}+V_{dt}^{(3)}+2\mbox{cov}\big(\overline{Y}_{dt}^{(3)},\overline{Y}_{dt}^{(2)}\big)
    +2\mbox{cov}\big(\overline{Y}_{dt}^{(3)},\overline{Y}_{dt}^{(1)}\big)+2\mbox{cov}\big(\overline{Y}_{dt}^{(2)},\overline{Y}_{dt}^{(1)}\big),
    \\&E\big[\hat{\overline{Y}}_{dt}^{btmq}-\overline{Y}_{dt}\big]=
    E_{dt}^{(1)}+E_{dt}^{(2)}+E_{dt}^{(3)}=E_{dt}^{(1)}+o(1).
\end{align*}
The covariances are $\mbox{cov}\big(\overline{Y}_{dt}^{(3)},\overline{Y}_{dt}^{(2)}\big)=E\big[\overline{Y}_{dt}^{(3)}\overline{Y}_{dt}^{(2)}\big]+o(1)$, $\mbox{cov}\big(\overline{Y}_{dt}^{(3)},\overline{Y}_{dt}^{(1)}\big)=
E\big[\overline{Y}_{dt}^{(3)}\overline{Y}_{dt}^{(1)}]+o(1)$ and
$\mbox{cov}\big(\overline{Y}_{dt}^{(2)},\overline{Y}_{dt}^{(1)}\big)=
E\big[\overline{Y}_{dt}^{(2)}\overline{Y}_{dt}^{(1)}]+o(1)$
Under regularity assumptions, the expectations of the previous cross-products should be $o(1)$. We also define the set  $\check{\G}_{dt}=\Big\{j\in s_{dt}: \big|\check{u}_{\psi,dtj}\big| <c_{\phi}\Big\}$.
The following theorem summarizes the final approximation of $MSE\big(\hat{\overline{Y}}_{dt}^{btmq}\big)$.
\begin{theorem}\label{first.order.MSE}
    \normalfont Under assumptions ($\Phi$1), (N1)-(N2), (Q1), (A1)-(A9), (B1)-(B4), (C1)-(C2), (D1)-(D2) and (E1)-(E3), listed in
    Section A.1 of Supplementary Material, a first-order approximation of $MSE(\hat{\overline{Y}}_{dt}^{btmq})$ is
    \begin{align*}
        \nonumber
        MS&E\big(\hat{\overline{Y}}_{dt}^{btmq}\big)= V_{dt}^{(1)}+V_{dt}^{(2)}+V_{dt}^{(3)}+E_{dt}^{(1)2}+o(1)=\\&=\sum_{j\in U_{dt}}\Big(\Big(1-\frac{n_{dt}}{N_{dt}}\Big)^2\frac{1}{n_{dt}^2}I_{\G_{dt}}(j)+ \frac{1}{N_{dt}^2}I_{r_{dt}}(j)\Big)
        \Big(\xx_{dtj}^\prime\kappa_\psi(\theta_{d},\ww_t)\Big)^2\xi_{dt}^2
        \\
        \nonumber
        &+\sum_{j\in U_{dt}}\Big(\Big(1-\frac{n_{dt}}{N_{dt}}\Big)^2\frac{1}{n_{dt}^2}I_{\tilde{\mathcal H}_{dt}}(j)+\frac{1}{N_{dt}^2}I_{r_{dt}}(j)\Big)
        \xx_{dtj}^\prime\text{var}\big(\hat\bbeta_{\psi}(\theta_d,\ww_t)\big)\xx_{dtj}
        \\
        \nonumber
        &+\sum_{j\in U_{dt}}    \Big(\Big(1-\frac{n_{dt}}{N_{dt}}\Big)^2\frac{1}{n_{dt}^2}I_{\hat{\mathcal H}_{dt}}(j)+\frac{1}{N_{dt}^2}I_{r_{dt}}(j)\Big)
        \xx_{dtj}^\prime\text{var}\big(\hat\bbeta_{\psi}(\hat\theta_d,\ww_t)\big)\xx_{dtj}
        \\
        \nonumber
        &+\Big(1-\frac{n_{dt}}{N_{dt}}\Big)^2\bigg( \frac{c_\phi}{n_{dt}}\sum_{j\in s_{dt}-\G_{dt}}E\big[\text{sgn}(e_{\psi,dtj})\big]
        +\frac{\sigma_{\theta_dt}}{n_{dt}}\sum_{j\in\G_{dt}}E[R_{dtj}]\bigg) ^2+o(1).
    \end{align*}
\end{theorem}
\begin{proof}
    The proof is given in Section A of the Supplementary Material.
\end{proof}
An analytical estimator of $MSE\big(\hat{\overline{Y}}_{dt}^{btmq}\big)$ is given by
\begin{align}\label{mse5btmq}
    \nonumber
    mse&_{3,dt}^{btmq}=
    \Big(1-\frac{n_{dt}}{N_{dt}}\Big)^2\frac{\hat\xi_{dt}^2}{n_{dt}^2}\Big( \frac{1}{\text{card}(\hat\G_{dt})}\sum_{j\in \hat\G_{dt}}{(\hat q_{dtj}-\hat\theta_d)^{2}}\Big) ^{-1}\sum_{j\in \hat\G_{dt}}\hat e_{\psi,dtj}^2+
    \\
    \nonumber
    &+\frac{N_{dt}-n_{dt}}{n_{dt}}\frac{\hat\xi_{dt}^2}{N_{dt}^2}\Big( \frac{1}{n_{dt}}\sum_{j\in s_{dt}}{(\hat q_{dtj}-\hat\theta_d)^{2}}\Big) ^{-1}\sum_{j\in s_{dt}}\hat e_{\psi,dtj}^2
    \\
    \nonumber
    &+2\Big(1-\frac{n_{dt}}{N_{dt}}\Big)^2\frac{1}{n_{dt}^2}\sum_{j\in \hat{\mathcal G}_{dt}}\xx_{dtj}^\prime \hat V_{\beta}\xx_{dtj}
    +\frac{1}{N_{dt}^2}\sum_{j\in r_{dt}}\xx_{dtj}^\prime \hat V_{\beta}\xx_{dtj}
    \\
    &+
    \Big(1-\frac{n_{dt}}{N_{dt}}\Big)^2\frac{1}{n_{dt}^2}\bigg( c_\phi\sum_{j\in s_{dt}-\hat\G_{dt}}\text{sgn}(\hat e_{\psi,dtj})
    +\frac{1}{2\sigma_{\theta_dt}}\sum_{j\in\hat\G_{dt}}\hat e_{\psi,dtj}^2\bigg) ^2,
\end{align}
where $\hat V_\beta$ is given in \eqref{var.beta},
$\hat{\G}_{dt}=\Big\{j\in s_{dt}: \big|\hat{u}_{\psi,dtj}\big| <c_{\phi}\Big\}$, expression $\text{card}(B)$ denotes the cardinal of a set $B$ and $\hat\xi_{dt}^2$ estimates $\text{var}(q_{dtj})$, e.g. by
$$
\hat\xi_{dt}^2=\widehat{\text{var}}(q_{dtj})=\frac{1}{n_{dt}-1}\sum_{j\in s_{dt}}(\hat q_{dtj}-\hat{\bar{q}}_{dt.})^2,\quad
\hat{\bar{q}}_{dt.}=\frac{1}{n_{dt}}\sum_{j\in s_{dt}}\hat q_{dtj}.
$$
We have denoted $mse_{3,dt}^{btmq}$ because estimators $mse_{1,dt}^{btmq}$ and $mse_{2,dt}^{btmq}$ will be proposed below, in parallel to Section \ref{MSE.TMQ}.
Having said that, an estimator of $RMSE\big(\hat{\overline{Y}}_{dt}^{btmq}\big)=\big(MSE\big(\hat{\overline{Y}}_{dt}^{btmq}\big)\big)^{1/2}$ is
$rmse_{3,dt}^{btmq}=\big(mse_{3,dt}^{btmq}\big)^{1/2}$.

On the other hand, two alternative estimators of $MSE\big(\hat{\overline{Y}}_{dt}^{btmq}\big)$ could be also derived. From (\ref{TMQ3psresid}) we have
$\hat u_{\psi, dtj}-\check u_{\psi, dtj}=\sigma_{\theta_d t}^{-1}\xx_{dtj}'(\bbeta_{\psi}(\hat\theta_d,\ww_t)-\hat\bbeta_{\psi}(\hat\theta_d,\ww_t))=\hat u_{\psi, dtj}(\hat\theta_d)$.
A Taylor series expansion of $\phi(\hat u_{\psi, dtj})$ around $\phi(\check u_{\psi, dtj})$ yields to
$$\phi(\hat u_{\psi, dtj})\approx\phi(\check u_{\psi, dtj})+\dot\phi(\check u_{\psi, dtj})\frac{\xx_{dtj}'(\bbeta_{\psi}(\hat\theta_d,\ww_t)-\hat\bbeta_{\psi}(\hat\theta_d,\ww_t))}{\sigma_{\theta_d t}}, \ j=1\dots, n_{dt}.$$
Let us define:
$$
\bar{\xx}_{\check{\G}_{dt}}'=\frac{1}{\text{card}(\check{\G}_{dt})}\sum_{j\in \check{\G}_{dt}}\xx_{dtj}',\quad
\bar{\xx}_{sdt}'=\frac{1}{n_{dt}}\sum_{j\in s_{dt}}x_{dtj}', \quad \bar{e}_{rdt}=\bar{y}_{rdt}-\bar{\xx}_{rdt}^\prime\bbeta_\psi(\hat\theta_d,\ww_t).
$$
If $\phi$ is the Huber function \eqref{psiHuber}, $\dot\phi(\check u_{\psi, dtj})= 1$ if $j\in \check{\G}_{dt}$ and $\dot\phi(\check u_{\psi, dtj})= 0$, otherwise. Following \cite{Chambers14}, it holds that
\begin{align*}
    \frac{\sigma_{\theta_d t}}{n_{dt}} \sum_{j\in s_{dt}}\phi&(\hat u_{\psi, dtj})
    \approx\frac{\sigma_{\theta_d t}}{n_{dt}}\sum_{j\in s_{dt}} \phi(\check u_{\psi, dtj})+\bar{\xx}_{sdt}'(\bbeta_{\psi}(\hat\theta_d,\ww_t)-\hat\bbeta_{\psi}(\hat\theta_d,\ww_t)).
\end{align*}

Therefore, the prediction error of $\hat{\overline{Y}}_{dt}^{btmq}$ can be approximated as
\begin{align*}
    &\hat{\overline{Y}}_{dt}^{btmq}-\overline{Y}_{dt}=\frac{1}{N_{dt}}\bigg(\sum_{j\in r_{dt}}\xx'_{dtj}\hat\bbeta_\psi(\hat\theta_d,\ww_t)- \sum_{j\in r_{dt}} y_{dtj}\bigg)+\frac{1}{n_{dt}}\Big(1-\frac{n_{dt}}{N_{dt}}\Big)\hat{B}_{dt}^{btmq}
    \\
    &\approx\Big(1-\frac{n_{dt}}{N_{dt}}\Big)\bigg( \frac{\sigma_{\theta_d t}}{n_{dt}} \sum_{j\in s_{dt}}\phi\big(\check u_{\psi, dtj}\big)+
    \bar{\xx}_{sdt}^\prime\bbeta_{\psi}(\hat\theta_d,\ww_t)+(\bar{\xx}_{rdt}-\bar{\xx}_{sdt})^\prime\hat\bbeta_{\psi}(\hat\theta_d,\ww_t)-\bar{y}_{rdt}\bigg)
    \\
    &=\Big(1-\frac{n_{dt}}{N_{dt}}\Big)\bigg(
    \frac{\sigma_{\theta_d t}}{n_{dt}} \sum_{j\in s_{dt}}\phi\big(\check u_{\psi, dtj}\big)
    +(\bar{\xx}_{rdt}-\bar{\xx}_{sdt})^\prime(\hat\bbeta_{\psi}(\hat\theta_d,\ww_t)-\bbeta_{\psi}(\hat\theta_d,\ww_t))-\bar{e}_{rdt}\bigg).
\end{align*}

An estimator of the variance $V_{dt}^{btmq}=\text{var}\big(\hat{\overline{Y}}_{dt}^{btmq}-\overline{Y}_{dt}\big)$  is
\begin{align}\label{hatVdt}
    \nonumber
    \hat{V}_{dt}^{btmq} &= \Big( 1-\frac{n_{dt}}{N_{dt}}\Big)^2\bigg[
    \Big(\frac{\sigma_{\theta_d t}}{n_{dt}}\Big) ^2 \underset{j\in s_{dt}}\sum\phi^2\big(\hat{u}_{\psi,dtj}\big)
    +(\bar{\xx}_{rdt}-\bar{\xx}_{sdt})^\prime
    \hat{V}_\beta(\bar{\xx}_{rdt}-\bar{\xx}_{sdt})+\widehat{\text{var}}(\bar{e}_{rdt})\bigg],
\end{align}
where the estimation of the variance matrix $\hat V_\beta$ is given in \eqref{var.beta} and $\text{var}(\bar{e}_{r_{dt}})$ can be estimated using $\widehat{\text{var}}_1\big(\bar{y}_{rdt}\big)$ or $\widehat{\text{var}}_2\big(\bar{y}_{rdt}\big)$, given in \eqref{var.opt}. Depending on which formula we choose, we obtain the estimators $\hat V_{21,dt}^{btmq}$ and $\hat V_{22,dt}^{btmq}$
of the variance $V_{dt}^{btmq}$, respectively.
Against this background, note that the bias correction in \eqref{btmq} is controlled by $c_\phi$, which can be either very large or close to zero. Consequently, it is not advisable to ignore the presence of a potential bias, as it may (and perhaps should) still persist. Using the bias estimator $\hat{B}_{dt}=\hat{B}\big(\hat{\overline{Y}}_{dt}^{tmq}\big)$, given in \eqref{ebiastmq}, two estimators of $MSE(\hat{\overline{Y}}_{dt}^{btmq})$ are 
\begin{align} 
    mse_{\ell,dt}^{btmq}&=\hat V_{2\ell,dt}^{btmq}+\Big(\hat{B}_{dt}+\frac{1}{n_{dt}}\Big(1-\frac{n_{dt}}{N_{dt}}\Big)\hat{B}_{dt}^{btmq}\Big) ^2, \quad \ell=1,2.
\end{align}

By taking squared roots of the above estimators, we have
$rmse_{1,dt}^{btmq}$ and $rmse_{2,dt}^{btmq}$, respectively.
They are estimators of $RMSE(\hat{\overline{Y}}_{dt}^{btmq})=\big(MSE(\hat{\overline{Y}}_{dt}^{btmq})\big)^{1/2}$.
}

\subsection{Selection of the robustness parameter}\label{sec.c}
The robustness parameter plays the main role in the improvements of the BTMQ predictor over the TMQ predictor.
To the best of our knowledge, its selection has been done routinely, in many cases by setting $c_{\phi}=3$ \citep{Tzavidis10,salvati2012,Chambers14} as a reference value,
but data-driven optimization has never been investigated. Although the results using $c_{\phi}=3$ are quite encouraging for bias correction \citep{Chambers14},
it seems logical that avoiding the subjective selection of $c_{\phi}\geq 0$ would improve the performance of the BTMQ predictors.

To have a balance between bias and variance, we are looking for the values of $c_\phi\geq 0$ that minimize $MSE(\hat{\overline{Y}}_{dt}^{btmq})$.
In practice, this involves solving the minimization problems
\begin{equation}\label{c.dt}
    \hat{c}_{\phi, dt}=\text{argmin}_{c_\phi\geq 0}\, mse_{dt}^{btmq}(c_\phi), \quad d=1,\dots, D, \ t=1,\dots, T,
\end{equation}
where $mse_{dt}^{btmq}$ is an estimation of $MSE(\hat{\overline{Y}}_{dt}^{btmq})$.

The following theorem proves the existence and uniqueness of solutions of the minimization problems (\ref{c.dt}) for $mse_{dt}^{btmq}\in\{mse_{1,dt}^{btmq}, mse_{2,dt}^{btmq}\}$.

\begin{theorem}\label{exis.uni} \normalfont
    Let $\phi$ be the Huber function \eqref{psiHuber}. Let $d=1,\dots,D$, $t=1,\dots, T$. For $mse_{dt}^{btmq}\in\{mse_{1,dt}^{btmq}, mse_{2,dt}^{btmq}\}$,
    it exists an unique solution $\hat c_{\phi,dt}$ of the minimization problem \eqref{c.dt} belonging to the interval $\big[0,\underset{j\in s_{dt}}{\max}|\hat{u}_{\psi,dtj}| \big]$
    and its explicit expression is calculable.
\end{theorem}
\begin{proof}
    The proof is reported in Section B of the Supplementary Material.
\end{proof}
Considering that the solutions are adaptive robustness parameters, calculable e.g. by using grid search methods, we call them area-time specific robustness parameters.
\section{Model-based simulations}\label{sec.sim}
A Monte Carlo simulation study was conducted to evaluate the performance of the TMQ and BTMQ predictors in estimating small area means and to compare them with other predictors proposed in the literature. Additionally, we show how the selection of area-time specific robustness parameters can be employed as a diagnostic tool for outlier detection.
Finally, we evaluate the performance of the MSE estimators proposed in Sections \ref{MSE.TMQ} and \ref{MSE.BTMQ} using boxplots. 


The simulation design is inspired by \cite{Chambers14},  but here we have included a time reference to all observations. 
Population data are generated for $D=40$ areas and $T=10$ time periods. From each population, the sample data have been selected by simple random sampling without replacement within each intersection between areas and time periods. Population and sample sizes are fixed at $N_{dt} =100$ and $n_{dt} =5$, respectively. Population values for $x_{dtj}$ are generated as independently and identically distributed from a log-normal distribution with a mean of 1.0 and a standard deviation of 0.5 on the log-scale. The response variable is generated as
\[y_{dtj}=100+5x_{dtj}+u_{1,d}+u_{2,t}+e_{dtj}, \ j=1,\dots, N_{dt}\]
where the area and individual effects are generated independently for three scenarios:\\
$[0,0]$ -- absence of outliers, $u_{1,d}\backsim N(0,3)$ and $e_{dtj}\backsim N(0,6)$;\\
$[e,0]$ -- only individual level outliers, $u_{1,d}\backsim N(0,3)$ and $e_{dtj}\backsim\delta N(0,6)+(1-\delta)N(20,150)$, where $\delta$    is an independently generated Bernoulli random variable with $P(\delta=1)=0.97$;
$[e,u]$ -- outliers affect both area and individual effects, $u_{1,d}\backsim N(0,3)$ for areas $1\leq d\leq 36$, $u_{1,d}\backsim N(9,20)$ for areas $37\leq d\leq 40$, and $e_{dtj}\backsim\delta N(0,6)+(1-\delta)N(20,150)$.\\
Regarding time dependency effects, three alternatives are considered:\\
-- Case 1. $\uu_2=\underset{1\leq t\leq T}{\mbox{col}}(u_{2,t})\backsim N_T(\bm{0}, \Sigma_u)$, where $\Sigma_u=\sigma_u^2\Omega_T(\rho)$ and the correlation matrix is
    \begin{align*}
        \Omega_T(\rho)=\frac{1}{1-\rho^2}
        \renewcommand\arraystretch{0.6}\arraycolsep=1pt
        \begin{pmatrix}
            1 & \rho & \cdots&\rho^{T-1}\\
            \rho&1&\ddots&\rho^{T-2}\\
            \vdots&\ddots&\ddots&\vdots\\
            \rho^{T-1}&\rho^{T-2}&\rho&1
        \end{pmatrix}\in\mathcal M_{T\times T}, \quad \rho\in(-1,1).
    \end{align*}
    $\ \ $Case 1.1: $\sigma_u=1$, \ $\rho=0.2$; \ and Case 1.2: $\sigma_u=1$, \ $\rho=0.8$.\\
-- Case 2. Each $u_{2,t}$ is independently generated according to a stationary $AR(3)$ model with coefficients $\phi_1=0.4$, $\phi_2=0.3$, $\phi_3=0.25$ and variance $\sigma=1$.\\
Each setting is independently simulated $S=500$ times. {We compare seven predictors for estimating the small area means:} the H\'ajek estimator with equal weights, the $\text{EBLUP}_1$ based on a $\text{LMM}_1$ with independent area-level random intercepts, the $\text{EBLUP}_2$ based on a $\text{LMM}_2$ with independent area-level and time-level random intercepts; the MQ and BMQ predictors based on the MQ3 models; and the TMQ and BTMQ predictors based on the TWMQ models. {See Section C of the Supplementary Material for further details.}

{The indicators of performance for the predictors are the average absolute relative bias (ARBIAS) and the average relative root-MSE (RRMSE)}:

\begin{eqnarray*}
	\text{ARBIAS} &=& \frac{1}{DT}\sum_{d=1}^D\sum_{t=1}^T\left|{\dfrac{1}{S}\underset{s=1}{\overset{S}{\sum}}\left(\hat {\overline Y}_{dt}^{(s)}- {\overline Y}_{dt}^{(s)} \right) }{{\overline Y_{dt}^*}^{-1}}\right|, \quad \overline{Y}_{dt}^*=\frac{1}{S}\underset{s=1}{\overset{S}{\sum}}\overline Y_{dt}^{(s)} \\
\text{RRMSE} &=& \frac{1}{DT}\sum_{d=1}^D\sum_{t=1}^T\left({\dfrac{1}{S}\underset{s=1}{\overset{S}{\sum}}\left(\hat {\overline Y}_{dt}^{(s)}- {\overline Y}_{dt}^{(s)} \right)^2 }\right)^{1/2}{{\overline Y_{dt}^*}^{-1}}.
\end{eqnarray*}

\subsection{Performance of predictors of small area means}\label{sim1}

Table \ref{ModelSim1} presents the performance measures for Case 1.1, Case 1.2 and Case 2, the simulation scenarios [0,0], [e,0], [e,u] and the predictors mentioned above. 

\begin{table}[t!]
    \centering
    \renewcommand{\arraystretch}{0.5}
    \begin{tabular}{|>{\small}l|rr|rr|rr|}
        \cline{2-7}
        \multicolumn{1}{c|}{}&  \multicolumn{2}{c}{ [0,0] \quad 1--40} &  \multicolumn{2}{|c}{ [e,0] \quad 1--40} & \multicolumn{2}{|c|}{[e,u] \quad 1--40} \\
        \multicolumn{1}{c|}{}&ARBIAS&RRMSE&ARBIAS&RRMSE&ARBIAS&RRMSE\\
        \cline{2-7}
        \multicolumn{7}{l}{Case 1.1  $\quad\quad\  \uu_2\backsim N_T(\mathbf{0}, \Sigma_u)$: $\quad \Sigma_u=\sigma^2_u\Omega_T(\rho)$,  $\quad\sigma_u=1$, $\quad \rho=0.2$} \\[2pt]
        \hline
        $\hat{\overline{Y}}_{dt}^{hajek}$&0.117&3.223&0.133&3.549&0.132&3.523\\
        $\hat{\overline{Y}}_{dt}^{eblup_1}$&0.039&0.833&0.047&0.992&0.067&0.999\\
        $\hat{\overline{Y}}_{dt}^{mq}$&0.041&0.956&0.416&1.105&0.400&1.093\\
        $\hat{\overline{Y}}_{dt}^{bmq}$ &0.036&0.794&0.410&0.964&0.406&0.973\\
        $\hat{\overline{Y}}_{dt}^{eblup_2}$&0.027&0.655&{\bf 0.042}&0.940&{\bf 0.064}&0.947\\
        $\hat{\overline{Y}}_{dt}^{tmq}$&0.023&0.701&0.415&0.864&0.398&0.884\\
        $\hat{\overline{Y}}_{dt}^{btmq}$ &{\bf 0.019}&{\bf 0.553}&0.409&{\bf 0.752}&0.396&{\bf 0.795}\\
        \hline
        \multicolumn{7}{l}{Case 1.2  $\quad\quad\  \uu_2\backsim N_T(\mathbf{0}, \Sigma_u)$: $\quad \Sigma_u=\sigma^2_u\Omega_T(\rho)$,  $\quad\sigma_u=1$, $\quad \rho=0.8$} \\[2pt]
        \hline
        $\hat{\overline{Y}}_{dt}^{hajek}$&0.117&3.223&0.132&3.548&0.132&3.523\\
        $\hat{\overline{Y}}_{dt}^{eblup_1}$&0.039&0.954&0.046&1.096&0.066&1.101\\
        $\hat{\overline{Y}}_{dt}^{mq}$&0.041&1.066&0.414&1.200&0.399&1.191\\
        $\hat{\overline{Y}}_{dt}^{bmq}$ &0.035&0.897&0.407&1.051&0.403&1.053\\
        $\hat{\overline{Y}}_{dt}^{eblup_2}$&0.027&0.682&{\bf 0.042}&0.995&{\bf 0.063}&1.001\\
        $\hat{\overline{Y}}_{dt}^{tmq}$&0.025&0.693&0.413&0.878&0.397&0.897\\
        $\hat{\overline{Y}}_{dt}^{btmq}$&{\bf 0.020}&{\bf 0.539}&0.407&{\bf 0.760}&0.396&{\bf 0.800}\\
        \hline
        \multicolumn{7}{l}{Case 2  $\quad\quad\quad u_{2,t}\backsim AR(3)$: $\ \ \phi_1=0.4$, $\ \  \phi_2=0.3$, $\ \ \phi_3=0.25$, $\quad  \sigma=1$} \\[2pt]
        \hline
        $\hat{\overline{Y}}_{dt}^{hajek}$&0.115&3.238&0.122&3.571&0.121&3.545\\
        $\hat{\overline{Y}}_{dt}^{eblup_1}$&0.021&0.838&0.031&1.001& 0.064&1.007\\
        $\hat{\overline{Y}}_{dt}^{mq}$&0.027&0.962&0.414&1.114&0.396&1.101\\
        $\hat{\overline{Y}}_{dt}^{bmq}$&0.022&0.800&0.408&0.971&0.402&0.979\\
        $\hat{\overline{Y}}_{dt}^{eblup_2}$&0.018&0.652&{\bf 0.030}&0.945&{\bf 0.062}&0.951\\
        $\hat{\overline{Y}}_{dt}^{tmq}$&0.022&0.698&0.416&0.874&0.394&0.892\\
        $\hat{\overline{Y}}_{dt}^{btmq}$ &{\bf 0.017}&{\bf 0.548}&0.409&{\bf 0.758}&0.394&{\bf 0.800}\\
        \hline
    \end{tabular}
    \caption{ARBIAS and RRMSE values (in \%) of predictors of small area population means. The best values of performance are in bold.}\label{ModelSim1}
\end{table}

It is noteworthy that if time effects are normal, as defined in Case 1, the improvement of the TWMQ models is, as expected, subject to the definition of the covariance matrix. These models do not aim to capture random effects over time, but rather well-founded relationships of time dependency. For more details, see Section C of Supplementary Material. Provided that the normality assumptions are met and there are no outliers, the values of $\rho$ and $\sigma^2$ are crucial. If the variance is large enough, the $\text{LMM}_2$ provides acceptable ARBIAS and RRMSE values. If the correlation is higher and/or the variance is lower, the latter is not achieved, and the best options are the TWMQ models. The same is true for the presence of outliers.  It should be mentioned, however, that the $\text{LMM}_2$ is not the model that generates the target variables as it does not take into account the correlation structure of the time-level random intercepts. As far as we are aware, the available correlation structures in \texttt{R} refer to the model errors $e_{dtj}$, and not to the random effects $u_{2,t}$. In addition, if the time effects follow an AR model, as defined in Case 2, the TWMQ models are better in terms of RRMSE, although not in terms of ARBIAS if there are area-level outliers. In fact, they are expected to be robust in the presence of atypical data, inheriting the well-known robustness properties of the MQ regression. In that case, the ARBIAS of the $\text{EBLUP}_2$ is smaller and the main source of contribution to its RRMSE comes from the variance.

Regarding the TMQ and BTMQ predictors, a proper selection of $c_{\phi}$ ensures that the latter is much better than the former, correcting for bias but also mitigating variability. The same applies to the comparison of predictors MQ and BMQ. However, the flexibility of the TWMQ models leads to less bias correction without severely affecting the variance.
Rather than merely accepting a default value as adequate, area-time specific values of $c_\phi$ are far preferable and very promising, being able to outperform the $\text{EBLUP}_2$ based on the $\text{LMM}_2$. Moreover, an unexpected advantage is that the set $\big\{\hat c_{\phi, dt}: \ d=1,\dots,D, \ t=1,\dots, T\big\}$ can be used as a diagnostic tool for outlier detection (see Section \ref{sim2}) 
and, not least important, computational effort and execution times are not affected either if we use optimum values of $c_\phi$. The BTMQ predictor is a plug-in type predictor, so it is easy to program and fast to calculate.

\subsection{Reliability of robustness parameters for outlier detection}\label{sim2}

We illustrate how the area-time-specific robustness parameters for bias correction can be used for outlier detection. The discussion focuses on Case 1.1, although the conclusions are similar for the remaining two cases. Let us define the average value, across simulations, of $\hat c_{\phi,dt}$, given by $\overline{\hat c}_{\phi,dt}=\frac{1}{S}\overset{S}{\underset{s=1}\sum} \hat c_{\phi,dt}^{(s)}$.


First, some basic descriptive measures are calculated.
For Scenario [0,0], $\overline{\hat c}_{\phi,dt}$ ranges from 0.380 to 0.624, with a median value of 0.491. A similar pattern is observed for Scenario [e,0], where $\overline{\hat c}_{\phi,dt}$ ranges from 0.365 to 0.594, and the median is 0.474. In Scenario [e,u], for the non-atypical areas ($1 \leq d \leq 36$), $\overline{\hat c}_{\phi,dt}$ ranges from 0.338 to 0.492, with a median of 0.433. However, for the atypical areas ($37 \leq d \leq 40$), $\overline{\hat c}_{\phi,dt}$ ranges from 1.736 to 2.075, with a median of 1.887. Therefore, $\overline{\hat c}_{\phi,dt}$ is unaffected by the presence of unit-level outliers, but a very different picture emerges with the presence of area-level outliers.
This can be observed in Figure \ref{selec.c.dt}, where the peaks corresponding to areas $37\leq d \leq 40$ are quite revealing. Compared to Scenario [0,0], the presence of individual-level outliers stops the improvement in bias earlier, to avoid overfitting. If outliers are also at the area level, 
it is possible to give more strength to the bias correction without triggering the variance.

\begin{figure}[h]
    \centering
    \begin{subfigure}{1\textwidth}
        \centering
        \includegraphics[width=0.38\linewidth]{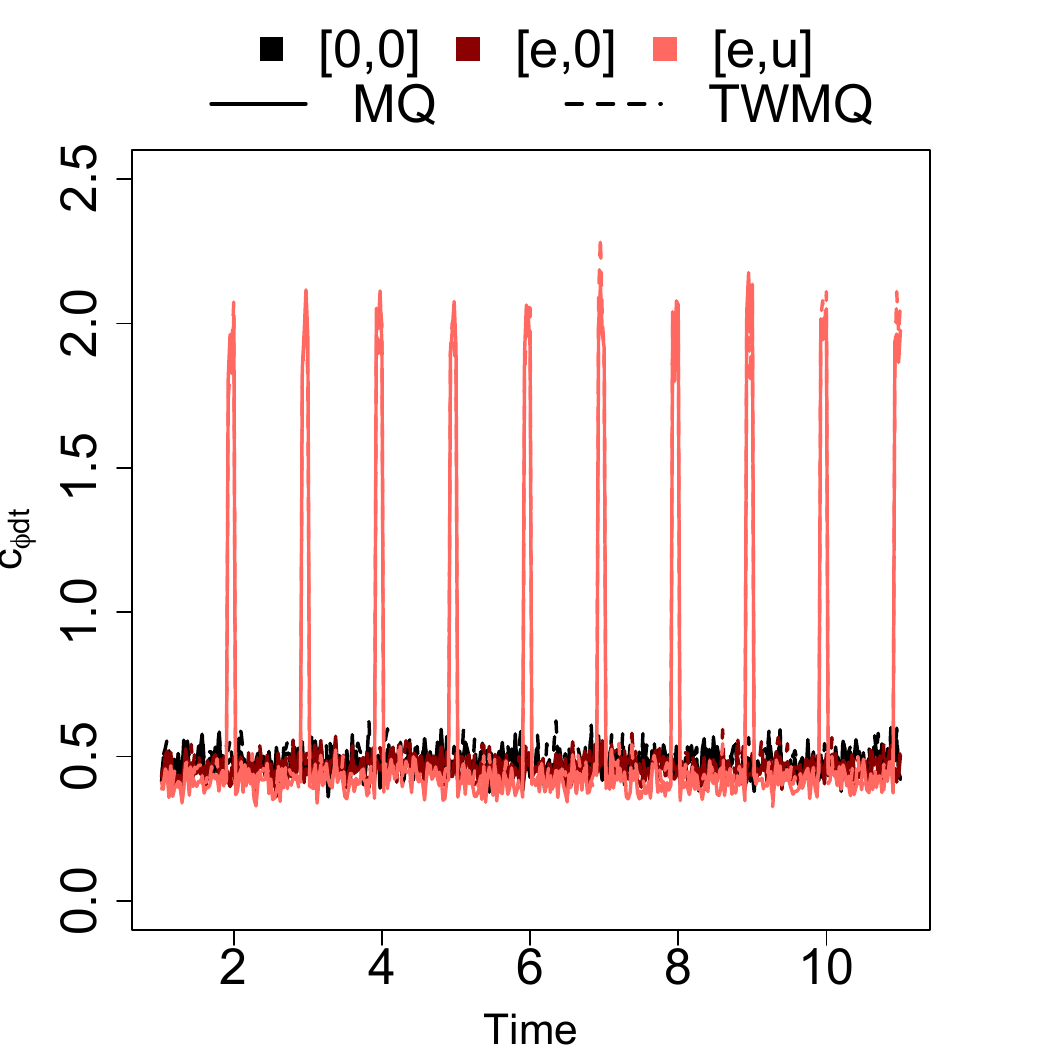}
    \end{subfigure}
    \caption{Scartterplot of $\overline{\hat c}_{\phi,dt}$ sorted by area and time period for Case 1.1.}
    \label{selec.c.dt}
\end{figure}

After reviewing the literature, the results are supported by statistical tests adapted to our problem. First,  \cite{Friedman1937} test is used for one-way repeated measures analysis of variance under different experimental conditions. The repeated measures are the values of $\overline{\hat c}_{\phi,dt}$ selected for an area $d$, $d=1,\dots, D$, over all time periods (which are the experimental conditions). It is assumed that each $\overline{\hat c}_{\phi,dt}$ is equally distributed except at most in terms of location, which may vary according to the experimental condition and area. We write
\[\overline{\hat c}_{\phi,dt}=c+\delta_{1,t}+\delta_{2,d}+e_{dt},\]
where $c$ is the global mean, independent of both area and time;
$\delta_{1,t}$ is the average effect of time period $t$, $t=1,\dots, T$; $\delta_{2,d}$ measures the average effect of the $d$-th area, $d=1,\dots, D$; and each $e_{dt}$ i.i.d. follows an unknown zero-mean distribution $F$. The objective is to test
\[H_0: \delta_{2,1}=\dots,=\delta_{2,D} \ \text{ vs }\  H_1: \exists\ d_1,d_2\in\{1,\dots,D\},\ d_1\neq d_2, \ \delta_{2,d_1},  \neq \delta_{2,d_2}.\]
In a second step, we apply an Honestly Significant Difference (HSD) test, or Tukey's multiple range test \citep{Tukey1949} in order to detect which groups of areas shift in scale. 


Our approach for outlier detection reports the following promising results. In Scenario [e,u], the mean of $\overline{\hat c}_{\phi,dt}$ differs between areas (p-value $< 2.2\cdot 10^{-16}$) but not between time periods (p-value $0.356$). As expected, the area-level outliers detected by Tukey's test are exactly those with index $37\leq d \leq 40$. The same test applied to time periods detects only a single group, formed by all of them. 
From the perspective of LMMs, the detection of atypical data using the methodology described in \cite{Zewotir2007} does not report conclusive results because it is performed at unit-level. 
In our research, the generation of unit-level outliers has been random, 
making the above-mentioned test useless. 

\subsection{Performance of MSE estimators}\label{simMSE}

Finally, we examine the performance of the different  estimators for the MSE of the TMQ and BTMQ predictors. We are particularly interested in those relating to the BTMQ predictor, as one of our MSE estimator comes from a first-order approximation.  Figure \ref{RBIAS.RRMSE} shows boxplots of RBIAS and RRMSE, in \%, for the RMSE estimators' performance in Case 1.1. For tabular results, see Section C of Supplementary Material. 

\begin{figure}[t!]
    \centering
    \begin{subfigure}{.3\textwidth}
        \includegraphics[width=1\linewidth]{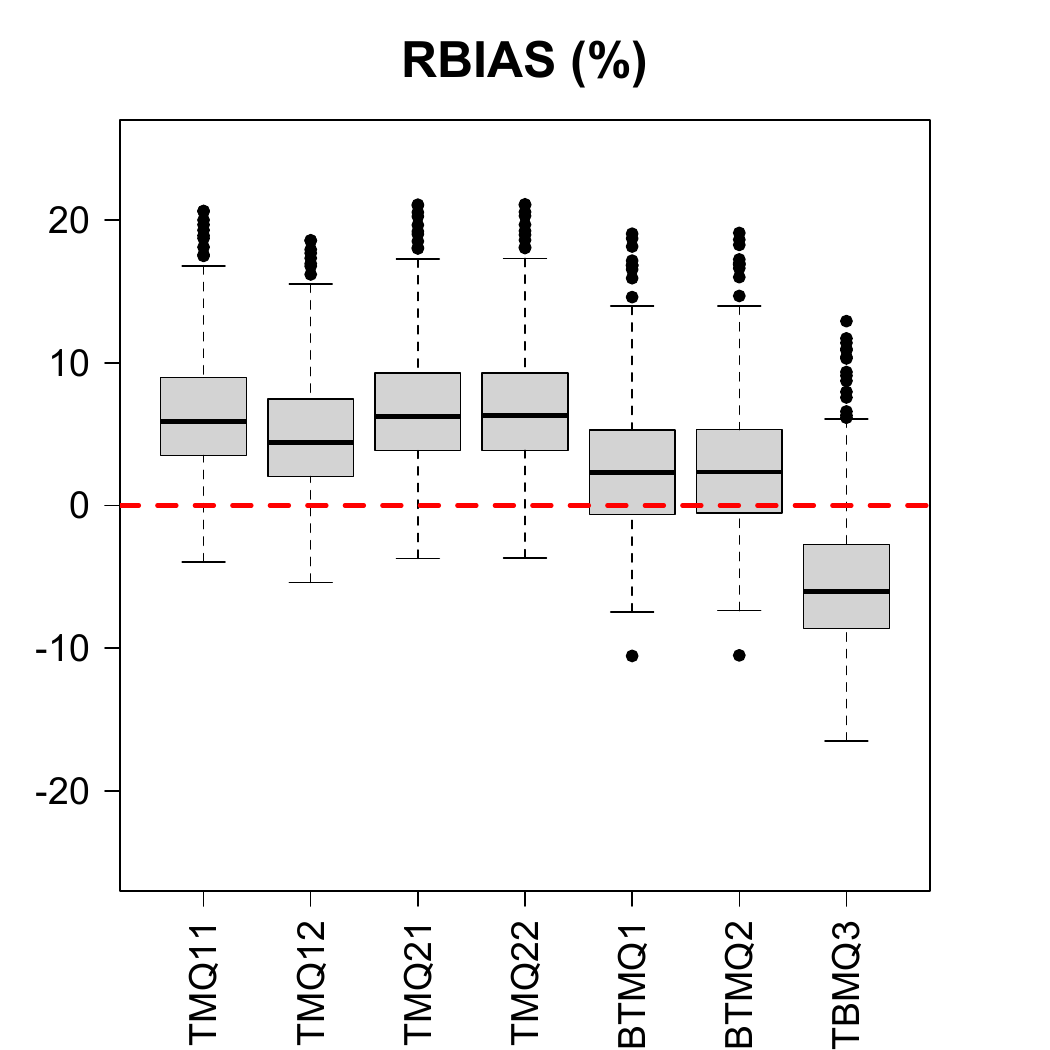}
    \end{subfigure}
    \begin{subfigure}{.3\textwidth}
        \includegraphics[width=1\linewidth]{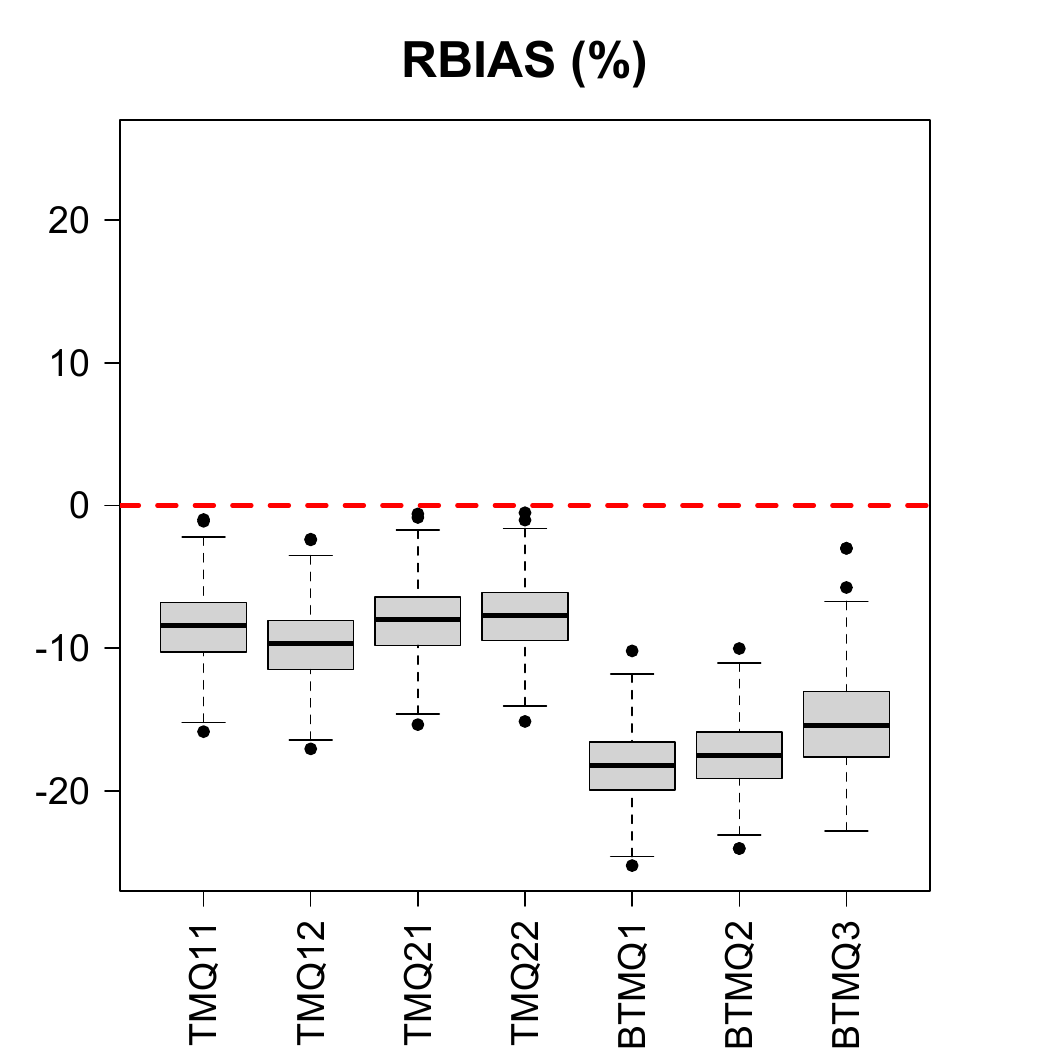}
    \end{subfigure}
    \begin{subfigure}{.3\textwidth}
        \includegraphics[width=1\linewidth]{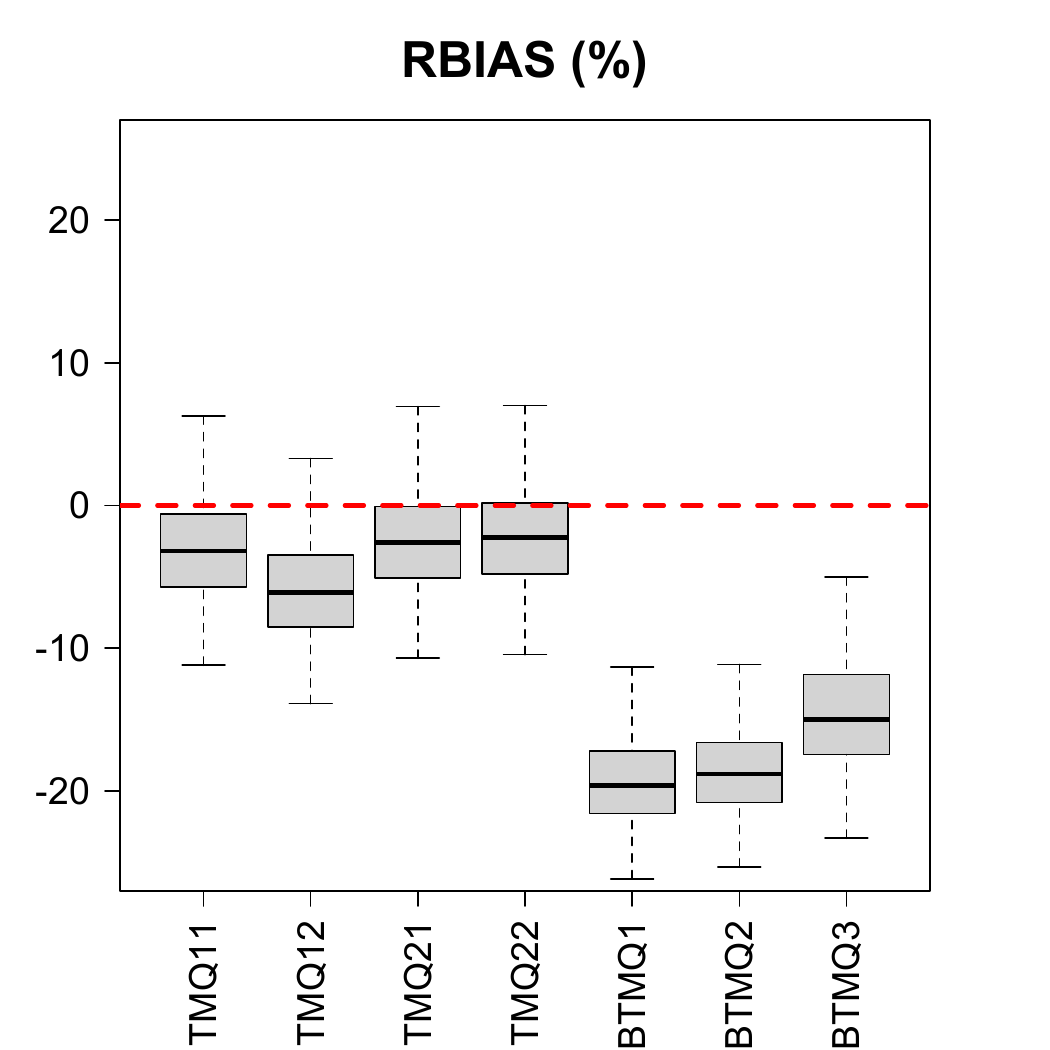}
    \end{subfigure}

    \centering
    \begin{subfigure}{.3\textwidth}
        \includegraphics[width=1\linewidth]{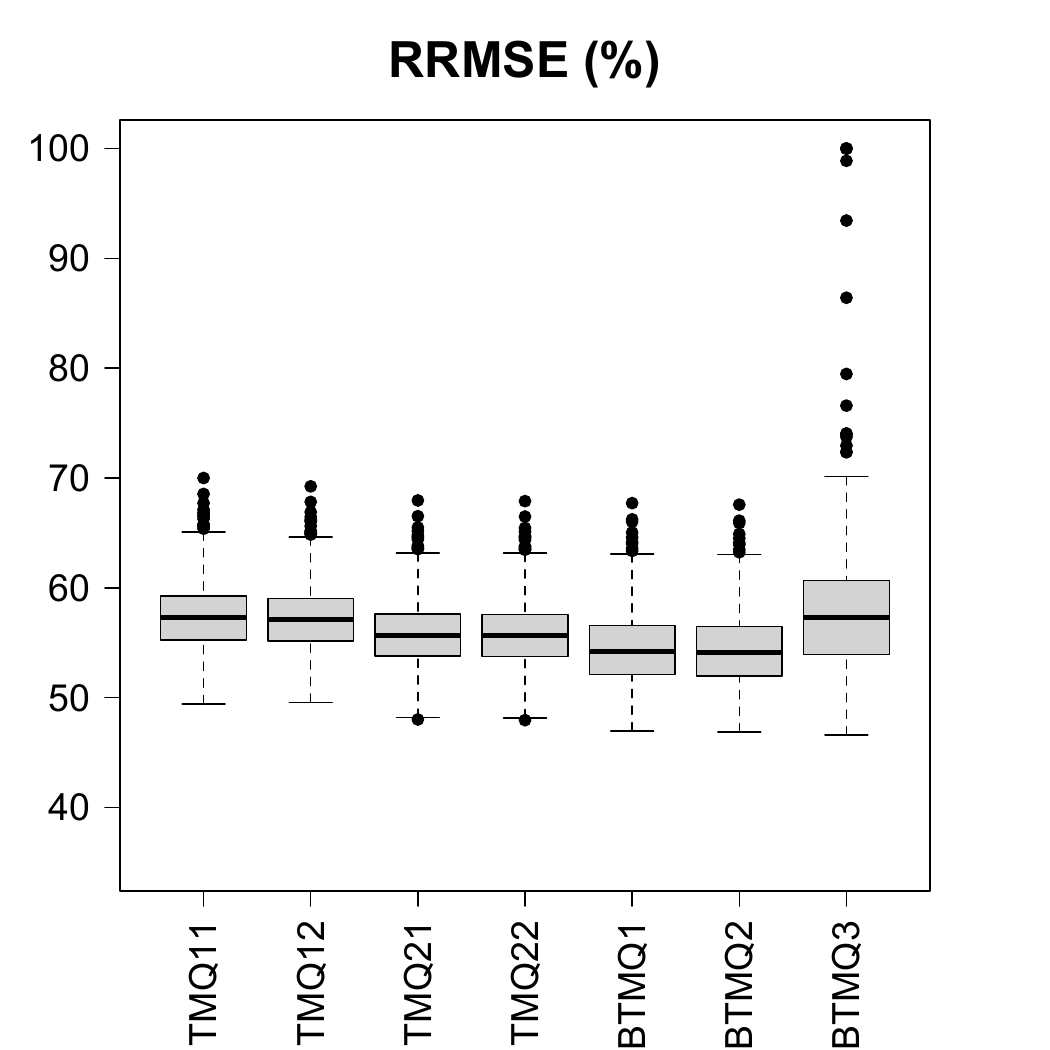}
        \caption{Scenario [0,0].}
    \end{subfigure}
    \begin{subfigure}{.3\textwidth}
        \includegraphics[width=1\linewidth]{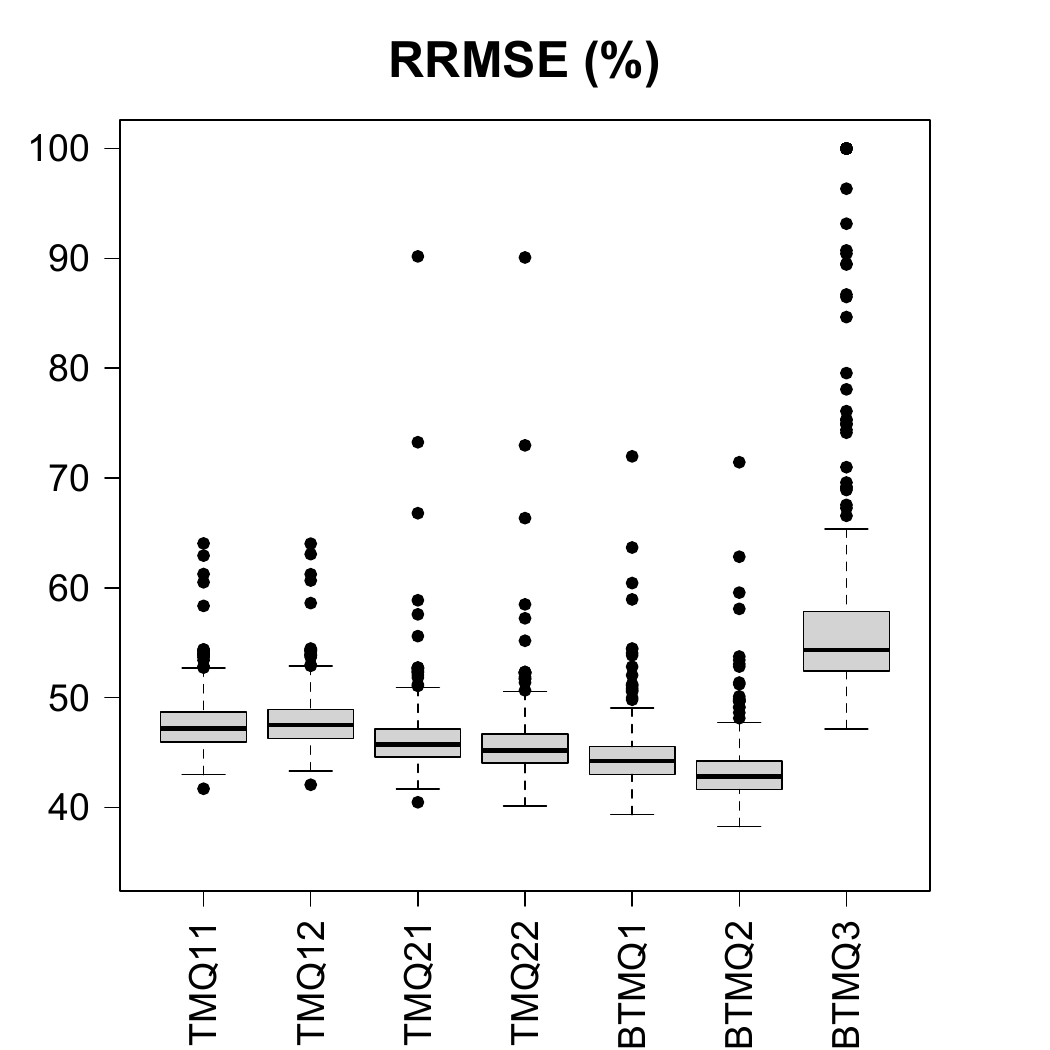}
        \caption{Scenario [e,0].}
    \end{subfigure}
    \begin{subfigure}{.3\textwidth}
        \includegraphics[width=1\linewidth]{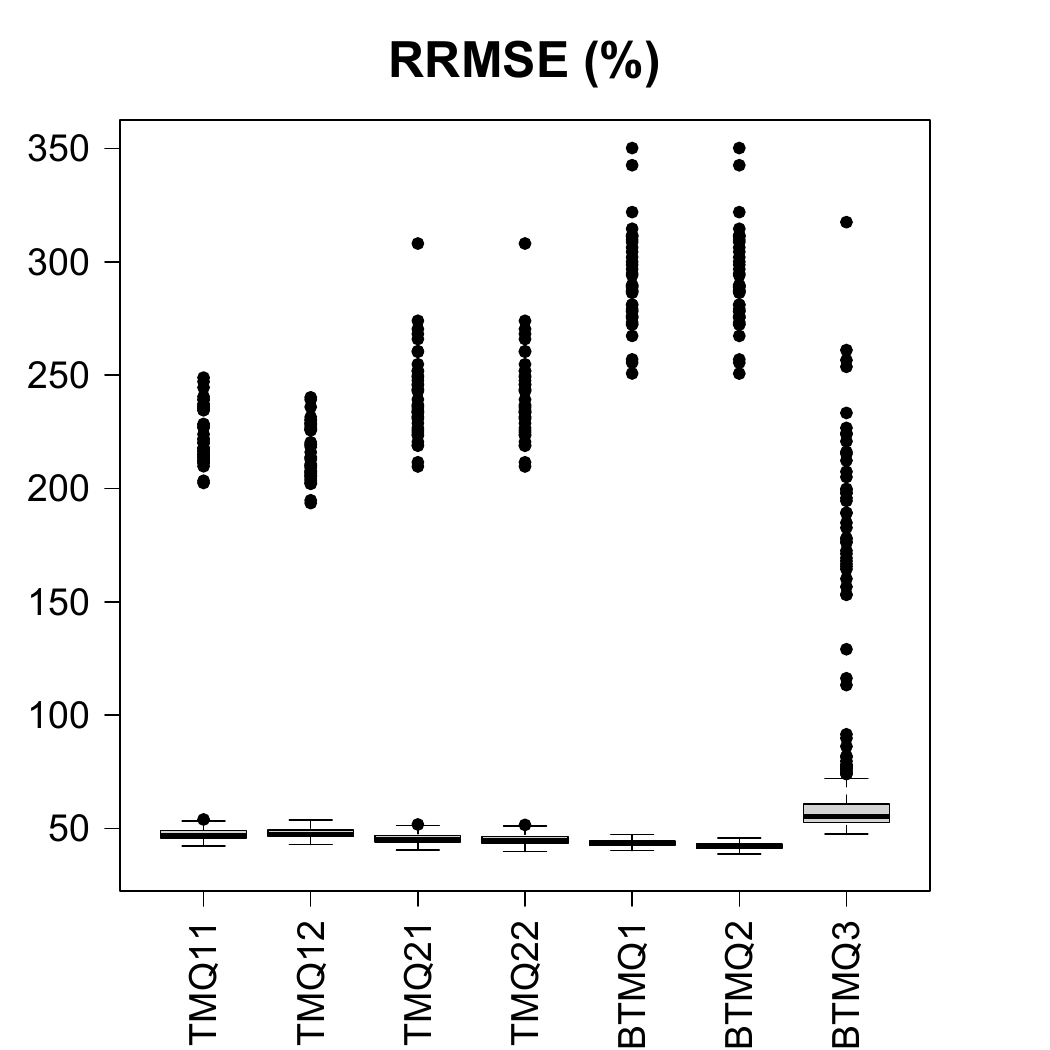}
        \caption{Scenario [e,u].}
    \end{subfigure}
    \caption{Boxplots of RBIAS and RRMSE in (\%) for the RMSE estimators of the proposed predictors in Case 1.1.}
    \label{RBIAS.RRMSE}
\end{figure}

Looking at Figure \ref{RBIAS.RRMSE}, the RBIAS of the RMSE estimators of the TMQ predictors is more positive (or less negative, as appropriate) than the corresponding one for the BTMQ predictors. As for the sign of the RBIAS, there are overestimates in Scenario [0,0] and underestimates in the other two scenarios. In addition, the RRMSE increases with the presence of area-level outliers.
In general terms, although $rmse_{3,dt}^{btmq}$ is calculated from a first-order approximation, we feel that the numerical instability problems involved in the estimation stage are responsible for its ``slightly worse'' performance. As discussed in Section A of Supplementary Material, the theoretical advantage of this estimator is largely overshadowed by the highly unstable estimation of one of its variance terms.
Taking into account the latter, in the application to real data we have used $rmse_{2,dt}^{btmq}$ to present error measures about the BTMQ predictor. It should be noted that the results for $rmse_{1,dt}^{btmq}$ are almost the same but slightly worse. The next steps will be to develop MSE estimators based on second-order unbiasedness to improve the results. For the TMQ predictor, we propose using $rmse_{2,dt}^{tmq}$.

{Other simulation results, reported in the Supplementary Material for brevity, show that the IRLS algorithm for the TWMQ model for the estimation of model parameters performs quite well; moreover, the proposed procedure also gives good results of the domain means of unit-level MQ coefficients. In addition, we also include more results for the MSE estimation and assess how the time component variability and the temporal correlation affect the TMQ and BTMQ predictors and MSE estimators. Finally, in the Supplementary Material we also show the asymptotic consistency of the estimators $\hat\theta_d$ of the domain population means of unit-level MQ coefficients $\theta_d$ for MQ3 models. }

\section{Application to real data}\label{sec.aplic}
In this section, we apply the proposed methodology to assess changes in the average level of income in 23 provinces (NUTS 3 level) of Empty Spain, which refers to those provinces that have lost inhabitants between 1950 and 2019 and that also have a population density below the national average. 
Survey data are from the 2013-2022 SLCS ($T = 10$ years) while the auxiliary variables come from the census in 2021 provided by the Spanish National Institute of Statistics. The SLCS is designed to obtain reliable direct estimators in NUTS 2 regions, but sample sizes are quite small in NUTS 3 territories. Indeed, they range from 36 
to 1762,
with a median value of 293. The response variable is the equivalized disposable income, per person and unit of consumption, measured in thousands of euros.

Since we are dealing with unit-level data, only the following auxiliary variables are available: {\it sex}, a categorical variable with two categories, (male, reference category, and female); and {\it age4}, a categorical variable with four categories: 0--25 (reference category, {\it age4-1}), 26--45 ({\it age4-2}), 46--64 ({\it age4-3}) and over 65 ({\it age4-4}). 

Figure \ref{beta.coef} shows boxplots of the estimation of the regression coefficients by year obtained by fitting the TWMQ models. In terms of model specification, it shows that the assumption of identical regression coefficients over time is not reasonable. It therefore supports the suitability of the TWMQ models.
To have more confidence in the fitted models as the true generating ones, their validation is addressed in Section D of the Supplementary Material.

\begin{figure}[H]
    \centering
    \hspace{-12mm}
    \includegraphics[width=0.23\linewidth]{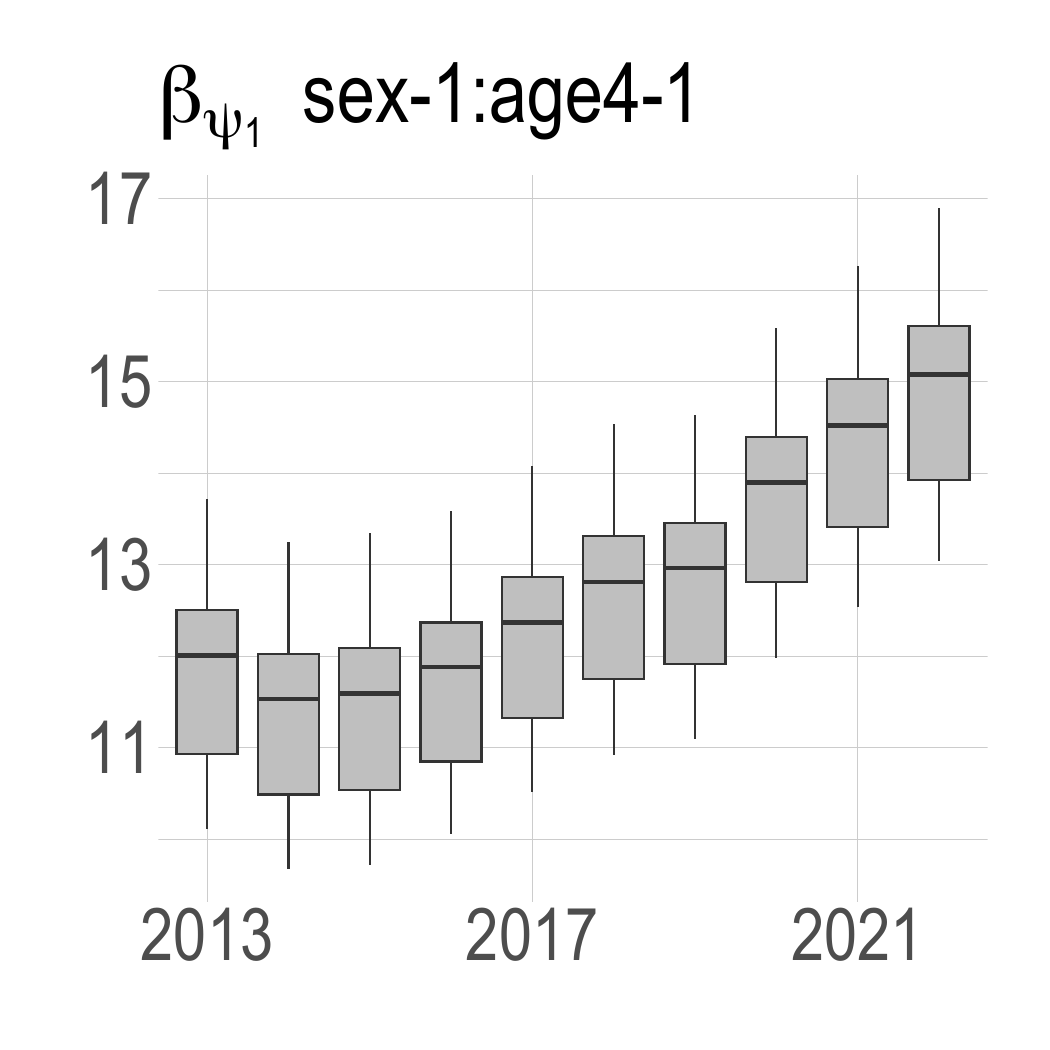}
    \hspace{-5.2mm}
    \includegraphics[width=0.23\linewidth]{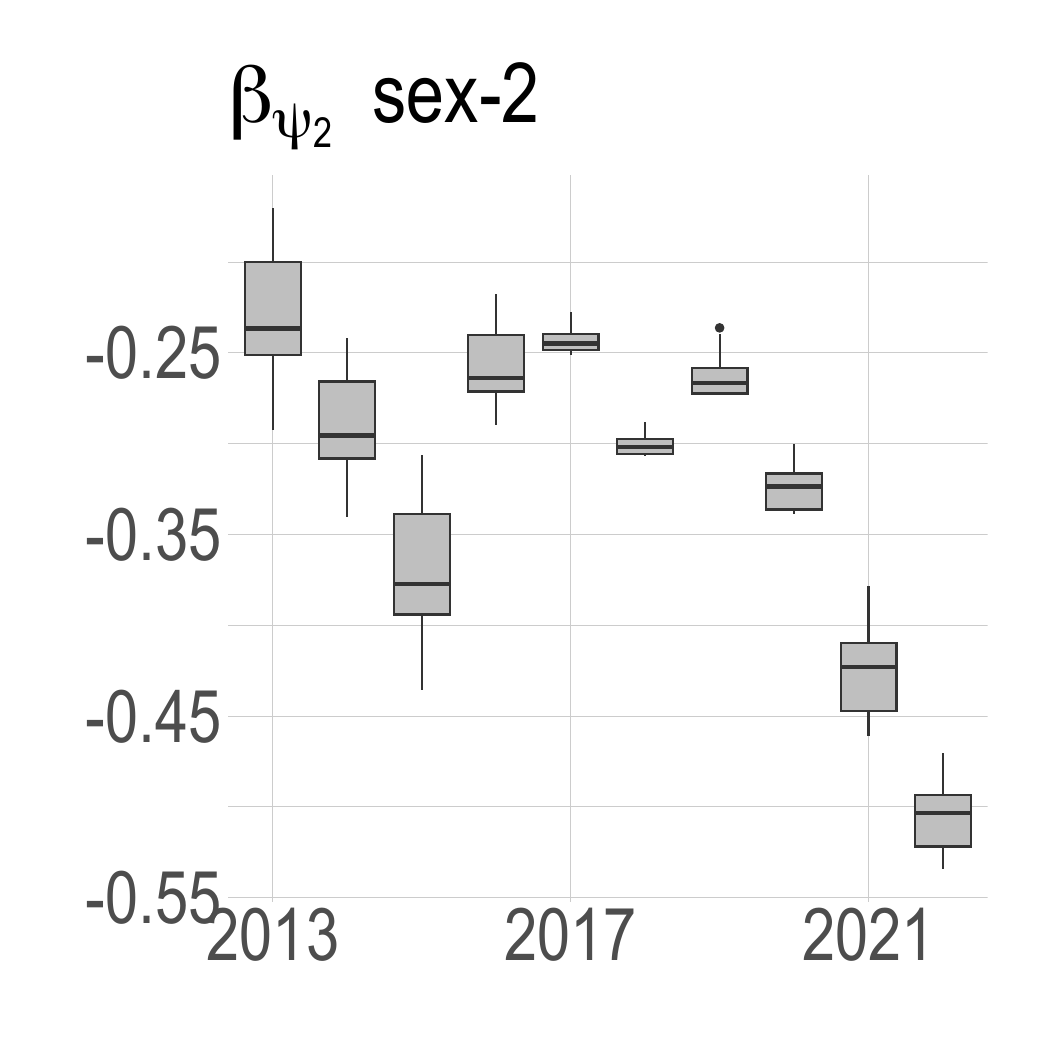}
    \hspace{-5.2mm}
    \includegraphics[width=0.23\linewidth]{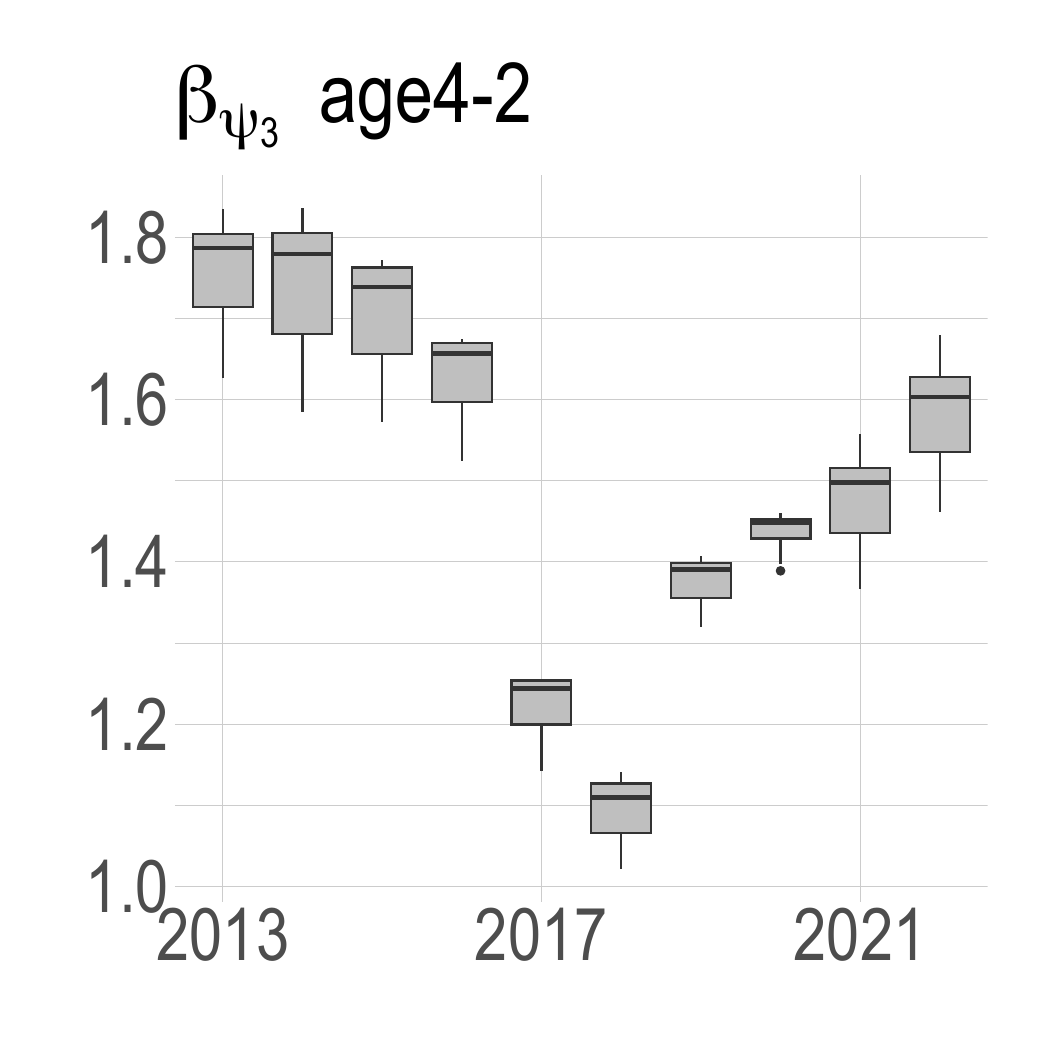}
    \hspace{-5.2mm}
    \includegraphics[width=0.23\linewidth]{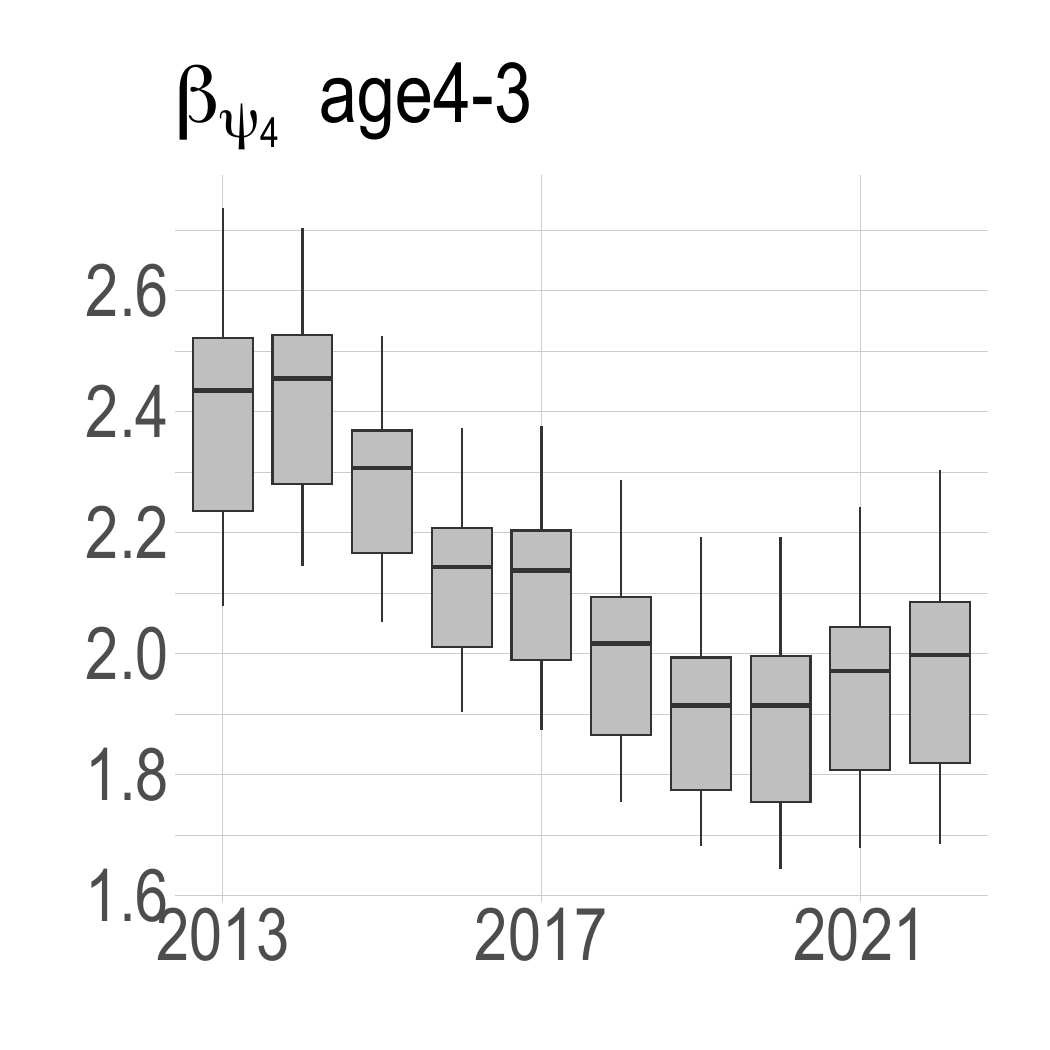}
    \hspace{-5.2mm}
    \includegraphics[width=0.23\linewidth]{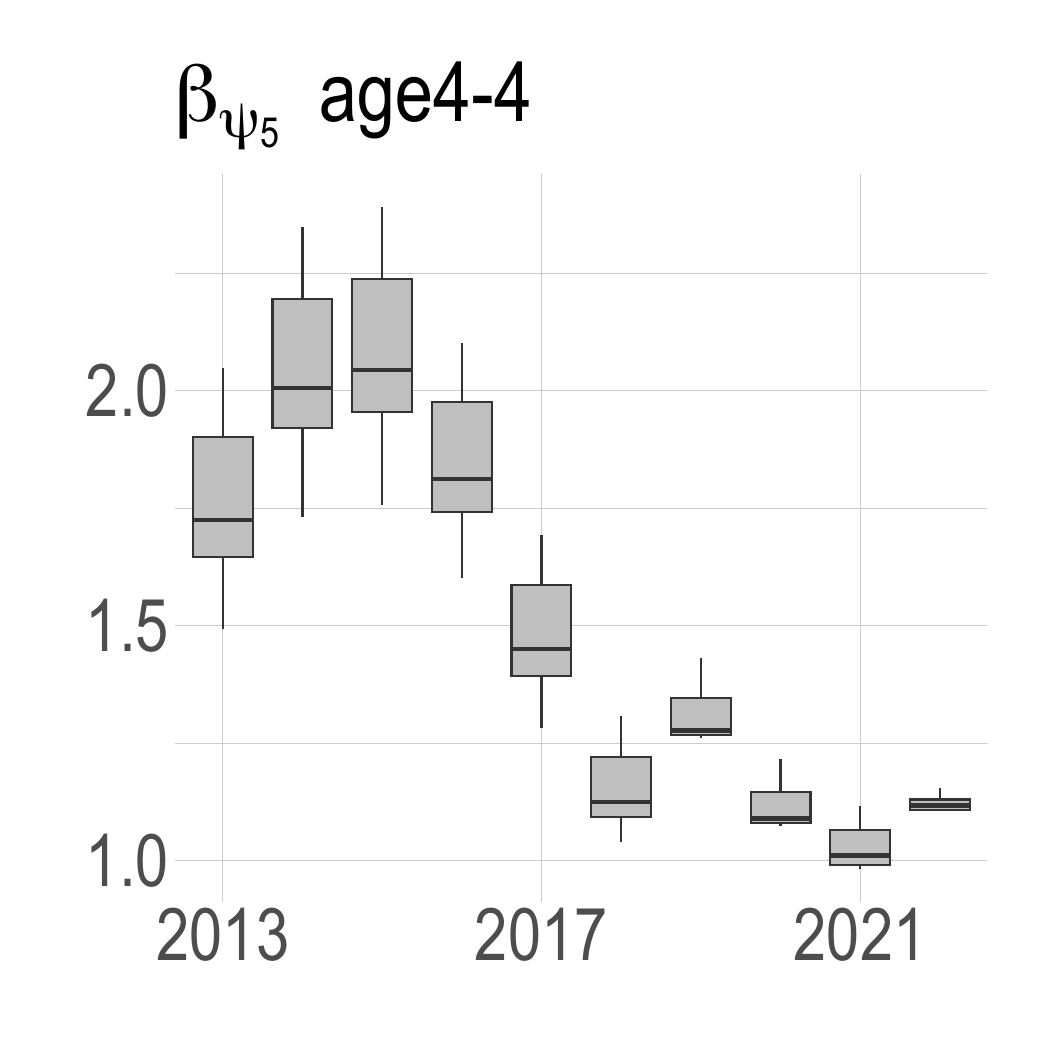}
    \hspace{-14mm}
    \caption{Boxplot of the estimates of TWMQ regression parameters by year.}
    \label{beta.coef}
\end{figure}


\subsection{Predictions and error measures}\label{pred.error}

This section provides H\'ajek direct estimates and model-based predictions of small area population means by province and year, as well as error measures.
{As pointed out in  Section \ref{sec.tmq}, the TMQ predictor may introduce nonnegligible prediction biases, but the BTMQ predictor can unbalance the bias-variance trade-off of the MSE. To set the value of the robustness parameter $c_\phi$, we use the selection criterion proposed in Section \ref{sec.c}.}

Figure \ref{pred.Hajek} plots H\'ajek estimates and model-based predictions for the TMQ (left) and BTMQ (right) predictors. 
{The BTMQ estimator seems to smoothen the estimates more, as expected, employing a bias correction term.}

\begin{figure}[ht]
    \centering
    \begin{subfigure}{.7\textwidth}
        \includegraphics[width=0.47\textwidth]{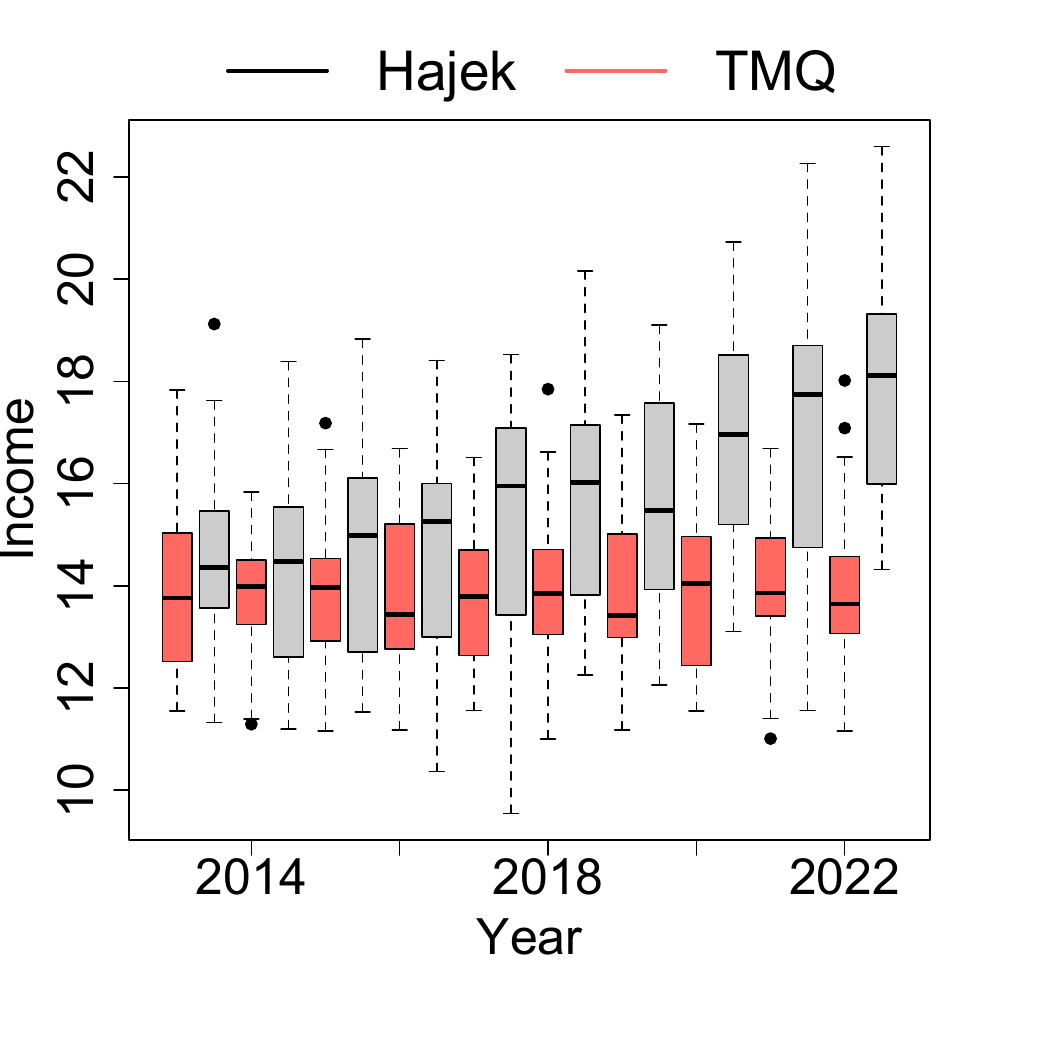}
        \includegraphics[width=0.47\textwidth]{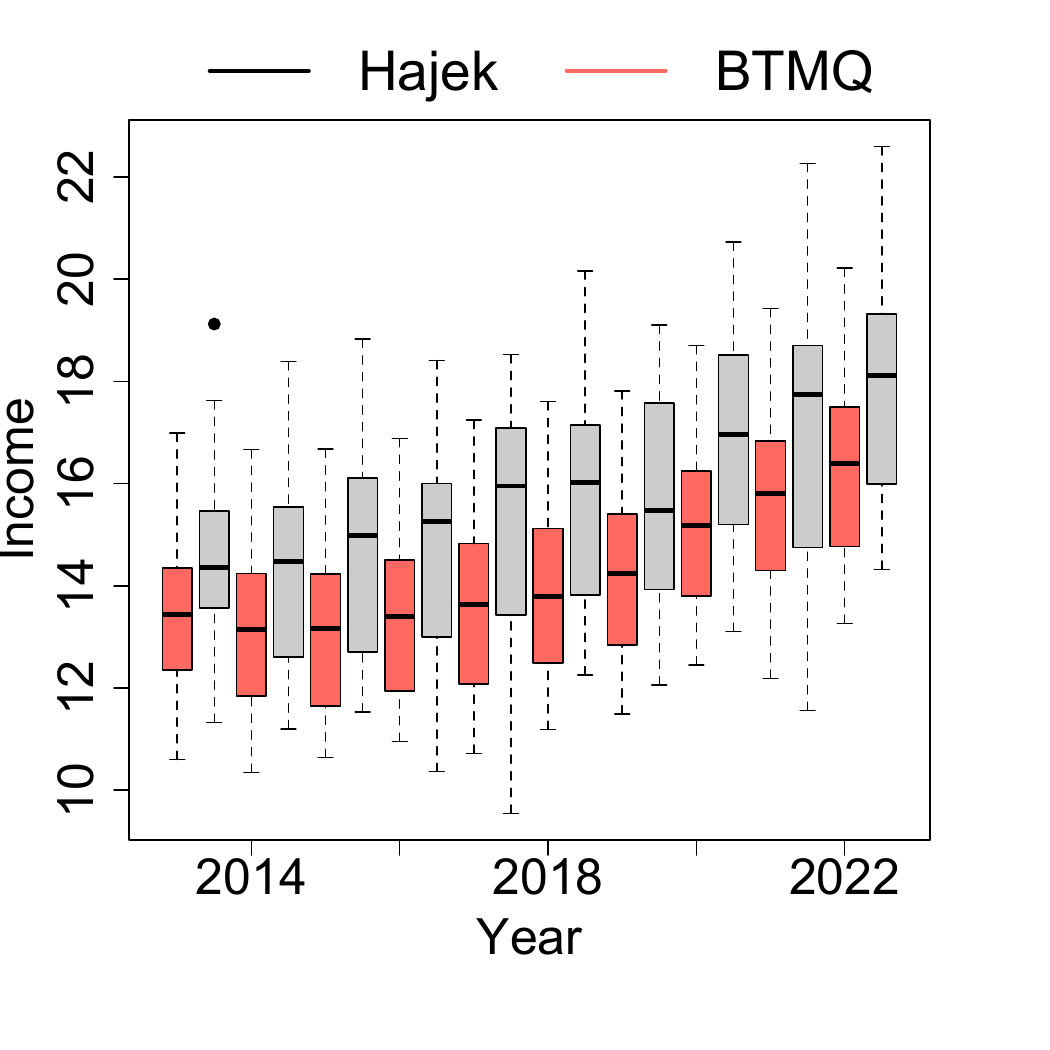}
    \end{subfigure}
    \caption{Boxplots of model-based predictions and direct estimates by year.}
    \label{pred.Hajek}
\end{figure}

{Back to the variability of the estimates, we focus on the MSE of the BTMQ predictor.}
Table \ref{var} contains the deciles of the sample sizes $n_{dt}$, of the standard deviations of the H\'ajek estimator (see \cite{morales2021}, Chap. 3) and of the BTMQ predictor. The reduction in variability is evident, especially when the sample sizes are small.

\begin{table}[H]
    \centering
    \small
    \renewcommand{\arraystretch}{0.5}
    \begin{tabular}{|l|ccccccccccc|}
        \cline{2-12}
        \multicolumn{1}{l|}{}
        &$q_0$&$q_{0.1}$&$q_{0.2}$&$q_{0.3}$&$q_{0.4}$&$q_{0.5}$&$q_{0.6}$&$q_{0.7}$&$q_{0.8}$&$q_{0.9}$&$q_1$\\
        \hline
        $n_{dt}$&36& 118& 146 &182 &239 &294& 346& 450& 568& 900 &1762\\
        H\'ajek&0.245  & 0.361   &0.407 &  0.455 &  0.511  & 0.561 &  0.648  & 0.727 &  0.859 &  1.118  &  2.393\\
        BTMQ& 0.032 &0.040 &0.088 &0.126 &0.160& 0.184& 0.213& 0.251& 0.307& 0.397 &1.023\\
        \hline
    \end{tabular}
    \caption{Sample sizes and standard deviations of the H\'ajek estimator and BTMQ predictor.}
    \label{var}
\end{table}

%
%
%
The proposed estimation procedure offers the opportunity to analytically read the evolution and differences between the provinces of Empty Spain over time. 
Figure \ref{map.Income} maps the equivalized disposable income obtained with the BTMQ predictor. 
\begin{figure}[h]
    \centering
    \hspace{-13mm}
    \includegraphics[width=0.39\linewidth]{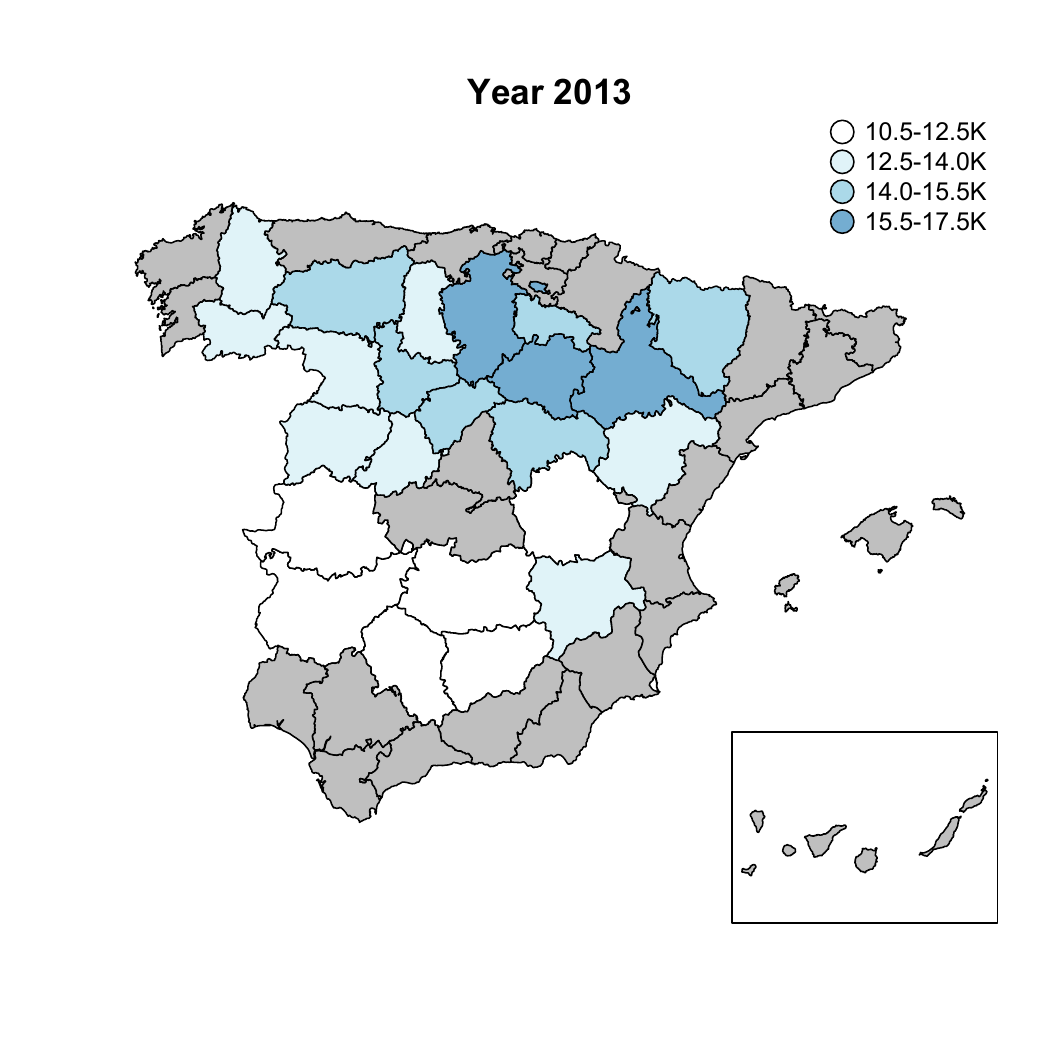}
    \hspace{-11.5mm}
    \includegraphics[width=0.39\linewidth]{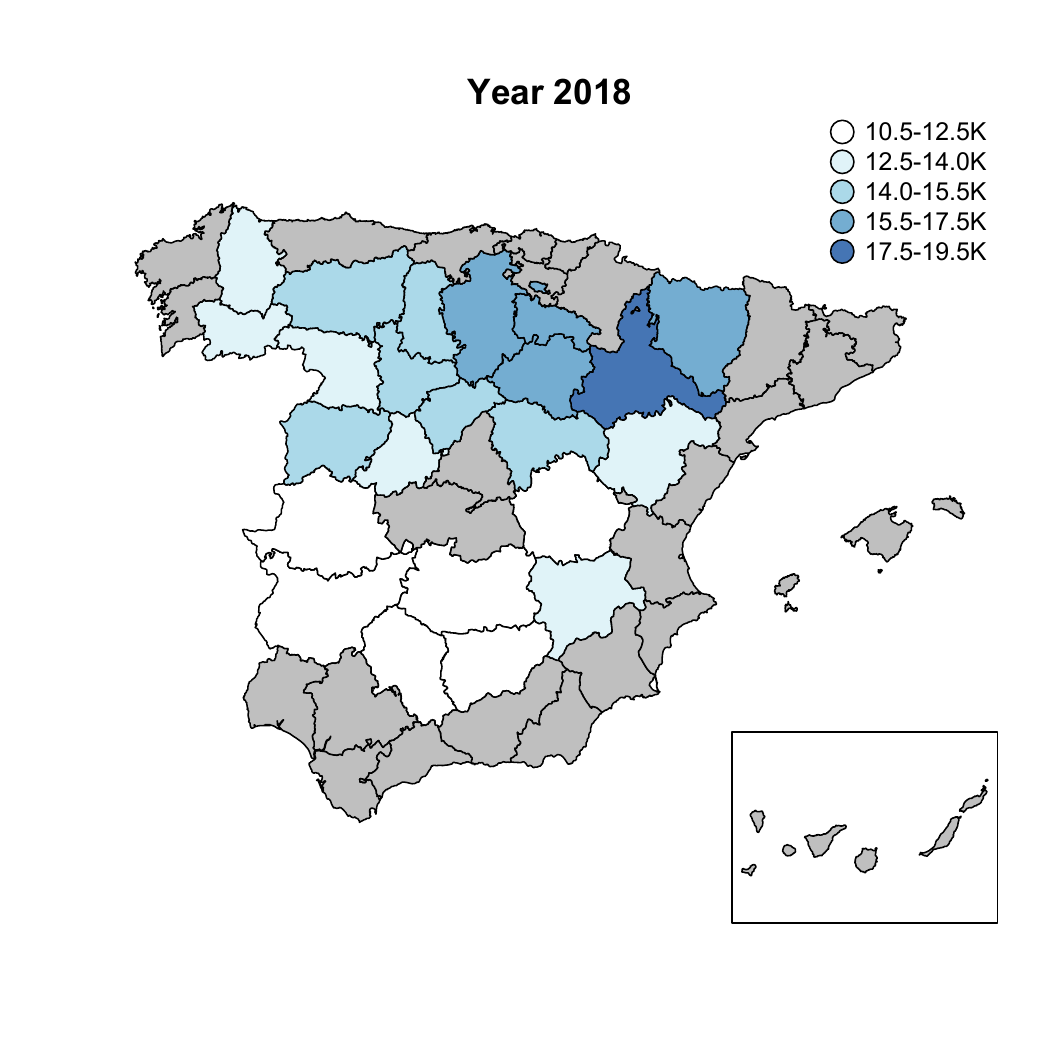}
    \hspace{-11.5mm}
    \includegraphics[width=0.39\linewidth]{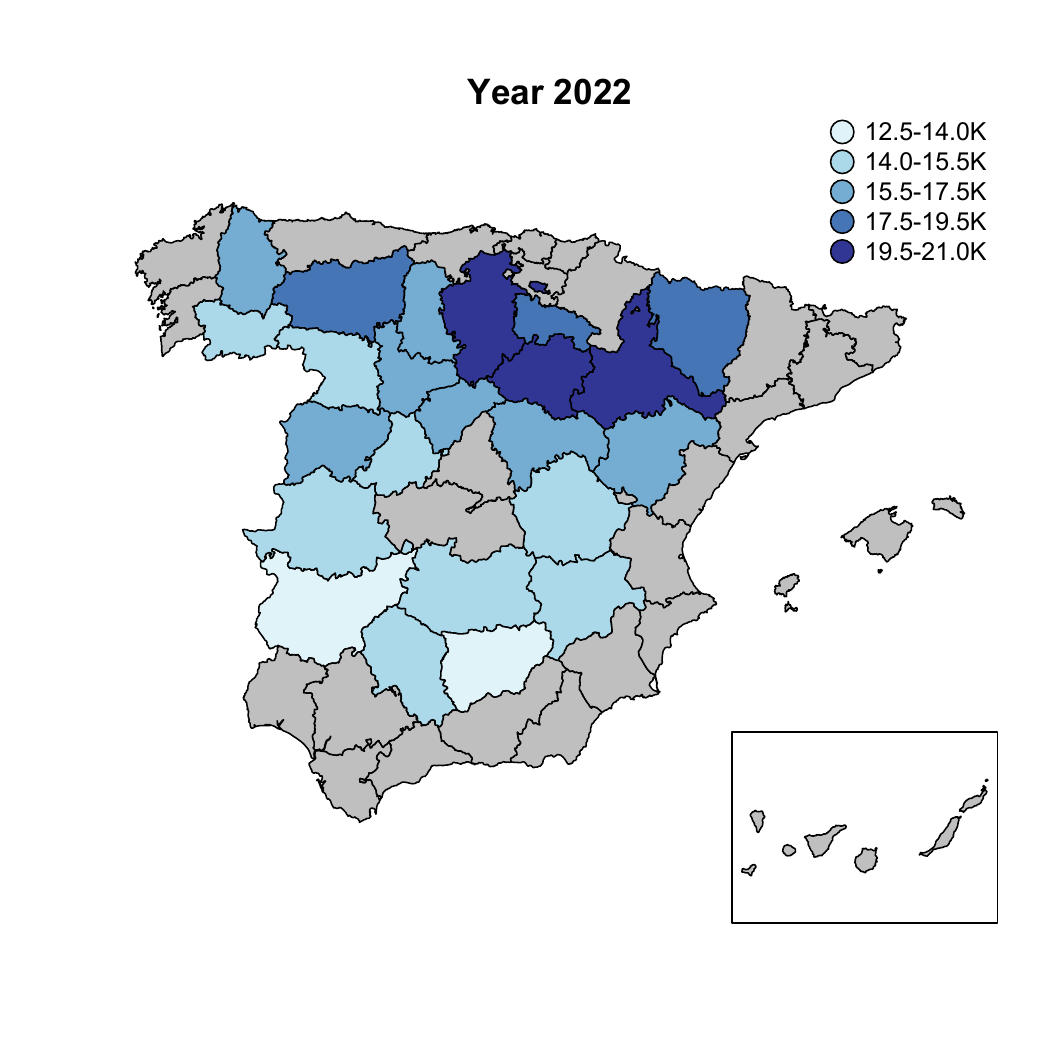}
    \hspace{-10mm}
    \caption{Estimates of Equivalized disposable income for Empty Spain in 2013 (left), 2018 (center) and 2022 (right). Results obtained using the BTMQ predictor.}
    \label{map.Income}
\end{figure}

Figure \ref{map.Income} points out that there are clear differences between the northern provinces, historically richer and more developed, and those in the centre-south, where agriculture and construction predominate and the industrial sector is less promoted. 
Finally, it is worth noting the increasing, or at least non-decreasing, trend during the study period.
Figure \ref{map.RRMSE} shows maps of the coefficient of variation using the $rmse_{2,dt}^{btmq}$ estimator proposed in Section \ref{MSE.BTMQ}. It follows that the relative margins of error are accurate enough for a SAE problem, with coefficients of variation lower than 9\% in most domains.

\begin{figure}[h]
    \centering
    \hspace{-13mm}
    \includegraphics[width=0.39\linewidth]{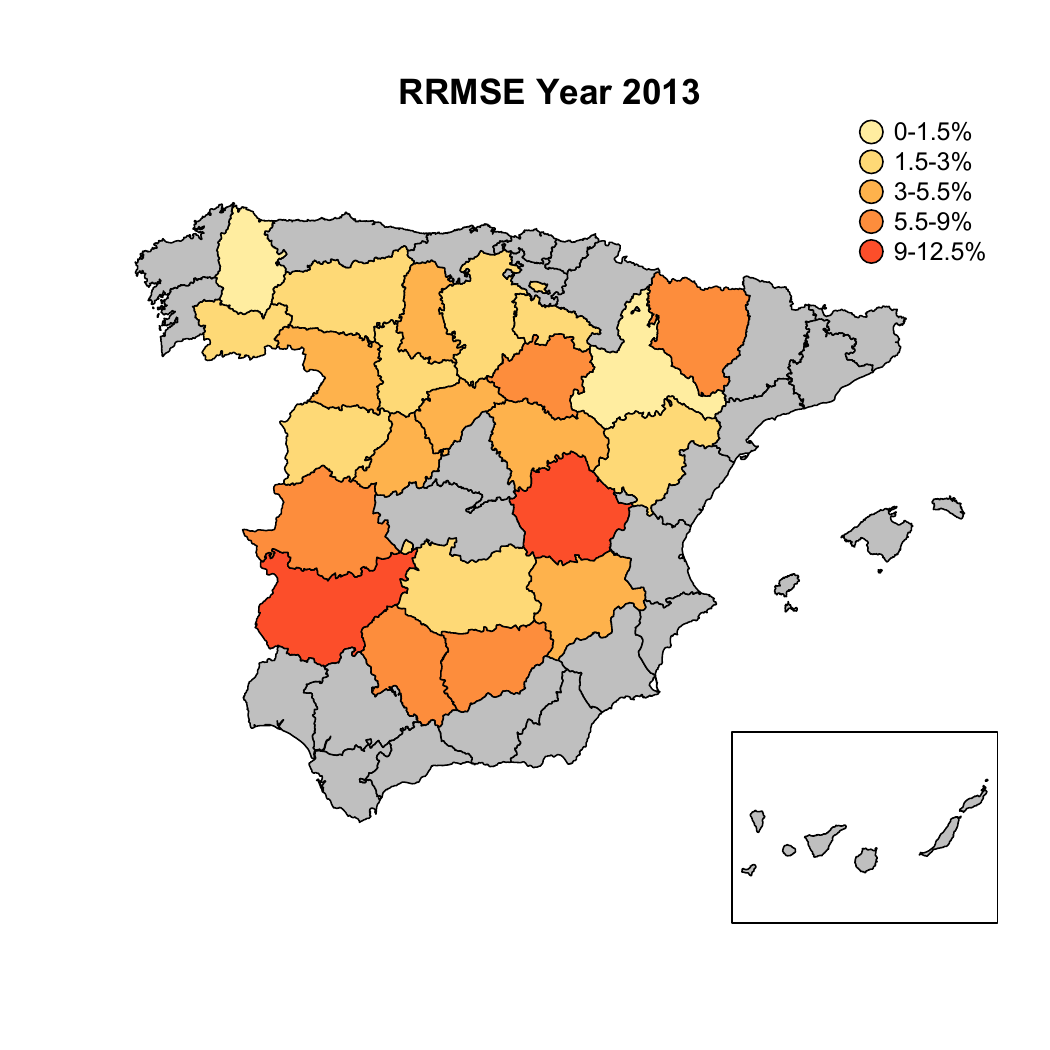}
    \hspace{-11.5mm}
    \includegraphics[width=0.39\linewidth]{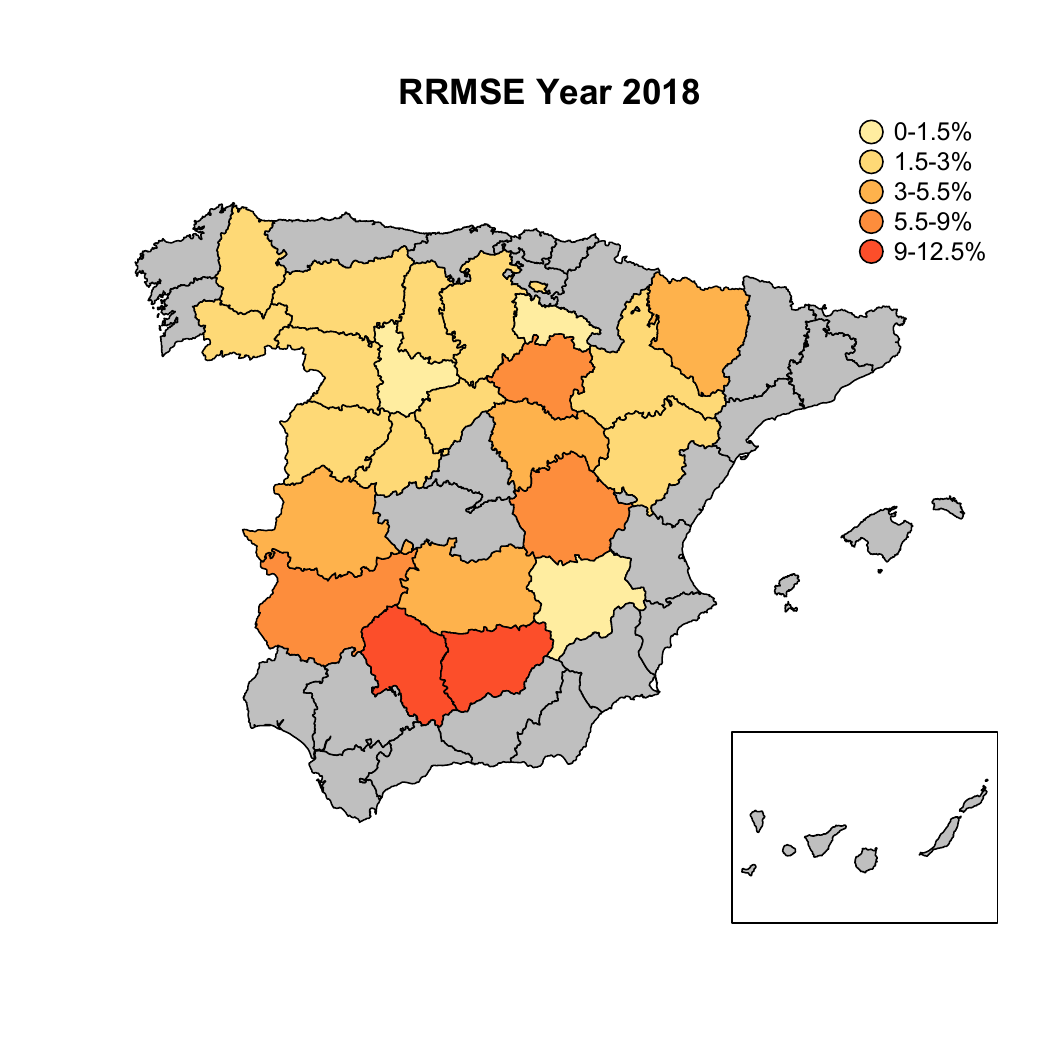}
    \hspace{-11.5mm}
    \includegraphics[width=0.39\linewidth]{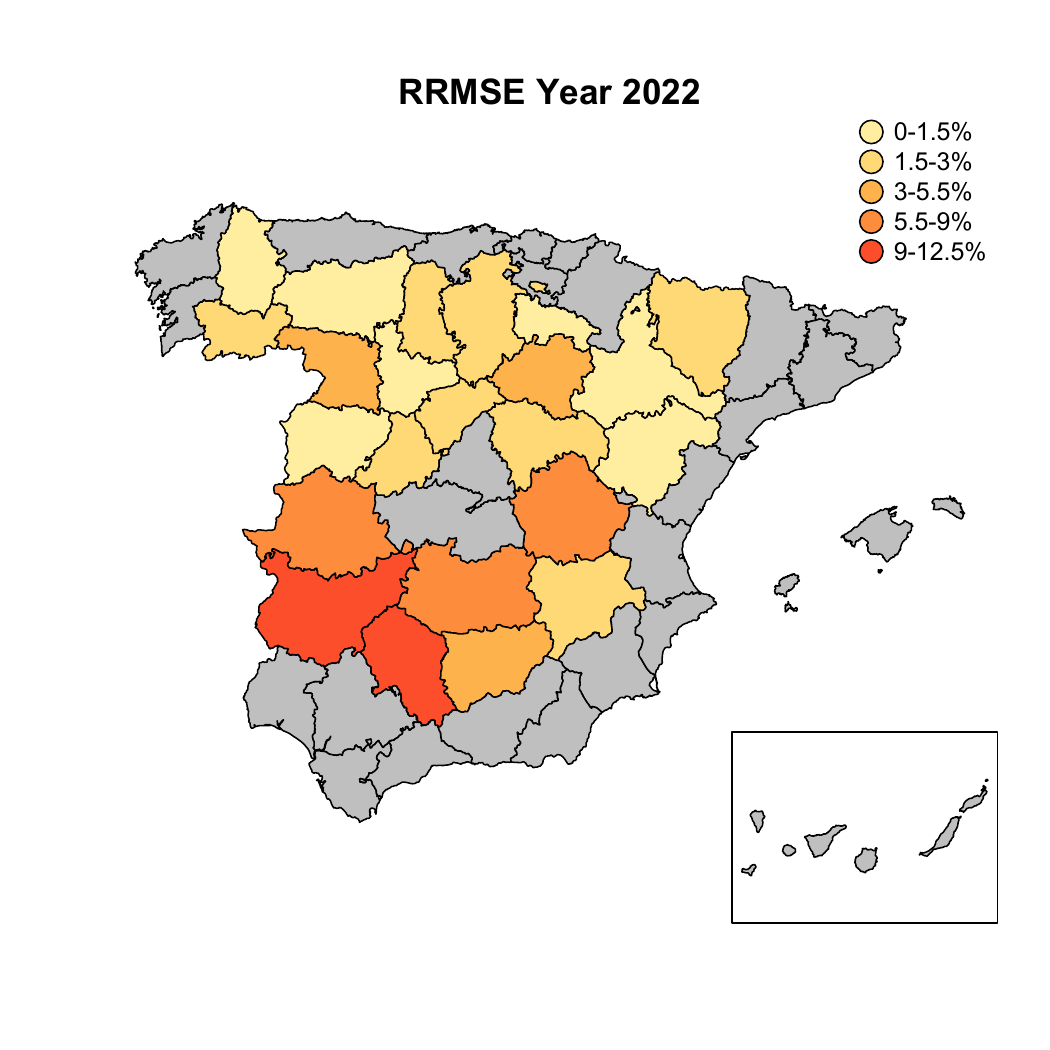}
    \hspace{-10mm}
    \caption{Coefficients of variation of the equivalized disposable income for  Empty Spain in 2013 (left), 2018 (center) and 2022 (right). Results for the BTMQ predictor.}
    \label{map.RRMSE}
\end{figure}

\subsection{Detection of outliers}\label{sec.out}
{Finally, we use the robustness parameters to detect the outlying subdomains (cf. Section \ref{sec.sim}).}
Outlier detection methods based on Linear Mixed Models (LMMs), like the one proposed by \cite{Zewotir2007}, are impractical for real-world data applications. This is due to the necessity of inverting matrices of size $n\times n$, where $n=89971$ for the provinces of Empty Spain in the SLCS2013-2022. Moreover, if we fit a LMM with random effects in provinces and years, Cook's distances do not detect deviations. This is probably due to the strong time dependencies --and also the presence of unit-level outliers-- which cloud the provincial variability of the random effects. Against this background, we propose using the set of $\big\{\hat c_{\phi, dt}: \ d=1,\dots,D, \ t=1,\dots, T\big\}$ and the same methodology as in Section \ref{sec.sim}.



    \label{table.c}


It is found that Burgos, Huesca, La Rioja, Guadalajara and Zaragoza are atypicals over time, with higher values than the average. For the rest of Empty Spain's provinces, no significant differences are detected. In socio-economic terms, our findings are reasonable. These provinces are traditionally prosperous and Guadalajara is very close to Madrid, serving as a ``commuter province'' for many workers from the capital. In short, all these regions deviate from the patterns that characterise Empty Spain.

\section{Conclusions}\label{sec.conc}

In this paper, we extended the MQ regression in SAE to time-dependent target variables and proposed predictors of small area linear quantities. We also derived a first-order approximation of the MSE and calculated several analytical estimators. To achieve this, the MQ models have been adapted to temporal data by including weights in the fitting process and defining temporal distance criteria. Subsequently, plug-in predictors have been calculated and robust bias-corrected predictors proposed. 

Furthermore, we introduced a data-driven criterion for selecting optimal robustness parameters for the influence function in the bias correction part of the predictors. Research on optimality criteria for robustness parameter selection was based on the premise that bias correction of MQ and TMQ predictors does not increase variability. While current literature often uses fixed values, we strongly recommend tailoring these parameters to the specific conditions of each subdomain, thus providing a customized solution that enhances the efficiency of MQ models. From a methodological point of view, the existence and uniqueness of area-time-specific solutions, along with the availability of a bounded search interval, support our finding. However, the study of the asymptotic distribution of the robustness parameters and the subsequent implementation of more appropriate hypothesis tests for atypical values detection is left for further research. 

Model-based simulations, have showed as the proposed optimal bias-corrected MQ predictors, derived from both temporal and non-temporal models, tend to outperform EBLUPs based on LMMs in terms of bias and MSE. However, it is worth noting that the proposed estimators are suitable only for estimating small area linear quantities, such as population means. Extending these estimators to assess non-linear small area quantities remains a subject for future research \citep{marchettinikos2016}. Additionally, the current approach is valid only for continuous outcome variables. Future work will extend the generalized MQ regression models \citep{Cha06, schirripaBJ2021} to time-dependent data to derive small area predictors for discrete response variables.  

Regarding the performance of the proposed MSE estimators, model-based simulations indicate that all proposed estimators perform well, especially in scenarios without outliers. However, in the presence of outliers, the performance of the MSE estimator derived from the first-order unbiased approximation tends to degrade slightly. Future work will focus on developing an MSE estimator based on second-order unbiasedness to address this issue.

In SAE, it is common for geographically closer areas to share similar characteristics. Therefore, developing estimators that incorporate both temporal and spatial correlations could lead to more accurate results. This topic will be addressed in future research.

\bibliographystyle{chicago} 
\bibliography{bibTMQ}


\clearpage
\setcounter{page}{1}

\setcounter{section}{0}

\renewcommand*{\thesection}{\Alph{section}}

\setcounter{equation}{0}
\renewcommand{\theequation}{\thesection.\arabic{equation}}

\small

\begin{center}
	{\LARGE\bf Supplementary Material}\\
	\noindent{\Large\bf \newline
		Temporal M-quantile models and robust bias-corrected small area predictors}\blfootnote{\footnotesize{Supported by the projects PID2022-136878NB-I00 (Spain), PROMETEO-2021-063 (Valencia, Spain), PRIN-Quantification in the Context of Dataset Shift - QuaDaSh (Grant P2022TB5JF, Italy), MAPPE, Bando a Cascata Programme “Growing Resilient, Inclusive and Sustainable (GRINS)” PE0000018.}}
	\vspace{.4cm}
	
	{Mar\'{\i}a Bugallo$^{1}$, Domingo Morales$^{1}$, Nicola Salvati$^{2}$, Francesco Schirripa Spagnolo$^{2}$\\
		\vspace{.4cm}
		{\small $^{1}$Center of Operations Research, Miguel Hern\'{a}ndez University of Elche, Spain\\
			$^{2}$Department of Economics and Management, University of Pisa, Italy}\\}
	\vspace{.4cm}
	\vspace{.4cm}
	
\end{center}
\section{First-order approximation of the MSE of BTMQ predictors}\label{sect.app1}
In this section we derive a first-order approximation of the mean squared error (MSE) of the robust bias-corrected M-quantile (BTMQ) predictor \eqref{btmq} developed from the time weighted M-quantile (TWMQ) linear models \eqref{TMQ3} and propose its analytical estimator. 

Let $d=1,\dots, D$, $t=1,\dots, T$.
Let the BTMQ predictor of $\overline{Y}_{dt}$ be
\begin{eqnarray*}
	\hat{\overline{Y}}_{dt}^{btmq}
	&=&\frac{1}{N_{dt}}
	\bigg\{\sum_{j\in s_{dt}}y_{dtj}+\sum_{j\in r_{dt}}\xx_{dtj}^\prime\hat\bbeta_\psi\big(\hat\theta_d,\ww_t\big)\bigg\}
	+\frac{1}{n_{dt}}\Big(1-\frac{n_{dt}}{N_{dt}}\Big)\hat
	B_{dt}^{btmq},
\end{eqnarray*}
where $\hat
B_{dt}^{btmq}$ has been defined in \eqref{btmq}. We now define
\begin{eqnarray*}
	\overline{Y}_{dt}^{btmq}&=&\frac{1}{N_{dt}}
	\bigg\{\sum_{j\in s_{dt}}y_{dtj}+\sum_{j\in r_{dt}}\xx_{dtj}^\prime\bbeta_\psi(\theta_d,\ww_t)\bigg\}
	+\frac{1}{n_{dt}}\Big(1-\frac{n_{dt}}{N_{dt}}\Big){B}_{dt}^{btmq},\,
	\\
	\tilde{\overline{Y}}_{dt}^{btmq}&=&\frac{1}{N_{dt}}
	\bigg\{\sum_{j\in s_{dt}}y_{dtj}+\sum_{j\in r_{dt}}\xx_{dtj}^\prime\hat\bbeta_\psi(\theta_d,\ww_t)\bigg\}
	+\frac{1}{n_{dt}}\Big(1-\frac{n_{dt}}{N_{dt}}\Big)\tilde{B}_{dt}^{btmq},
	\\
	{B}_{dt}^{btmq}&=&\sum_{j\in s_{dt}}\sigma_{\theta_dt}\phi(u_{\psi,dtj}),\quad
	\tilde{B}_{dt}^{btmq}=\sum_{j\in s_{dt}}\sigma_{\theta_dt}\phi(\tilde{u}_{\psi,dtj}),
\end{eqnarray*}
where $u_{\psi,dtj}$ and $\tilde{u}_{\psi,dtj}$, $j=1,\dots, n_{dt}$, have been obtained in
\eqref{TMQ3resid1} and \eqref{TMQ3resid}, respectively.
The approximation of $MSE(\hat{\overline{Y}}_{dt}^{btmq})$ that we propose below accounts for the randomness of the unit-level M-quantile (MQ) coefficients $q_{dtj}$, $j=1,\ldots, N_{dt}$, coming from the TWMQ models, but assumes that $\theta_d$ and $\hat\theta_d$ are known. In fact, their estimates are derived from the three-level MQ linear (MQ3) models \eqref{MQ3}.
Not least, the standard deviations $\sigma_{\theta_dt}$ are assumed to be known, but their contribution to the final error is of minor importance (Huber, 1981).

Section \ref{Assect.ass} introduces the probabilistic framework and necessary assumptions to obtain a first-order asymptotic approximation of the MSE of the BTMQ predictor.

\subsection{Assumptions}\label{Assect.ass}
This section presents a set of assumptions necessary to obtain a first-order approximation of $MSE(\hat{\overline{Y}}_{dt}^{btmq})$.
For the influence function $\phi(u)$, we assume that
\begin{itemize}\setlength\itemsep{-0.2em}
	\item[($\Phi$1)]
	$\phi$ is differentiable at $u=0$, with $\phi(0)=0$ and $\dot{\phi}(0)=1$. If $|u|\geq c_\phi$,  $\phi(u)=c_\phi\ \text{sgn}(u)$.
\end{itemize}
This assumption is quite common for influence functions in the field of robust statistics.
The non-atypical data subsets are $\G_{dt}=\big\{j\in s_{dt}: \big|u_{\psi,dtj}\big| <c_{\phi}\big\}$, $\tilde{\G}_{dt}=\Big\{j\in s_{dt}: \big|\tilde{u}_{\psi,dtj}\big| <c_{\phi}\Big\}$ and $\hat{\G}_{dt}=\Big\{j\in s_{dt}: \big|\hat{u}_{\psi,dtj}\big| <c_{\phi}\Big\}$,
and the intersection subsets are $\tilde{\mathcal H}_{dt}=\mathcal G_{dt}\cap \tilde{\mathcal G}_{dt}$, $\hat{\mathcal H}_{dt}=\tilde{\mathcal G}_{dt}\cap \hat{\mathcal G}_{dt}$
and ${\mathcal G}_{dt}\cap\tilde{\mathcal G}_{dt}\cap \hat{\mathcal G}_{dt}=\tilde{\mathcal H}_{dt}\cap \hat{\mathcal H}_{dt}$. From assumption ($\Phi$1), we can write
\begin{eqnarray*}
	{B}_{dt}^{btmq}&=&
	\sum_{j\in \G_{dt}}\sigma_{\theta_dt}\phi(u_{\psi,dij})
	+c_\phi \sum_{j\in s_{dt}-\G_{dt}}\mbox{sgn}(e_{\psi,dtj}),
	\\
	\tilde{B}_{dt}^{btmq}&=&
	\sum_{j\in \tilde{\G}_{dt}}\sigma_{\theta_dt}\phi(\tilde{u}_{\psi,dtj})
	+c_\phi \sum_{j\in s_{dt}-\tilde{\G}_{dt}}\mbox{sgn}\big(\tilde{e}_{\psi,dtj}\big),
	\\
	\hat{B}_{dt}^{btmq}&=&
	\sum_{j\in \hat{\G}_{dt}}\sigma_{\theta_dt}\phi(\hat{u}_{\psi,dtj})
	+c_\phi \sum_{j\in s_{dt}-\hat{\G}_{dt}}\mbox{sgn}\big(\hat{e}_{\psi,dtj}\big).
\end{eqnarray*}

Related to the variables $\text{sgn}(e_{\psi,dtj})$, we define the probabilities
$$
\pi_{dtj}=P\big(\text{sgn}(e_{\psi,dtj})=-1\big),\quad
1-\pi_{dtj}=P\big(\text{sgn}(e_{\psi,dtj})=1\big), \quad j=1,\dots, N_{dt},
$$
so that
$$E[\text{sgn}(e_{\psi,dtj})]=1-2\pi_{dtj}, \quad \mbox{var}(\text{sgn}(e_{\psi,dtj}))=4\pi_{dtj}-4\pi_{dtj}^2.$$
Related to the variables $\text{sgn}(\tilde{e}_{\psi,dtj})-\text{sgn}(e_{\psi,dtj})$, we define the probabilities
$$
\tilde\pi_{a,dtj}=P\big(\text{sgn}(\tilde{e}_{\psi,dtj})-\text{sgn}(e_{\psi,dtj})=a\big),\quad a=-2,0,2, \quad j=1,\dots, N_{dt},
$$
so that $E[\text{sgn}(\tilde{e}_{\psi,dtj})-\text{sgn}(e_{\psi,dtj})]=2(\tilde\pi_{2,dtj}-\tilde\pi_{-2,dtj})$ and
$$\mbox{var}(\text{sgn}(\tilde{e}_{\psi,dtj})-\text{sgn}(e_{\psi,dtj}))=4(\tilde\pi_{2,dtj}+\tilde\pi_{-2,dtj})-4(\tilde\pi_{2,dtj}-\tilde\pi_{-2,dtj})^2.$$
Related to the variables $\text{sgn}(\hat{e}_{\psi,dtj})-\text{sgn}(\tilde{e}_{\psi,dtj})$, we define the probabilities
$$
\hat\pi_{a,dtj}=P\big(\text{sgn}(\hat{e}_{\psi,dtj})-\text{sgn}(\tilde{e}_{\psi,dtj})=a\big),\quad a=-2,0,2, \quad j=1,\dots, N_{dt},
$$
so that
$E[\text{sgn}(\hat{e}_{\psi,dtj})-\text{sgn}(\tilde{e}_{\psi,dtj})]=2(\hat\pi_{2,dtj}-\hat\pi_{-2,dtj})$ and
$$\mbox{var}(\text{sgn}(\hat{e}_{\psi,dtj})-\text{sgn}(\tilde{e}_{\psi,dtj}))=4(\hat\pi_{2,dtj}+\hat\pi_{-2,dtj})-4(\hat\pi_{2,dtj}-\hat\pi_{-2,dtj})^2.$$

The asymptotic theory will be developed under the following assumptions.\\
For the sample sizes, we assume
\begin{itemize}\setlength\itemsep{-0.2em}
	\item[(N1)]
	There exist $0<\pi_{dt}<1$ such that $\sum_{d=1}^D\sum_{t=1}^T\pi_{dt}=1$ and $\frac{n_{dt}}{n}\to\pi_{dt}$ as $n\to\infty$.
	\item[(N2)]
	There exist $0<f_{dt}<1$ such that $\frac{n_{dt}}{N_{dt}}\to f_{dt}$ as $n\to\infty$.
\end{itemize}
The asymptotic assumption (N1) avoid the possibility of domains with zero sample size.
Assumption (N2)  states that sample sizes and population sizes converge to the sample fractions reasonably far from the extremes values 0 and 1.

For the unit-level MQ coefficients, we assume
\begin{itemize}\setlength\itemsep{-0.2em}
	\item[(Q1)] For $i\in \T_t$, $j\in s_{di}$, $q_{dij}$ are independent variables with common variance $\xi_{dt}^2=\mbox{var}(q_{dij})$.
\end{itemize}
The unit-level MQ coefficients $q_{dij}$, defined in (\ref{qdtjwt}), play a similar role to that of random intercepts in linear mixed models (LMM).
In such models, it is common to assume that the random effects have constant variance.
Because of the parallelism that can be established between methodologies based on MQ models and LMMs, we include assumption (Q1).

For the estimator of the vector of regression parameters $\hat\bbeta_\psi(q,\ww_t)$, $0<q<1$, we adapt the conditions proposed in Bianchi and Salvati (2015) and write
\begin{itemize}
	\item[(A1)] $\bbeta_\psi(q,\ww_t)\in \Theta\subset \mathbb R^p$ is a twice-differentiable continuous function in its first component, $q$, where $\Theta$ is a compact subset of $\mathbb R^p$.
	\item[(A2)]
	$\psi$ is continuous, bounded and with bounded derivative, except at a finite number of points.
	\item[(A3)]
	For $i\in \T_t$, $j\in s_{di}$, $E[|\xx_{dij}|^4]<\infty$ and $E[|e_{dij}|^4]<\infty$.
	\item[(A4)]
	For $i\in \T_t$, $j\in s_{di}$, $E[\xx_{dij}\xx_{dij}'\dot{\psi}(e_{\psi, dij}(q,\ww_t),\sigma_{qt})]$ is uniformly non singular, where $\dot{\psi}_{qt}$ is the partial derivative of $\psi_{qt}$ with respect to the first argument.
	\item[(A5)]
	The preliminary estimator $\hat\bbeta_\psi^{(0)}(q,\ww_t)$ of $\bbeta_\psi(q,\ww_t)$
	is such that $$\sqrt{n}(\hat\bbeta_\psi^{(0)}(q,\ww_t)-\bbeta_\psi(q,\ww_t))=O_p(1).$$
	\item[(A6)]
	$\exists \delta_1>0$: $\forall e\in (-\delta_1, \delta_1)\ \exists \dot F_{qt}(\mbox{med}_{\psi,n}(q,\ww_t)+e)$ and is continuous and positive at $e=0$. $\dot F_{qt}$ is the first order derivative of $F_{qt}$, which is the cumulative distribution function (c.d.f). of $e_{\psi,dij}(q,\ww_t)$, $i\in \T_t$, $j\in s_{di}$.
	\item[(A7)]
	$\exists \delta_2>0$: $\forall e\in (-\delta_2, \delta_2)\ \exists  \dot F_{qt}(\mbox{med}_{\psi,n}(q,\ww_t)\pm0.6745\sigma_{qt}+e)$ and is continuous at $e=0$.
	\item[(A8)]
	$\exists \delta_3>0$: $\forall e\in (\sigma_{qt}-\delta_3, \sigma_{qt}+\delta_3)$ $\exists\dot F_{qt}(\mbox{med}_{\psi,n}(q,\ww_t)+e)+\dot F_{qt}(\mbox{med}_{\psi,n}(q,\ww_t)-e)>0$.
\end{itemize}
Assumption (A3) is a technical moment condition required for the application of the Uniform Law of Large Numbers and the asymptotic representation.
Assumption (A4) is an identifiability condition.
Assumptions (A5)-(A8) are needed for the Bahadur representation of the median absolute deviation (MAD) estimator (see Welsh, 1986).
In the case of the Huber influence function (\ref{psiHuber}), assumptions (A1) and (A2) are satisfied.
To guarantee assumption (A4), one may require that for any $\bbeta_\psi(q,\ww_t)\in \Theta$ and $c>0$,
\[P(\sigma_{qt}^{-1}|y_{dij}-\xx_{dij}^\prime\bbeta_\psi(q,\ww_t)|\leq c |\,\xx_{dij})>\varepsilon>0,  \ i\in \T_t, \ j\in s_{di}.\]
In practice, this is verified if most of the residuals belong to the strictly convex region of $\psi$.
Under assumptions (A1)-(A8), it holds that $\hat{\bbeta}_\psi(\theta_d,\ww_t)-\bbeta_\psi(\theta_d,\ww_t)=O_p(n^{-1/2})$. We include
\begin{itemize}
	\item[(A9)]
	$\exists\delta_4>0$: $\forall\theta\in (\theta_d-\delta_4, \theta_d+\delta_4)$,
	$\hat{\bbeta}_\psi(\theta,\ww_t)-\bbeta_\psi(\theta_d,\ww_t)=O_p(n^{-1/2})$.
\end{itemize}
Assumption (A9) is needed to maintain the asymptotic plausibility of assumption (A8) in a neighborhood of $\theta_d$, as in practice
$\theta_d$ is substituted by $\hat\theta_d$, obtained from models (\ref{MQ3}).

For the MSE of the BTMQ predictor, we assume three groups of assumptions.
The first group concerns the model errors (\ref{TMQ3resid1}) of the TWMQ models (\ref{TMQ3}).
The assumptions are
\begin{itemize}\setlength\itemsep{-0.2em}
	\item[(B1)]
	$\frac{1}{N_{dt}}\sum_{j\in U_{dt}}E\big[e_{\psi,dtj}^2\big]=O(1)$.
	\item[(B2)]
	$\frac{1}{N_{dt}}\sum_{j\in U_{dt}}\pi_{dtj}=O(1)$ and $\frac{1}{N_{dt}}\sum_{j\in U_{dt}}\pi_{dtj}^2=O(1)$.
	\item[(B3)] $ \frac{1}{n_{dt}}\sum_{j\in s_{dt}}E\big[e_{\psi,dtj}\big]
	=\frac{1}{N_{dt}-n_{dt}} \sum_{j\in r_{dt}}E\big[e_{\psi,dtj}\big]+o(1)$.
	\item[(B4)]
	$\exists\delta_5>0$: $\forall\theta\in (\theta_d-\delta_5, \theta_d+\delta_5)$,
	$\frac{1}{N_{dt}}\sum_{j\in U_{dt}}\Big( \xx_{dtj}^\prime\frac{\partial^2\bbeta_\psi(q,\ww_t)}{\partial q^2}\Big|_{q=\theta}\Big)^2=O(1)$.
\end{itemize}
Assumption (B1) states that the second order moment of the model errors is bounded in average.
Assumption (B2) is fulfilled as $\pi_{dtj}$'s are the probabilities that the model errors are negative.
Assumption (B3) states that the sample average of the expected model errors behaves similarly in the sample and the non-sample subsets.
Assumption (B4) is a technical condition required for the application of the Uniform Law of Large Numbers and the asymptotic representation.

The second group concerns the pseudo-residuals (\ref{TMQ3resid}) of the TWMQ models (\ref{TMQ3}).
The assumptions are
\begin{itemize}\setlength\itemsep{-0.2em}
	\item[(C1)]
	$\frac{1}{N_{dt}}\sum_{j\in U_{dt}}E\big[(\tilde{e}_{\psi,dtj}-e_{\psi,dtj})^4\big]=o(1)$.
	\item[(C2)]
	$\frac{1}{N_{dt}}\sum_{j\in U_{dt}}\tilde{\pi}_{a,dtj}=o(1)$ and $\frac{1}{N_{dt}}\sum_{j\in U_{dt}}\tilde{\pi}^2_{a,dtj}=o(1)$, $a=-2,2$.
	\item[(D1)]
	$\frac{1}{N_{dt}}\sum_{j\in U_{dt}}E\big[(\hat{e}_{\psi,dtj}-\tilde{e}_{\psi,dtj})^4\big]=o(1)$.
	\item[(D2)]
	$\frac{1}{N_{dt}}\sum_{j\in U_{dt}}\hat{\pi}_{a,dtj}^{b}=o(1)$ and $\frac{1}{N_{dt}}\sum_{j\in U_{dt}}\hat{\pi}^2_{a,dtj}=o(1)$, $a=-2,2$.
\end{itemize}
Assumptions (C1) and (D1) state that the fourth order moment of the differences between pseudo-residuals and pseudo-residuals and model errors is bounded in average.
Assumptions (C2) and (D2) are fulfilled as $\tilde{\pi}_{a,dtj}$'s and $\hat{\pi}_{a,dtj}$'s  are the probabilities that the pseudo-residuals are negative. The third group concerns the independence of the model errors and pseudo-residuals of the TWMQ models \eqref{TMQ3}.
The assumptions are
\begin{itemize}\setlength\itemsep{-0.2em}
	\item[(E1)] For $j\in U_{dt}$, $e_{\psi,dtj}$ are independent random variables.
	\item[(E2)] For $j\in U_{dt}$, $\tilde e_{\psi,dtj}$ are independent random variables.
	\item[(E3)] For $j\in U_{dt}$, $\hat e_{\psi,dtj}$ are independent random variables.
\end{itemize}
\subsection{Part I: Dealing with the differences {\boldmath $\overline{Y}_{dt}^{btmq}-\overline{Y}_{dt}$}}\label{appex.P1}
A Taylor series expansion of $\bbeta_\psi(q_{dtj},\ww_t)$ around $\theta_d$ yields to
\begin{equation}\label{T0}
	\bbeta_\psi(q_{dtj},\ww_t)=\bbeta_\psi(\theta_d,\ww_t)
	+\frac{\partial\bbeta_\psi(q,\ww_t)}{\partial q}\Big|_{q=\theta_d}(q_{dtj}-\theta_d)+\bm r_{\psi,dtj}(\theta_d), \ \ j=1,\ldots,N_{dt},
\end{equation}
where $\bm r_{\psi,dtj}(\theta_d)\triangleq\bm r_{dtj}=(r_{dtj1},\dots, r_{dtjp})^\prime$ and $r_{dtjk}=O_p\big((q_{dtj}-\theta_d)^2\big)$, $k=1,\ldots,p$,
with
$$
\Vert \bm r_{dtj}\Vert_{2}=\frac12(q_{dtj}-\theta_d)^2\Big\Vert\frac{\partial^2\bbeta_\psi(q,\ww_t)}{\partial q^2}\Big|_{q=\theta_{dtj}^*}\Big\Vert_{2}\leq
\Big\Vert\frac{\partial^2\bbeta_\psi(q,\ww_t)}{\partial q^2}\Big|_{q=\theta_{dtj}^*}\Big\Vert_{2},\ \
|\theta_{dtj}^*-\theta_d|<|q_{dtj}-\theta_d|.
$$
We introduce the notation
$$
\bm\kappa_\psi(\theta_d,\ww_t)=\big(\kappa_{\psi 1}(\theta_d,\ww_t),\ldots,\kappa_{\psi p}(\theta_d,\ww_t)\big)^\prime;\quad
\kappa_{\psi k}(\theta_d,\ww_t)=\frac{\partial\beta_{\psi k}(q,\ww_t)}{\partial q}\Big|_{q=\theta_d},\,\, k=1,\ldots,p.
$$
From the Taylor series expansion (\ref{T0}), the model errors defined in \eqref{TMQ3resid1} can be written as
\begin{eqnarray*}\label{edtjTaylor}
	e_{\psi,dtj}&=&y_{dtj}-\xx_{dtj}^\prime\bbeta_\psi(\theta_d,\ww_t)
	=\xx_{dtj}^\prime(\bbeta_\psi(q_{dtj},\ww_t)-\bbeta_\psi(\theta_d,\ww_t))
	\nonumber
	\\
	&=&\xx_{dtj}^\prime\bm\kappa_\psi(\theta_{d},\ww_t)(q_{dtj}-\theta_d)+\xx_{dtj}^\prime \bm r_{dtj}, \quad j=1,\ldots,N_{dt}.
\end{eqnarray*}
From assumption (Q1),  $\mbox{var}(q_{dtj}-\theta_d)=\mbox{var}(q_{dtj})=\xi_{dt}^2$, $j=1,\dots, N_{dt}$, and
\begin{equation}\label{varedtjTaylor}
	\mbox{var}(e_{\psi,dtj})=\big(\xx_{dtj}^\prime\bm\kappa_\psi(\theta_{d},\ww_t)\big)^2\xi_{dt}^2+\xx_{dtj}^\prime\mbox{var}(\bm r_{dtj})\xx_{dtj}
	+\big(\xx_{dtj}^\prime\bm\kappa_\psi(\theta_{d},\ww_t)\big)^2\text{cov}\big(q_{dtj},q_{dtj}^2\big).
\end{equation}
Assumption (B4) implies that
\begin{eqnarray*}\label{HPI.1}
	\frac{1}{n_{dt}}\sum_{j\in s_{dt}}\xx_{dtj}^\prime\mbox{var}(\bm r_{dtj})\xx_{dtj}&=&
	\frac{1}{4n_{dt}}\sum_{j\in s_{dt}}\Big\{ \xx_{dtj}^\prime\bm\kappa_\psi(\theta_{dtj}^*,\ww_t)
	\bm\kappa_\psi^\prime(\theta_{dtj}^*,\ww_t) \xx_{dtj}\Big\}\,\text{var}(q_{dtj}^2)
	\nonumber
	\\
	&\leq&\frac{1}{4n_{dt}}\sum_{j\in s_{dt}}\big(\xx_{dtj}^\prime\bm\kappa_\psi(\theta_{dtj}^*,\ww_t)\big)^2=O(1),
\end{eqnarray*}
and
\begin{eqnarray*}
	\Big|\frac{1}{n_{dt}}\sum_{j\in s_{dt}}\big(\xx_{dtj}^\prime\bm\kappa_\psi(\theta_{dtj}^*,\ww_t)\big)^2E[q_{dtj}^3]\Big|&\leq&
	\frac{1}{n_{dt}}\sum_{j\in s_{dt}}\big(\xx_{dtj}^\prime\bm\kappa_\psi(\theta_{dtj}^*,\ww_t)\big)^2\big|E[q_{dtj}]\big|^3=O(1),
	\\
	\Big|\frac{1}{n_{dt}}\sum_{j\in s_{dt}}\big(\xx_{dtj}^\prime\bm\kappa_\psi(\theta_{dtj}^*,\ww_t)\big)^2E[q_{dtj}]\Big|&\leq&
	\frac{1}{n_{dt}}\sum_{j\in s_{dt}}\big(\xx_{dtj}^\prime\bm\kappa_\psi(\theta_{dtj}^*,\ww_t)\big)^2\big|E[q_{dtj}]\big|=O(1).
\end{eqnarray*}
Collecting these results, we obtain that
\begin{equation}\label{meanvaredtj}
	\frac{1}{n_{dt}}\sum_{j\in s_{dt}}\mbox{var}(e_{\psi,dtj})=\frac{1}{n_{dt}}\sum_{j\in s_{dt}}\big(\xx_{dtj}^\prime\bm\kappa_\psi(\theta_{d},\ww_t)\big)^2\xi_{dt}^2+O(1).
\end{equation}
The corresponding standardized model errors, defined in \eqref{TMQ3resid1}, can be written as
\begin{eqnarray*}
	u_{\psi,dtj}=\sigma_{\theta_dt}^{-1}\,\xx_{dtj}^\prime\bm\kappa_\psi(\theta_d,\ww_t)(q_{dtj}-\theta_d)
	+\sigma_{\theta_dt}^{-1}\,\xx_{dtj}^\prime \bm r_{dtj}, \quad j=1,\dots, N_{dt}.
\end{eqnarray*}
As $\phi(0)=0$ and $\dot{\phi}(0)=1$, a Taylor series expansion of $\phi(u_{\psi,dtj})$ around $u=0$ yields to
\begin{equation}\label{T1}
	\phi(u_{\psi,dtj})=\phi(0)+\dot{\phi}(0)u_{\psi,dtj}+R_{dtj}=u_{\psi,dtj}+R_{dtj}, \quad j=1,\dots, N_{dt},
\end{equation}
where $R_{dtj}=\frac12\,u_{\psi,dtj}^{*2}$, $0<|u_{\psi,dtj}^{*}|<|u_{\psi,dtj}|$ and
$\mbox{var}(R_{dtj})\leq\mbox{var}(u_{\psi,dtj})$.

From assumption (B1), we have
\begin{equation}\label{HPI.2}
	\frac{1}{n_{dt}}\sum_{j\in\G_{dt}}E[R_{dtj}]=O(1),\quad
	\frac{1}{n_{dt}}\sum_{j\in\G_{dt}}\mbox{var}(R_{dtj})=O(1).
\end{equation}
From the Taylor series expansion (\ref{T1}), we have
\begin{eqnarray*}
	B_{dt}^{btmq}&=&
	\sum_{j\in \G_{dt}}\sigma_{\theta_dt}\phi(u_{\psi,dtj})+c_\phi \sum_{j\in s_{dt}-\G_{dt}}\mbox{sgn}\big(e_{\psi,dtj}\big)\\
	&=&\sum_{j\in \G_{dt}}e_{\psi,dtj}+\sigma_{\theta_dt}\sum_{j\in \G_{dt}} R_{dtj}+c_\phi \sum_{j\in s_{dt}-\G_{dt}}\mbox{sgn}\big(e_{\psi,dtj}\big).
\end{eqnarray*}
The prediction difference $\overline{Y}_{dt}^{(1)}=\overline{Y}_{dt}^{btmq}-\overline{Y}_{dt}$ is
\begin{eqnarray}\label{barYdt1}
	\nonumber
	\overline{Y}_{dt}^{(1)}&=&\frac{1}{N_{dt}}\bigg(\sum_{j\in s_{dt}}y_{dtj}+\sum_{j\in r_{dt}} \xx_{dtj}^\prime\bbeta_{\psi}(\theta_d,\ww_t)\bigg)
	+\Big(1-\frac{n_{dt}}{N_{dt}}\Big)\frac{1}{n_{dt}}B_{dt}^{btmq}\\
	&-&\frac{1}{N_{dt}}\underset{j\in U_{dt}}{\sum}\xx_{dtj}^\prime\bbeta_{\psi}(q_{dtj},\ww_t)
	\nonumber
	\\
	&=&\Big(1-\frac{n_{dt}}{N_{dt}}\Big)\frac{1}{n_{dt}}B_{dt}^{btmq}
	+\frac{1}{N_{dt}}\bigg(\sum_{j\in r_{dt}}\xx_{dtj}^\prime(\bbeta_{\psi}(\theta_d,\ww_t)-\bbeta_{\psi}(q_{dtj},\ww_t))\bigg)
	\nonumber
	\\
	&=&\sum_{j\in U_{dt}}\left(\Big(1-\frac{n_{dt}}{N_{dt}}\Big)\frac{1}{n_{dt}}I_{\G_{dt}}(j)-\frac{1}{N_{dt}}I_{r_{dt}}(j)\right) e_{\psi,dtj}
	\nonumber
	\\
	&+&\Big(1-\frac{n_{dt}}{N_{dt}}\Big)\frac{c_\phi}{n_{dt}}\sum_{j\in s_{dt}-\G_{dt}}\,\text{sgn}(e_{\psi,dtj})
	+\Big(1-\frac{n_{dt}}{N_{dt}}\Big)\frac{\sigma_{\theta_dt}}{n_{dt}}\sum_{j\in\G_{dt}}R_{dtj}.
\end{eqnarray}
From assumptions (E1) and (B2), we obtain that
\begin{eqnarray}\label{HPI.3}
	\nonumber
	\frac{1}{n_{dt}}\sum_{j\in s_{dt}-\G_{dt}}E[\text{sgn}(e_{\psi,dtj})]&=&\frac{1}{n_{dt}}\Big( \text{card}(s_{dt}-\G_{dt})-2\sum_{j\in s_{dt}-\G_{dt}}\pi_{dtj}\Big) =O(1),
	\\
	\frac{1}{n_{dt}}\sum_{j\in s_{dt}-\G_{dt}}\text{var}(\text{sgn}(e_{\psi,dtj}))&=&\frac{4}{n_{dt}}\sum_{j\in s_{dt}-\G_{dt}}\pi_{dtj}(1-\pi_{dtj})=O(1).
\end{eqnarray}
From (\ref{meanvaredtj}), (\ref{barYdt1}), (\ref{HPI.2}), (\ref{HPI.3}) and assumptions (B4) and (E1), the variance of $\overline{Y}_{dt}^{(1)}$ is
\begin{eqnarray*}
	V_{dt}^{(1)}&=&\text{var}\big(\overline{Y}_{dt}^{(1)}\big)=\sum_{j\in U_{dt}}
	\Big(\Big(1-\frac{n_{dt}}{N_{dt}}\Big)\frac{1}{n_{dt}}I_{\G_{dt}}(j)-\frac{1}{N_{dt}}I_{r_{dt}}(j)\Big)^2 \text{var}(e_{\psi,dtj})
	\\
	&+&\Big(1-\frac{n_{dt}}{N_{dt}}\Big)^2\frac{c_\phi^2}{n_{dt}^2}\sum_{j\in s_{dt}-\G_{dt}}\text{var}(\text{sgn}(e_{\psi,dtj}))
	+\Big(1-\frac{n_{dt}}{N_{dt}}\Big)^2\frac{\sigma_{\theta_dt}^2}{n_{dt}^2}\sum_{j\in\G_{dt}}\mbox{var}(R_{dtj})
	\\
	&=&\sum_{j\in U_{dt}}\Big(\Big(1-\frac{n_{dt}}{N_{dt}}\Big)^2\frac{1}{n_{dt}^2}I_{\G_{dt}}(j)+ \frac{1}{N_{dt}^2}I_{r_{dt}}(j)\Big)
	\big(\xx_{dtj}^\prime\bm\kappa_\psi(\theta_{d},\ww_t)\big)^2\xi_{dt}^2+o(n^{-1}).
\end{eqnarray*}
From (\ref{barYdt1}), (\ref{HPI.2}), (\ref{HPI.3}) and assumptions (B2) and (B3), we have that
\begin{eqnarray*}
	E_{dt}^{(1)}&=&E\big[\overline{Y}_{dt}^{(1)}\big]=\sum_{j\in U_{dt}}\Big(
	\Big(1-\frac{n_{dt}}{N_{dt}}\Big)\frac{1}{n_{dt}}I_{\G_{dt}}(j)- \frac{1}{N_{dt}}I_{r_{dt}}(j)\Big) E\big[e_{\psi,dtj}\big]
	\\
	&+&
	\Big(1-\frac{n_{dt}}{N_{dt}}\Big)\frac{c_\phi}{n_{dt}}\sum_{j\in s_{dt}-\G_{dt}}E\big[\text{sgn}(e_{\psi,dtj})\big]
	+\Big(1-\frac{n_{dt}}{N_{dt}}\Big)\frac{\sigma_{\theta_dt}}{n_{dt}}\sum_{j\in\G_{dt}}E[R_{dtj}]=\\&=&
	\Big(1-\frac{n_{dt}}{N_{dt}}\Big)\frac{c_\phi}{n_{dt}}\sum_{j\in s_{dt}-\G_{dt}}E\big[\text{sgn}(e_{\psi,dtj})\big]
	+\Big(1-\frac{n_{dt}}{N_{dt}}\Big)\frac{\sigma_{\theta_dt}}{n_{dt}}\sum_{j\in\G_{dt}}E[R_{dtj}]+o(1).
\end{eqnarray*}
From (\ref{HPI.2}) and (\ref{HPI.3}), it holds that $E_{dt}^{(1)}=O(1)$.
\subsection{Part II: Dealing with the differences {\boldmath $\tilde{\overline{Y}}_{dt}^{btmq}-\overline{Y}_{dt}^{btmq}$}}\label{appex.P2}
As $\phi(0)=0$ and $\dot{\phi}(0)=1$, a Taylor series expansion of $\phi(\tilde{u}_{\psi,dtj})$ around $u=0$ yields to
\begin{equation}\label{PIIT1}
	\phi(\tilde{u}_{\psi,dtj})=\phi(0)+\dot{\phi}(0)\tilde{u}_{\psi,dtj}+\tilde{R}_{dtj}=\tilde{u}_{\psi,dtj}+\tilde{R}_{dtj}, \quad j=1,\dots, N_{dt},
\end{equation}
where $\tilde{R}_{dtj}=\frac12\,\tilde{u}_{\psi,dtj}^{*2}$, $0<|\tilde{u}_{\psi,dtj}^{*}|<|\tilde{u}_{\psi,dtj}|$ and
$\mbox{var}(\tilde{R}_{dtj})\leq\mbox{var}(\tilde{u}_{\psi,dtj})$.

From assumption (C1), we have
\begin{equation}\label{HPII.2}
	\frac{1}{n_{dt}}\sum_{j\in\tilde\G_{dt}}E[\tilde{R}_{dtj}-{R}_{dtj}]=o(1),\quad
	\frac{1}{n_{dt}}\sum_{j\in\tilde\G_{dt}}\mbox{var}\big((\tilde{R}_{dtj}-{R}_{dtj})^2\big)=o(1).
\end{equation}
From the Taylor series expansion (\ref{PIIT1}), we have
\begin{eqnarray*}
	\tilde{B}_{dt}^{btmq}&=&
	\sum_{j\in \tilde \G_{dt}}\sigma_{\theta_dt}
	\phi\big(\tilde{u}_{\psi,dtj}\big)+c_\phi \sum_{j\in s_{dt}-\tilde\G_{dt}}\mbox{sgn}\big(\tilde{e}_{\psi,dtj}\big),
	\\
	&=&
	\sum_{j\in\tilde\G_{dt}}\tilde{e}_{\psi,dtj}+\sigma_{\theta_dt}\sum_{j\in\tilde\G_{dt}}\tilde{R}_{dtj}+c_\phi\sum_{j\in s_{dt}-\tilde\G_{dt}}\text{sgn}(\tilde{e}_{\psi,dtj}).
\end{eqnarray*}
As $\tilde{e}_{\psi,dtj}-e_{\psi,dtj}
=\xx_{dtj}^\prime\big(\bbeta_\psi(\theta_d,\ww_t)-\hat\bbeta_\psi(\theta_d,\ww_t)\big)=\tilde{e}_{\psi,dtj}(\theta_d)$,
$B_{dt}^{(2)}=\tilde{B}_{dt}^{btmq}-B_{dt}^{btmq}$ is
\begin{eqnarray*}
	B_{dt}^{(2)}
	&=&\sum_{j\in\tilde\G_{dt}}\tilde{e}_{\psi,dtj}+\sigma_{\theta_dt}\sum_{j\in\tilde\G_{dt}}\tilde{R}_{dtj}
	+c_\phi\sum_{j\in s_{dt}-\tilde\G_{dt}}\text{sgn}(\tilde{e}_{\psi,dtj})
	\\
	&-&\sum_{j\in\G_{dt}}e_{\psi,dtj}-\sigma_{\theta_dt}\sum_{j\in\G_{dt}}R_{dtj}-c_\phi\sum_{j\in s_{dt}-\G_{dt}}\text{sgn}(e_{\psi,dtj})
	\\
	&=&\sum_{j\in\tilde{\mathcal H}_{dt}}\tilde{e}_{\psi,dtj}(\theta_d)
	+\sum_{j\in\tilde\G_{dt}-\tilde{\mathcal H}_{dt}}\tilde{e}_{\psi,dtj}-\sum_{j\in\G_{dt}-\tilde{\mathcal H}_{dt}}e_{\psi,dtj}
	\\
	&+&\sigma_{\theta_dt}\Big( \sum_{j\in\tilde{\mathcal H}_{dt}}(\tilde R_{dtj}-R_{dtj})
	+\sum_{j\in\tilde\G_{dt}-\tilde{\mathcal H}_{dt}}\tilde R_{dtj}-\sum_{j\in\G_{dt}-\tilde{\mathcal H}_{dt}}R_{dtj}\Big)
	\\
	&+&c_\phi\Big( \sum_{j\in s_{dt}-\tilde{\mathcal H}_{dt}}(\text{sgn}(\tilde{e}_{\psi,dtj})-\text{sgn}({e}_{\psi,dtj}))
	\\
	&+&\sum_{j\in s_{dt}-(\tilde\G_{dt}-\tilde{\mathcal H}_{dt})}\text{sgn}(\tilde{e}_{\psi,dtj})-\sum_{j\in s_{dt}-(\G_{dt}-\tilde{\mathcal H}_{dt})}\text{sgn}({e}_{\psi,dtj})\Big).
\end{eqnarray*}
The prediction difference  $\overline{Y}_{dt}^{(2)}=\tilde{\overline{Y}}_{dt}^{btmq}-\overline{Y}_{dt}^{btmq}$ is
\begin{eqnarray*}
	\overline{Y}_{dt}^{(2)}&=&\frac{1}{N_{dt}}\bigg(\sum_{j\in s_{dt}}y_{dtj}+\sum_{j\in r_{dt}} \xx_{dtj}^\prime\hat\bbeta_{\psi}(\theta_d,\ww_t)\bigg)
	+\Big(1-\frac{n_{dt}}{N_{dt}}\Big)\frac{1}{n_{dt}}\tilde{B}_{dt}^{btmq}
	\\
	&-&\frac{1}{N_{dt}}\bigg(\sum_{j\in s_{dt}}y_{dtj}+\sum_{j\in r_{dt}} \xx_{dtj}^\prime\bbeta_{\psi}(\theta_d,\ww_t)\bigg)
	-\Big(1-\frac{n_{dt}}{N_{dt}}\Big)\frac{1}{n_{dt}}B_{dt}^{btmq}
	\\
	&=&\Big(1-\frac{n_{dt}}{N_{dt}}\Big)\frac{1}{n_{dt}}B_{dt}^{(2)}
	-\frac{1}{N_{dt}}\sum_{j\in r_{dt}}\tilde{e}_{\psi,dtj}(\theta_d).
\end{eqnarray*}
By substituting $B_{dt}^{(2)}$, we obtain
\begin{eqnarray*}
	\overline{Y}_{dt}^{(2)}&=&\Big(1-\frac{n_{dt}}{N_{dt}}\Big)\frac{1}{n_{dt}}
	\Bigg\{
	\sum_{j\in\tilde{\mathcal H}_{dt}}\tilde{e}_{\psi,dtj}(\theta_d)
	+\sum_{j\in\tilde\G_{dt}-\tilde{\mathcal H}_{dt}}\tilde{e}_{\psi,dtj}-\sum_{j\in\G_{dt}-\tilde{\mathcal H}_{dt}}e_{\psi,dtj}
	\\
	&+&\sigma_{\theta_dt}\Big( \sum_{j\in\tilde{\mathcal H}_{dt}}(\tilde R_{dtj}-R_{dtj})
	+\sum_{j\in\tilde\G_{dt}-\tilde{\mathcal H}_{dt}}\tilde R_{dtj}-\sum_{j\in\G_{dt}-\tilde{\mathcal H}_{dt}}R_{dtj}\Big)
	\\
	&+&c_\phi\Big( \sum_{j\in s_{dt}-\tilde{\mathcal H}_{dt}}(\text{sgn}(\tilde{e}_{\psi,dtj})-\text{sgn}({e}_{\psi,dtj}))
	+\sum_{j\in s_{dt}-(\tilde\G_{dt}-\tilde{\mathcal H}_{dt})}\text{sgn}(\tilde{e}_{\psi,dtj})
	\\
	&-&\sum_{j\in s_{dt}-(\G_{dt}-\tilde{\mathcal H}_{dt})}\text{sgn}({e}_{\psi,dtj})\Big)
	\Big\}
	-\frac{1}{N_{dt}}\sum_{j\in r_{dt}}\tilde{e}_{\psi,dtj}(\theta_d).
\end{eqnarray*}
We can $\overline{Y}_{dt}^{(2)}$ in the form
\begin{eqnarray*}
	&&\overline{Y}_{dt}^{(2)}=\sum_{j\in U_{dt}}
	\Big(\Big(1-\frac{n_{dt}}{N_{dt}}\Big)\frac{1}{n_{dt}}I_{\tilde{\mathcal H}_{dt}}(j)-\frac{1}{N_{dt}}I_{r_{dt}}(j)\Big)\tilde{e}_{\psi,dtj}(\theta_d)
	\\
	&+&
	\Big(1-\frac{n_{dt}}{N_{dt}}\Big)\frac{1}{n_{dt}}\bigg\{
	\sigma_{\theta_dt}\sum_{j\in\tilde{\mathcal H}_{dt}}\big\{\tilde{R}_{dtj}-R_{dtj}\big\}
	+c_\phi\sum_{j\in s_{dt}-\tilde{\mathcal H}_{dt}}\big(\text{sgn}(\tilde{e}_{\psi,dtj})-\text{sgn}(e_{\psi,dtj})\big)\bigg\}
	\\
	&+&\Big(1-\frac{n_{dt}}{N_{dt}}\Big)\frac{1}{n_{dt}}\bigg\{
	\sum_{j\in\tilde\G_{dt}-\tilde{\mathcal H}_{dt}}\tilde{e}_{\psi,dtj}-\sum_{j\in\G_{dt}-\tilde{\mathcal H}_{dt}}e_{\psi,dtj}
	+\sum_{j\in\tilde\G_{dt}-\tilde{\mathcal H}_{dt}}\tilde R_{dtj}-\sum_{j\in\G_{dt}-\tilde{\mathcal H}_{dt}}R_{dtj}\bigg\}
	\\
	&+&\Big(1-\frac{n_{dt}}{N_{dt}}\Big)\frac{c_\phi}{n_{dt}}\bigg\{
	\sum_{j\in s_{dt}-(\tilde\G_{dt}-\tilde{\mathcal H}_{dt})}\text{sgn}(\tilde{e}_{\psi,dtj})
	-\sum_{j\in s_{dt}-(\G_{dt}-\tilde{\mathcal H}_{dt})}\text{sgn}({e}_{\psi,dtj})\bigg\}.
\end{eqnarray*}
Assumption (E2) implies that
\begin{align*}
	\sum_{j\in s_{dt}-(\G_{dt}-\tilde{\mathcal H}_{dt})}E[\text{sgn}(\tilde{e}_{\psi,dtj})]&=2\sum_{j\in s_{dt}-\G_{dt}}(\hat\pi_{2,dtj}-\hat\pi_{-2,dtj}),
	\\
	\sum_{j\in s_{dt}-(\G_{dt}-\tilde{\mathcal H}_{dt})}\text{var}(\text{sgn}(\tilde{e}_{\psi,dtj}))&=
	4\sum_{j\in s_{dt}-\G_{dt}}(\hat\pi_{2,dtj}+\hat\pi_{-2,dtj})-4\sum_{j\in s_{dt}-\G_{dt}}(\hat\pi_{2,dtj}-\hat\pi_{-2,dtj})^2.
\end{align*}
From assumption (C2), we obtain that
\begin{equation}\label{HPII.3}
	\frac{1}{n_{dt}}\sum_{j\in s_{dt}-(\G_{dt}-\tilde{\mathcal H}_{dt})}E[\text{sgn}(\tilde e_{\psi,dtj})]=o(1),\quad
	\frac{1}{n_{dt}}\sum_{j\in s_{dt}-(\G_{dt}-\tilde{\mathcal H}_{dt})}\text{var}(\text{sgn}(\tilde e_{\psi,dtj}))=o(1).
\end{equation}
From (\ref{HPII.2}), (\ref{HPII.3}) and assumption (E2), the variance of $\overline{Y}_{dt}^{(2)}$ is
\begin{eqnarray*}
	V_{dt}^{(2)}&=&\text{var}\big(\overline{Y}_{dt}^{(2)}\big)
	=\sum_{j\in U_{dt}}
	\Big(\Big(1-\frac{n_{dt}}{N_{dt}}\Big)\frac{1}{n_{dt}}I_{\tilde{\mathcal H}_{dt}}(j)-\frac{1}{N_{dt}}I_{r_{dt}}(j)\Big)^2
	\text{var}\big(\tilde{e}_{\psi,dtj}(\theta_d)\big)
	\\
	&+&\Big(1-\frac{n_{dt}}{N_{dt}}\Big)^2\frac{\sigma_{\theta_dt}^2}{n_{dt}^2}
	\sum_{j\in\tilde{\mathcal H}_{dt}}\text{var}\big(\tilde{R}_{dtj}-R_{dtj}\big)
	\\
	&+&\Big(1-\frac{n_{dt}}{N_{dt}}\Big)^2\frac{c_\phi^2}{n_{dt}^2}
	\sum_{j\in s_{dt}-\tilde{\mathcal H}_{dt}}\text{var}\big(\text{sgn}(\tilde{e}_{\psi,dtj})-\text{sgn}(e_{\psi,dtj})\big)
	+ o(n^{-1})
	\\
	&=&
	\sum_{j\in U_{dt}}
	\Big(\Big(1-\frac{n_{dt}}{N_{dt}}\Big)^2\frac{1}{n_{dt}^2}I_{\tilde{\mathcal H}_{dt}}(j)+\frac{1}{N_{dt}^2}I_{r_{dt}}(j)\Big)
	\text{var}\big(\tilde{e}_{\psi,dtj}(\theta_d)\big)+o(n^{-1}),
\end{eqnarray*}
where $\text{var}\big(\tilde{e}_{\psi,dtj}(\theta_d)\big)=
\xx_{dtj}^\prime\text{var}\big(\hat\bbeta_{\psi}(\theta_d,\ww_t)\big)\xx_{dtj}$, $j=1,\ldots,N_{dt}$.

From assumption (A9), it holds that
\begin{equation}\label{tildeEedtj}
	\frac{1}{N_{dt}}\sum_{j\in U_{dt}}E\big[\tilde{e}_{\psi,dtj}(\theta_d)\big]=
	\overline{\xx}_{dt}^\prime E\big[\bbeta_{\psi}(\theta_d,\ww_t)-\hat\bbeta_{\psi}(\theta_d,\ww_t)\big]=O(n^{-1/2}),
\end{equation}
where we have defined the population mean of the vector of explanatory variables as
$$\bar{\xx}_{dt}'=\frac{1}{N_{dt}}\sum_{j\in U_{dt}}x_{dtj}'.$$

From (\ref{tildeEedtj}), (\ref{HPII.2}) and (\ref{HPII.3}), the expected prediction difference is
\begin{align*}
	E_{dt}^{(2)}&=E\big[\overline{Y}_{dt}^{(2)}\big]=\sum_{j\in U_{dt}}\Big(
	\Big(1-\frac{n_{dt}}{N_{dt}}\Big)\frac{1}{n_{dt}}I_{s_{dt}}(j)- \frac{1}{N_{dt}}I_{r_{dt}}(j)\Big) E\big[\tilde{e}_{\psi,dtj}(\theta_d)\big]
	\\&
	+\Big(1-\frac{n_{dt}}{N_{dt}}\Big)\frac{1}{n_{dt}}\bigg\{
	\sigma_{\theta_dt}\sum_{j\in\tilde{\mathcal H}_{dt}}E\big[\tilde{R}_{dtj}-R_{dtj}\big]
	\\&+c_\phi\sum_{j\in s_{dt}-\tilde{\mathcal H}_{dt}}E\big[\text{sgn}(\tilde{e}_{\psi,dtj})-\text{sgn}(e_{\psi,dtj})\big]\bigg\}=
	o(1).
\end{align*}
\subsection{Part III: Dealing with the differences {\boldmath $\hat{\overline{Y}}_{dt}^{btmq}-\tilde{\overline{Y}}_{dt}^{btmq}$}}\label{appex.P3}
As $\phi(0)=0$ and $\dot{\phi}(0)=1$, a Taylor series expansion of $\phi(\hat{u}_{\psi,dtj})$ around $u=0$ yields to
\begin{equation}\label{PIIIT1}
	\phi(\hat{u}_{\psi,dtj})=\phi(0)+\dot{\phi}(0)\hat{u}_{\psi,dtj}+\hat{R}_{\psi,dtj}=\hat{u}_{\psi,dtj}+\hat{R}_{\psi,dtj}, \quad j=1,\dots, N_{dt},
\end{equation}
where $\hat{R}_{dtj}=\frac12\,\hat{u}_{\psi,dtj}^{*2}$, $0<|\hat{u}_{\psi,dtj}^{*}|<|\hat{u}_{\psi,dtj}|$ and
$\mbox{var}(\hat{R}_{dtj})\leq\mbox{var}(\hat{u}_{\psi,dtj})$.

From assumption (D1), we have
\begin{equation}\label{HPIII.2}
	\frac{1}{n_{dt}}\sum_{j\in\hat\G_{dt}}E[\hat{R}_{dtj}-\tilde{R}_{dtj}]=o(1),\quad
	\frac{1}{n_{dt}}\sum_{j\in\hat\G_{dt}}\mbox{var}\big((\hat{R}_{dtj}-\tilde{R}_{dtj})^2\big)=o(1).
\end{equation}
From the Taylor series expansion (\ref{PIIIT1}), we have
\begin{eqnarray*}
	\hat{B}_{dt}^{btmq}&=&
	\sum_{j\in \hat \G_{dt}}\sigma_{\theta_dt}
	\phi\big(\hat{u}_{\psi,dtj}\big)+c_\phi \sum_{j\in s_{dt}-\hat\G_{dt}}\mbox{sgn}\big(\hat{e}_{\psi,dtj}\big),
	\\
	&=&
	\sum_{j\in\hat\G_{dt}}\hat{e}_{\psi,dtj}+\sigma_{\theta_dt}\sum_{j\in\hat\G_{dt}}\hat{R}_{dtj}+c_\phi\sum_{j\in s_{dt}-\hat\G_{dt}}\text{sgn}(\hat{e}_{\psi,dtj}).
\end{eqnarray*}
As $\hat{e}_{\psi,dtj}-\tilde{e}_{\psi,dtj}
=\xx_{dtj}^\prime\big(\bbeta_\psi(\hat\theta_d,\ww_t)-\hat\bbeta_\psi(\hat\theta_d,\ww_t)\big)=\hat {e}_{\psi,dtj}(\theta_d)$, $B_{dt}^{(3)}=\hat{B}_{dt}^{btmq}-\tilde{B}_{dt}^{btmq}$ is
\begin{eqnarray*}
	B_{dt}^{(3)}
	&=&\sum_{j\in\hat\G_{dt}}\hat{e}_{\psi,dtj}+\sigma_{\theta_dt}\sum_{j\in\hat\G_{dt}}\hat{R}_{dtj}
	+c_\phi\sum_{j\in s_{dt}-\hat\G_{dt}}\text{sgn}(\hat{e}_{\psi,dtj})
	\\
	&-&\sum_{j\in\tilde\G_{dt}}\tilde{e}_{\psi,dtj}+\sigma_{\theta_dt}\sum_{j\in\tilde\G_{dt}}\tilde{R}_{dtj}
	+c_\phi\sum_{j\in s_{dt}-\tilde\G_{dt}}\text{sgn}(\tilde{e}_{\psi,dtj})
	\\
	&=&\sum_{j\in\hat{\mathcal H}_{dt}}\hat {e}_{\psi,dtj}(\theta_d)
	+\sum_{j\in\hat\G_{dt}-\hat{\mathcal H}_{dt}}\hat{e}_{\psi,dtj}-\sum_{j\in\tilde\G_{dt}-\tilde{\mathcal H}_{dt}}\tilde{e}_{\psi,dtj}
	\\
	&+&\sigma_{\theta_dt}\Big( \sum_{j\in\hat{\mathcal H}_{dt}}(\hat R_{dtj}-\tilde{R}_{dtj})
	+\sum_{j\in\hat\G_{dt}-\hat{\mathcal H}_{dt}}\hat R_{dtj}-\sum_{j\in\tilde\G_{dt}-\hat{\mathcal H}_{dt}}\tilde R_{dtj}\Big)
	\\
	&+&c_\phi\Big( \sum_{j\in s_{dt}-\hat{\mathcal H}_{dt}}(\text{sgn}(\hat{e}_{\psi,dtj})-\text{sgn}(\tilde{e}_{\psi,dtj}))
	\\
	&+&\sum_{j\in s_{dt}-(\hat\G_{dt}-\hat{\mathcal H}_{dt})}\text{sgn}(\hat{e}_{\psi,dtj})-\sum_{j\in s_{dt}-(\tilde\G_{dt}-\hat{\mathcal H}_{dt})}\text{sgn}(\tilde{e}_{\psi,dtj})\Big).
\end{eqnarray*}
The prediction difference $\overline{Y}_{dt}^{(3)}=\hat{\overline{Y}}_{dt}^{btmq}-\tilde{\overline{Y}}_{dt}^{btmq}$ is
\begin{eqnarray*}
	\overline{Y}_{dt}^{(3)}&=&\frac{1}{N_{dt}}\bigg(\sum_{j\in s_{dt}}y_{dtj}+\sum_{j\in r_{dt}} \xx_{dtj}^\prime\hat\bbeta_{\psi}(\hat\theta_d,\ww_t)\bigg)
	+\Big(1-\frac{n_{dt}}{N_{dt}}\Big)\frac{1}{n_{dt}}\hat{B}_{dt}^{btmq}
	\\
	&-&\frac{1}{N_{dt}}\bigg(\sum_{j\in s_{dt}}y_{dtj}+\sum_{j\in r_{dt}} \xx_{dtj}^\prime\hat\bbeta_{\psi}(\theta_d,\ww_t)\bigg)
	+\Big(1-\frac{n_{dt}}{N_{dt}}\Big)\frac{1}{n_{dt}}\tilde{B}_{dt}^{btmq}
	\\
	&=&\Big(1-\frac{n_{dt}}{N_{dt}}\Big)\frac{1}{n_{dt}}B_{dt}^{(3)}
	-\frac{1}{N_{dt}}\sum_{j\in r_{dt}} \hat {e}_{\psi,dtj}(\theta_d).
\end{eqnarray*}
By substituting $B_{dt}^{(3)}$, we obtain
\begin{eqnarray*}
	\overline{Y}_{dt}^{(3)}&=&\Big(1-\frac{n_{dt}}{N_{dt}}\Big)\frac{1}{n_{dt}}
	\Bigg\{
	\sum_{j\in\tilde{\mathcal H}_{dt}}\hat {e}_{\psi,dtj}(\theta_d)
	+\sum_{j\in\hat\G_{dt}-\hat{\mathcal H}_{dt}}\hat{e}_{\psi,dtj}-\sum_{j\in\tilde\G_{dt}-\hat{\mathcal H}_{dt}}\tilde{e}_{\psi,dtj}\\
	&+&\sigma_{\theta_dt}\Big( \sum_{j\in\hat{\mathcal H}_{dt}}(\hat R_{dtj}-\tilde R_{dtj})
	+\sum_{j\in\hat\G_{dt}-\hat{\mathcal H}_{dt}}\hat R_{dtj}-\sum_{j\in\tilde\G_{dt}-\hat{\mathcal H}_{dt}}\tilde R_{dtj}\Big)
	\\
	&+&c_\phi\Big( \sum_{j\in s_{dt}-\hat{\mathcal H}_{dt}}(\text{sgn}(\hat{e}_{\psi,dtj})-\text{sgn}(\tilde{e}_{\psi,dtj}))
	+\sum_{j\in s_{dt}-(\hat\G_{dt}-\hat{\mathcal H}_{dt})}\text{sgn}(\hat{e}_{\psi,dtj})
	\\
	&-&\sum_{j\in s_{dt}-(\tilde\G_{dt}-\hat{\mathcal H}_{dt})}\text{sgn}(\tilde{e}_{\psi,dtj})\Big)
	\Big\}
	-
	\frac{1}{N_{dt}}\sum_{j\in r_{dt}}\hat {e}_{\psi,dtj}(\theta_d).
\end{eqnarray*}
We can write $Y_{dt}^{(3)}$ in the form
\begin{eqnarray*}
	&&\overline{Y}_{dt}^{(3)}=\sum_{j\in U_{dt}}
	\Big(\Big(1-\frac{n_{dt}}{N_{dt}}\Big)\frac{1}{n_{dt}}I_{\hat{\mathcal H}_{dt}}(j)-\frac{1}{N_{dt}}I_{r_{dt}}(j)\Big)\hat{e}_{\psi,dtj}(\theta_d)
	\\
	&+&
	\Big(1-\frac{n_{dt}}{N_{dt}}\Big)\frac{1}{n_{dt}}\bigg\{
	\sigma_{\theta_dt}\sum_{j\in\hat{\mathcal H}_{dt}}\big\{\hat{R}_{dtj}-\tilde R_{dtj}\big\}
	+c_\phi\sum_{j\in s_{dt}-\hat{\mathcal H}_{dt}}\big(\text{sgn}(\hat{e}_{\psi,dtj})-\text{sgn}(\tilde e_{\psi,dtj})\big)\bigg\}
	\\
	&+&\Big(1-\frac{n_{dt}}{N_{dt}}\Big)\frac{1}{n_{dt}}\bigg\{
	\sum_{j\in\hat\G_{dt}-\hat{\mathcal H}_{dt}}\hat{e}_{\psi,dtj}-\sum_{j\in\tilde\G_{dt}-\hat{\mathcal H}_{dt}}\tilde e_{\psi,dtj}
	+\sum_{j\in\hat\G_{dt}-\hat{\mathcal H}_{dt}}\hat R_{dtj}-\sum_{j\in\G_{dt}-\hat{\mathcal H}_{dt}}R_{dtj}\bigg\}
	\\
	&+&\Big(1-\frac{n_{dt}}{N_{dt}}\Big)\frac{1}{n_{dt}}\bigg\{
	c_\phi\sum_{j\in s_{dt}-(\hat\G_{dt}-\hat{\mathcal H}_{dt})}\text{sgn}(\hat{e}_{\psi,dtj})
	-c_\phi\sum_{j\in s_{dt}-(\tilde\G_{dt}-\hat{\mathcal H}_{dt})}\text{sgn}(\tilde{e}_{\psi,dtj})\bigg\}.
\end{eqnarray*}
Assumption (E3) implies that
\begin{align*}
	\sum_{j\in s_{dt}-(\G_{dt}-\hat{\mathcal H}_{dt})}E[\text{sgn}(\hat{e}_{\psi,dtj})]&=2\sum_{j\in s_{dt}-\G_{dt}}(\hat\pi_{2,dtj}-\hat\pi_{-2,dtj}),
	\\
	\sum_{j\in s_{dt}-(\G_{dt}-\hat{\mathcal H}_{dt})}\text{var}(\text{sgn}(\hat{e}_{\psi,dtj}))&=
	4\sum_{j\in s_{dt}-\G_{dt}}(\hat\pi_{2,dtj}+\hat\pi_{-2,dtj})-4\sum_{j\in s_{dt}-\G_{dt}}(\hat\pi_{2,dtj}-\hat\pi_{-2,dtj})^2.
\end{align*}
From assumption (D2), we obtain that
\begin{equation}\label{HPIII.3}
	\frac{1}{n_{dt}}\sum_{j\in s_{dt}-(\G_{dt}-\hat{\mathcal H}_{dt})}E[\text{sgn}(\hat{e}_{\psi,dtj})]=o(1),\quad
	\frac{1}{n_{dt}}\sum_{j\in s_{dt}-(\G_{dt}-\hat{\mathcal H}_{dt})}\text{var}(\text{sgn}(\hat{e}_{\psi,dtj}))=o(1).
\end{equation}
From (\ref{HPIII.2}), (\ref{HPIII.3}) and assumption (E3), the variance of $\overline{Y}_{dt}^{(3)}$ is
\begin{eqnarray*}
	V_{dt}^{(3)}&=&\text{var}\big(\overline{Y}_{dt}^{(3)}\big)
	=\sum_{j\in U_{dt}}
	\Big(\Big(1-\frac{n_{dt}}{N_{dt}}\Big)\frac{1}{n_{dt}}I_{\hat{\mathcal H}_{dt}}(j)-\frac{1}{N_{dt}}I_{r_{dt}}(j)\Big)^2
	\text{var}\big(\hat{e}_{\psi,dtj}(\theta_d)\big)
	\\
	&+&\Big(1-\frac{n_{dt}}{N_{dt}}\Big)^2\frac{1}{n_{dt}^2}
	\sigma_{\theta_dt}^2\sum_{j\in\hat{\mathcal H}_{dt}}\text{var}\big(\hat{R}_{dtj}-\tilde R_{dtj}\big)
	\\
	&+&\Big(1-\frac{n_{dt}}{N_{dt}}\Big)^2\frac{1}{n_{dt}^2}
	c_\phi^2\sum_{j\in s_{dt}-\hat{\mathcal H}_{dt}}\text{var}\big(\text{sgn}(\hat{e}_{\psi,dtj})-\text{sgn}(\tilde e_{\psi,dtj})\big)
	+ o(n^{-1})
	\\
	&=&
	\sum_{j\in U_{dt}}
	\Big(\Big(1-\frac{n_{dt}}{N_{dt}}\Big)^2\frac{1}{n_{dt}^2}I_{\hat{\mathcal H}_{dt}}(j)+\frac{1}{N_{dt}^2}I_{r_{dt}}(j)\Big)
	\text{var}\big(\hat{e}_{\psi,dtj}(\theta_d)\big)+o(n^{-1}),
\end{eqnarray*}
where $\text{var}\big(\hat{e}_{\psi,dtj}(\theta_d)\big)=
\xx_{dtj}^\prime\text{var}\big(\hat\bbeta_{\psi}(\hat\theta_d,\ww_t)\big)\xx_{dtj}$, $j=1,\ldots,N_{dt}$.

From assumption (A9), it holds that
\begin{equation}\label{hatEedtj}
	\frac{1}{N_{dt}}\sum_{j\in U_{dt}}E\big[\hat{e}_{\psi,dtj}(\theta_d)\big]=
	\overline{\xx}_{dt}^\prime E\big[\bbeta_{\psi}(\theta_d,\ww_t)-\hat\bbeta_{\psi}(\hat\theta_d,\ww_t)\big]=O(n^{-1/2}).
\end{equation}

From (\ref{hatEedtj}), (\ref{HPIII.2}) and (\ref{HPIII.3}), the expected prediction difference is
\begin{align*}
	E_{dt}^{(3)}&=E\big[\overline{Y}_{dt}^{(3)}\big]=\sum_{j\in U_{dt}}\Big(
	\Big(1-\frac{n_{dt}}{N_{dt}}\Big)\frac{1}{n_{dt}}I_{s_{dt}}(j)- \frac{1}{N_{dt}}I_{r_{dt}}(j)\Big) E\big[\hat{e}_{\psi,dtj}(\theta_d)\big]
	\\&
	+\Big(1-\frac{n_{dt}}{N_{dt}}\Big)\frac{1}{n_{dt}}\bigg\{
	\sigma_{\theta_dt}\sum_{j\in\hat{\mathcal H}_{dt}}E\big[\hat{R}_{dtj}-\tilde R_{dtj}\big]
	\\&+c_\phi\sum_{j\in s_{dt}-\hat{\mathcal H}_{dt}}E\big[\text{sgn}(\hat{e}_{\psi,dtj})-\text{sgn}(\tilde e_{\psi,dtj})\big]\bigg\}=o(1).
\end{align*}

\subsection{Final expression of the MSE}\label{Assec.final}
First of all, it holds that
\[MSE\big(\hat{\overline{Y}}_{dt}^{btmq}\big)=
E\big[\big(\hat{\overline{Y}}_{dt}^{btmq}-\overline{Y}_{dt}\big)^2\big]=\mbox{var}\big(\hat{\overline{Y}}_{dt}^{btmq}-\overline{Y}_{dt}\big)+(E\big[\hat{\overline{Y}}_{dt}^{btmq}-\overline{Y}_{dt}\big])^2.\]
Based on the decomposition
$$
\hat{\overline{Y}}_{dt}^{btmq}-\overline{Y}_{dt}=
\big(\hat{\overline{Y}}_{dt}^{btmq}-\tilde{\overline{Y}}_{dt}^{btmq}\big)
+\big(\tilde{\overline{Y}}_{dt}^{btmq}-\overline{Y}_{dt}^{btmq}\big)
+\big(\overline{Y}_{dt}^{btmq}-\overline{Y}_{dt}\big)=\overline{Y}_{dt}^{(3)}+\overline{Y}_{dt}^{(2)}+\overline{Y}_{dt}^{(1)},
$$
we can write
\begin{align*}
	\mbox{var}\big(\hat{\overline{Y}}_{dt}^{btmq}-\overline{Y}_{dt}\big)&=
	V_{dt}^{(1)}+V_{dt}^{(2)}+V_{dt}^{(3)}+2\mbox{cov}\big(\overline{Y}_{dt}^{(3)},\overline{Y}_{dt}^{(2)}\big)
	+2\mbox{cov}\big(\overline{Y}_{dt}^{(3)},\overline{Y}_{dt}^{(1)}\big)+2\mbox{cov}\big(\overline{Y}_{dt}^{(2)},\overline{Y}_{dt}^{(1)}\big),
	\\
	E\big[\hat{\overline{Y}}_{dt}^{btmq}-\overline{Y}_{dt}\big]&=
	E_{dt}^{(1)}+E_{dt}^{(2)}+E_{dt}^{(3)}=E_{dt}^{(1)}+o(1).
\end{align*}
The covariances are $\mbox{cov}\big(\overline{Y}_{dt}^{(3)},\overline{Y}_{dt}^{(2)}\big)=E\big[\overline{Y}_{dt}^{(3)}\overline{Y}_{dt}^{(2)}\big]+o(1)$, $\mbox{cov}\big(\overline{Y}_{dt}^{(3)},\overline{Y}_{dt}^{(1)}\big)=
E\big[\overline{Y}_{dt}^{(3)}\overline{Y}_{dt}^{(1)}]+o(1)$ and
$\mbox{cov}\big(\overline{Y}_{dt}^{(2)},\overline{Y}_{dt}^{(1)}\big)=
E\big[\overline{Y}_{dt}^{(2)}\overline{Y}_{dt}^{(1)}]+o(1)$.
Under regularity assumptions, the expectations of the previous cross-products should be $o(1)$.
Therefore, an approximation of $MSE\big(\hat{\overline{Y}}_{dt}^{btmq}\big)$ is
$$
MSE\big(\hat{\overline{Y}}_{dt}^{btmq}\big)=V_{dt}^{(1)}+V_{dt}^{(2)}+V_{dt}^{(3)}+E_{dt}^{(1)2}+o(1).
$$
The following theorem summarizes the final MSE approximation of the manuscript.
\setcounter{theorem}{0}
\begin{theorem}
	\normalfont Under assumptions ($\Phi$1), (N1)-(N2), (Q1), (A1)-(A9), (B1)-(B4), (C1)-(C2), (D1)-(D2), (E1)-(E3), listed in
	Section \ref{Assect.ass}, a first-order approximation of $MSE(\hat{\overline{Y}}_{dt}^{btmq})$ is
	\begin{eqnarray*}
		\nonumber
		MSE\big(\hat{\overline{Y}}_{dt}^{btmq}\big)&=&\sum_{j\in U_{dt}}\Big(\Big(1-\frac{n_{dt}}{N_{dt}}\Big)^2\frac{1}{n_{dt}^2}I_{\G_{dt}}(j)+ \frac{1}{N_{dt}^2}I_{r_{dt}}(j)\Big)
		\Big(\xx_{dtj}^\prime\bm\kappa_\psi(\theta_{d},\ww_t)\Big)^2\xi_{dt}^2
		\\
		\nonumber
		&+&\sum_{j\in U_{dt}}\Big(\Big(1-\frac{n_{dt}}{N_{dt}}\Big)^2\frac{1}{n_{dt}^2}I_{\tilde{\mathcal H}_{dt}}(j)+\frac{1}{N_{dt}^2}I_{r_{dt}}(j)\Big)
		\xx_{dtj}^\prime\text{var}\big(\hat\bbeta_{\psi}(\theta_d,\ww_t)\big)\xx_{dtj}
		\\
		\nonumber
		&+&\sum_{j\in U_{dt}}    \Big(\Big(1-\frac{n_{dt}}{N_{dt}}\Big)^2\frac{1}{n_{dt}^2}I_{\hat{\mathcal H}_{dt}}(j)+\frac{1}{N_{dt}^2}I_{r_{dt}}(j)\Big)
		\xx_{dtj}^\prime\text{var}\big(\hat\bbeta_{\psi}(\hat\theta_d,\ww_t)\big)\xx_{dtj}
		\\
		\nonumber
		&+&\Big(1-\frac{n_{dt}}{N_{dt}}\Big)^2\bigg( \frac{c_\phi}{n_{dt}}\sum_{j\in s_{dt}-\G_{dt}}E\big[\text{sgn}(e_{\psi,dtj})\big]
		+\frac{\sigma_{\theta_dt}}{n_{dt}}\sum_{j\in\G_{dt}}E[R_{dtj}]\bigg) ^2+o(1).
	\end{eqnarray*}
\end{theorem}
\subsection{Estimation of the final expression of the MSE}\label{Assec.final.aprrox}
This section uses the simplified notation $MSE\big(\hat{\overline{Y}}_{dt}^{btmq}\big)=S_1+S_2+S_3+S_4+o(1)$.
The first summand of $MSE\big(\hat{\overline{Y}}_{dt}^{btmq}\big)$ is
\begin{eqnarray*}
	\nonumber
	S_1&=&S_{11}+S_{12}
	=\sum_{j\in U_{dt}}\Big(\Big(1-\frac{n_{dt}}{N_{dt}}\Big)^2\frac{1}{n_{dt}^2}I_{\G_{dt}}(j)+ \frac{1}{N_{dt}^2}I_{r_{dt}}(j)\Big)
	\Big(\xx_{dtj}^\prime\bm\kappa_\psi(\theta_{d},\ww_t)\Big)^2\xi_{dt}^2
	\\
	\nonumber
	&=&\Big(1-\frac{n_{dt}}{N_{dt}}\Big)^2\frac{1}{n_{dt}^2}\sum_{j\in \G_{dt}}\Big(\xx_{dtj}^\prime\bm\kappa_\psi(\theta_d,\ww_t)\Big)^2\xi_{dt}^2
	+\frac{1}{N_{dt}^2} \sum_{j\in r_{dt}}\Big(\xx_{dtj}^\prime\bm\kappa_\psi(\theta_d,\ww_t)\Big)^2\xi_{dt}^2.
\end{eqnarray*}

Provided $q_{dtj}-\theta_d\neq 0$, $j=1,\dots,N_{dt}$, we get from the definition of the TWMQ models \eqref{TMQ3} and the Taylor series expansion \eqref{T0} that
\begin{align*}
	\bm\kappa_{\psi}(\theta_d,\ww_t)&=(q_{dtj}-\theta_d)^{-1}(\bbeta_\psi(q_{dtj},\ww_t)-\bbeta_\psi(\theta_d,\ww_t)-\bm r_{dtj}),\\
	\xx_{dtj}'\bm\kappa_{\psi}(\theta_d,\ww_t)&=(q_{dtj}-\theta_d)^{-1}(y_{dtj}-\xx_{dtj}'\bbeta_\psi(\theta_d,\ww_t)-{\xx}_{dtj}'\bm r_{dtj})\\&=(q_{dtj}-\theta_d)^{-1}(e_{\psi,dtj}-\xx_{dtj}'\bm r_{dtj}).
\end{align*}
If $\xx_{dtj}'\bm r_{dtj}\approx 0$, the first term of $S_1$ can be estimated by
\begin{align*}
	\hat S_{11,0}=\Big(1-\frac{n_{dt}}{N_{dt}}\Big)^2\frac{\hat\xi_{dt}^2}{n_{dt}^2}\sum_{j\in \hat\G_{dt}}\frac{1}{(\hat q_{dtj}-\hat\theta_d)^{2}}\hat e_{\psi,dtj}^2,
\end{align*}
where $\hat\xi_{dt}^2$ is an estimator of $\text{var}(q_{dtj})$, i.e.
$$
\hat\xi_{dt}^2=\widehat{\text{var}}(q_{dtj})=\frac{1}{n_{dt}-1}\sum_{j\in s_{dt}}(\hat q_{dtj}-\hat{\bar{q}}_{dt.})^2,\quad
\hat{\bar{q}}_{dt.}=\frac{1}{n_{dt}}\sum_{j\in s_{dt}}\hat q_{dtj}.
$$
The second term of $S_1$ can be estimated as
\begin{align*}
	\hat S_{12,0}=\frac{N_{dt}-n_{dt}}{n_{dt}}\frac{\hat\xi_{dt}^2}{N_{dt}^2}\sum_{j\in s_{dt}}\frac{1}{(\hat q_{dtj}-\hat\theta_d)^{2}}\hat e_{\psi,dtj}^2.
\end{align*}
Therefore, we can estimate $V_{dt}^{(1)}$ by $\hat S_{1,0}= \hat S_{11,0}+\hat S_{12,0}$. Nevertheless, the differences $\hat q_{dtj}-\hat\theta_d$ are expected to be close to zero, leading to instability problems due to multiplication by $(\hat q_{dtj}-\hat\theta_d)^{-2}$, $j=1,\dots,n_{dt}$, both in the estimation of $S_{11,0}$ and $S_{12,0}$. For this reason, we do not recommend using them. Potential solutions include proposing the following estimators
\begin{eqnarray*}
	\hat S_{11}&=&\Big(1-\frac{n_{dt}}{N_{dt}}\Big)^2\frac{\hat\xi_{dt}^2}{n_{dt}^2}\Big( \frac{1}{\text{card}(\hat\G_{dt})}\sum_{j\in \hat\G_{dt}}{(\hat q_{dtj}-\hat\theta_d)^{2}}\Big) ^{-1}\sum_{j\in \hat\G_{dt}}\hat e_{\psi,dtj}^2,\\
	\hat S_{12}&=&\frac{N_{dt}-n_{dt}}{n_{dt}}\frac{\hat\xi_{dt}^2}{N_{dt}^2}\Big( \frac{1}{n_{dt}}\sum_{j\in s_{dt}}{(\hat q_{dtj}-\hat\theta_d)^{2}}\Big) ^{-1}\sum_{j\in s_{dt}}\hat e_{\psi,dtj}^2,
\end{eqnarray*}
such that we can finally write $\hat S_1=\hat S_{11}+\hat S_{12}$. Clearly, this new estimation of $S_1$ largely avoids the problems of numerical instability, but it may also bias the final estimation. All things considered, we strongly recommend this second approach, which provides more stable results.

On the other hand, we can use estimators of $S_2$ and $S_3$ to estimate the variance terms $V_{dt}^{(2)}$ and $V_{dt}^{(3)}$, respectively. An estimator for both $S_2$ and $S_3$ is obtained as follows
\begin{eqnarray*}
	\hat V_{dt}^{(2)}&=&\hat V_{dt}^{(3)}=\sum_{j\in U_{dt}}
	\Big(\Big(1-\frac{n_{dt}}{N_{dt}}\Big)^2\frac{1}{n_{dt}^2}I_{\hat{\mathcal G}_{dt}}(j)+\frac{1}{N_{dt}^2}I_{r_{dt}}(j)\Big)
	\xx_{dtj}^\prime \hat V_{\beta}\xx_{dtj}
	\\
	&=&\Big(1-\frac{n_{dt}}{N_{dt}}\Big)^2\frac{1}{n_{dt}^2}
	\sum_{j\in \hat{\mathcal G}_{dt}}\xx_{dtj}^\prime \hat V_{\beta}\xx_{dtj}
	+\frac{1}{N_{dt}^2}\sum_{j\in r_{dt}}\xx_{dtj}^\prime \hat V_{\beta}\xx_{dtj},
\end{eqnarray*}
where $\hat V_{\beta}$ is an estimator of $V_{\beta}=\text{var}\big(\hat\bbeta_{\psi}(\hat\theta_d,\ww_t))$, as the one given in \eqref{var.beta}.

The last element $S_4$ of $MSE\big(\hat{\overline{Y}}_{dt}^{btmq}\big)$ corresponds to the bias term $E_{dt}^{(1)2}$. First, we propose
\begin{eqnarray*}
	\sum_{j\in s_{dt}-\G_{dt}}\hat {E}\big[\text{sgn}(e_{\psi,dtj})\big]&=&\sum_{j\in s_{dt}-\hat\G_{dt}}\text{sgn}(\hat e_{\psi,dtj}).
\end{eqnarray*}
From the Taylor expansion \eqref{T1},  $R_{dtj}=\frac12\,u_{\psi,dtj}^{*2}$, $0<|u_{\psi,dtj}^{*}|<|u_{\psi,dtj}|$, $j=1,\ldots,N_{dt}$, so
\begin{eqnarray*}
	\sum_{j\in\G_{dt}}\hat E[R_{dtj}]&=&\frac{1}{2\sigma_{\theta_dt}^2}\sum_{j\in\hat\G_{dt}}\hat e_{\psi,dtj}^2.
\end{eqnarray*}
Finally, we derive the following estimator of $MSE\big(\hat{\overline{Y}}_{dt}^{btmq}\big)$:
\begin{eqnarray*}
	mse_{3,dt}^{btmq}&=&
	\Big(1-\frac{n_{dt}}{N_{dt}}\Big)^2\frac{\hat\xi_{dt}^2}{n_{dt}^2}\Big( \frac{1}{\text{card}(\hat\G_{dt})}\sum_{j\in \hat\G_{dt}}{(\hat q_{dtj}-\hat\theta_d)^{2}}\Big) ^{-1}\sum_{j\in \hat\G_{dt}}\hat e_{\psi,dtj}^2+
	\\
	&+&\frac{N_{dt}-n_{dt}}{n_{dt}}\frac{\hat\xi_{dt}^2}{N_{dt}^2}\Big( \frac{1}{n_{dt}}\sum_{j\in s_{dt}}{(\hat q_{dtj}-\hat\theta_d)^{2}}\Big) ^{-1}\sum_{j\in s_{dt}}\hat e_{\psi,dtj}^2
	\\
	&+&2\Big(1-\frac{n_{dt}}{N_{dt}}\Big)^2\frac{1}{n_{dt}^2}\sum_{j\in \hat{\mathcal G}_{dt}}\xx_{dtj}^\prime \hat V_{\beta}\xx_{dtj}
	+\frac{1}{N_{dt}^2}\sum_{j\in r_{dt}}\xx_{dtj}^\prime \hat V_{\beta}\xx_{dtj}
	\\
	&+&
	\Big(1-\frac{n_{dt}}{N_{dt}}\Big)^2\frac{1}{n_{dt}^2}\bigg( c_\phi\sum_{j\in s_{dt}-\hat\G_{dt}}\text{sgn}(\hat e_{\psi,dtj})
	+\frac{1}{2\sigma_{\theta_dt}}\sum_{j\in\hat\G_{dt}}\hat e_{\psi,dtj}^2\bigg) ^2.
\end{eqnarray*}
Finally, an estimator of $RMSE\big(\hat{\overline{Y}}_{dt}^{btmq}\big)=\big(MSE\big(\hat{\overline{Y}}_{dt}^{btmq}\big)\big)^{1/2}$ is
$rmse_{3,dt}^{btmq}=\big(mse_{3,dt}^{btmq}\big)^{1/2}$.

\section{Selection of area-time specific robustness parameters}\label{sect.app2}

\setcounter{equation}{0}
\renewcommand{\theequation}{\thesection.\arabic{equation}}

In this Section \ref{sect.app2}, we provide the proof of Theorem \ref{exis.uni}.

\begin{proof}
	Given $mse_{dt}^{btmq}\in\{mse_{1,dt}^{btmq}, mse_{2,dt}^{btmq}\}$, we write $A_{dt}(c_\phi)$ as a function of $c_\phi$, i.e.
	\begin{eqnarray*}
		A_{dt}(c_\phi)=\Big(1-\frac{n_{dt}}{N_{dt}}\Big)^2\Big(\frac{\sigma_{\theta_d t}}{n_{dt}}\Big)^2 \underset{j\in s_{dt}}\sum\phi^2\big(\hat{u}_{\psi,dtj}\big)
		+\Big( \hat{B}_{dt}+\Big(1-\frac{n_{dt}}{N_{dt}}\Big)\frac{\sigma_{\theta_d t}}{n_{dt}}\sum_{j\in s_{dt}}
		\phi\big(\hat{u}_{\psi,dtj}\big)\Big) ^2.
	\end{eqnarray*}
	Let $\U=\{|\hat{u}_{\psi,dtj}|_{(1)}, \dots, |\hat{u}_{\psi,dtj}|_{(n_{dt})}\}$ be the ordered version of the set
	$U=\{|\hat{u}_{\psi,dt1}|,\dots, |\hat{u}_{\psi,dtn_{dt}}|\}$, so that
	$|\hat{u}_{\psi,dtj}|_{(1)}=\underset{j\in s_{dt}}{\min}|\hat{u}_{\psi,dtj}|$ and $|\hat{u}_{\psi,dtj}|_{(n_{dt})}=\underset{j\in s_{dt}}{\max}|\hat{u}_{\psi,dtj}|$.
	Define $|\hat{u}_{\psi,dtj}|_{(0)}=0$ and
	$\Lambda_{dt,\ell}=\left\lbrace j\in\{1,\dots, n_{dt}\}:  |\hat{u}_{\psi,dtj}|\leq|\hat{u}_{\psi,dtj}|_{(\ell+1)}\right\rbrace$,
	where $\text{card}(\Lambda_{dt,\ell})=n_{dt}-n_{dt,\ell}$, \ $n_{dt,\ell}=n_{dt,\ell}^++n_{dt,\ell}^-\in\mathbb N$, $\ell=0,\dots, n_{dt}-1$, and
	\begin{align*}
		&n_{dt,\ell}^+=\text{card}\left\lbrace |\hat u_{\psi,dtj}|\in\U: |\hat{u}_{\psi,dtj}|>|\hat{u}_{\psi,dtj}|_{(\ell+1)},\ \hat{u}_{\psi,dtj}> 0\right\rbrace \in\mathbb N,\\&
		n_{dt,\ell}^-=\text{card}\left\lbrace |\hat u_{\psi,dtj}|\in\U: |\hat{u}_{\psi,dtj}|>|\hat{u}_{\psi,dtj}|_{(\ell+1)},\ \hat{u}_{\psi,dtj}< 0\right\rbrace \in\mathbb N.
	\end{align*}
	
	

The $c_\phi$-dependent terms of $A_{dt}(c_\phi)$ are continuous in $c_\phi\in [0,\infty)$,
piecewise quadratic in $c_\phi\in I_{\ell}=[|\hat{u}_{\psi,dtj}|_{(\ell)}, |\hat{u}_{\psi,dtj}|_{(\ell+1)})$, $\ell=0,\dots, n_{dt}-1$, and
constant in $c_\phi\in I_{n_{dt}}=[|\hat{u}_{\psi,dtj}|_{(n_{dt})}, \infty)$.
Therefore, we have
\begin{eqnarray*}
	A_{dt}(c_\phi)&=&\Big(\hat B_{dt}+\Big( 1-\frac{n_{dt}}{N_{dt}}\Big)\frac{\sigma_{\theta_dt}}{n_{dt}}\Big(c_\phi(n_{dt,\ell}^+-n_{dt,\ell}^-)
	+ \underset{j\in \Lambda_{dt,\ell}}\sum  \hat{u}_{\psi,dtj}\Big) \Big)^2
	\\
	&+&\Big( 1-\frac{n_{dt}}{N_{dt}}\Big)^2\frac{\sigma_{\theta_dt}^2}{n_{dt}^2}
	\Big( c_\phi^2 n_{dt,\ell}+\underset{j\in \Lambda_{dt,\ell}}\sum \hat{u}^2_{\psi,dtj}\Big),\,\,\,
	c_\phi\in I_{\ell},\, \ell=0,\dots, n_{dt}-1,
	\\
	A_{dt}(c_\phi)&=&\Big(\hat B_{dt}+\Big( 1-\frac{n_{dt}}{N_{dt}}\Big)\frac{\sigma_{\theta_dt}}{n_{dt}}\underset{j\in s_{dt}}\sum \hat{u}_{\psi,dtj} \Big)^2
	\\ &+&\Big( 1-\frac{n_{dt}}{N_{dt}}\Big)^2\frac{\sigma_{\theta_dt}^2}{n_{dt}^2}
	\underset{j\in s_{dt}}\sum \hat{u}_{\psi,dtj}^2,\,\,\,
	c_\phi\geq |\hat{u}_{\psi,dtj}|_{(n_{dt})}.
\end{eqnarray*}

\textit{Existence.} Since $A_{dt}(c_\phi)$ is a constant function in $I_{n_{dt}}$, the search for extreme values is reduced to the compact (closed and bounded) interval $\left[ 0, |\hat{u}_{\psi,dtj}|_{(n_{dt})}\right]$.
Further, $A_{dt}(c_\phi)$ is a piecewise function in $\left[ 0, |\hat{u}_{\psi,dtj}|_{(n_{dt})}\right]$ but, as an inherited property of function $\phi(u)$, $A_{dt}(c_\phi)$ is continuous in $[0,\infty)$. By virtue of the Weierstrass Theorem, $A_{dt}(c_\phi)$ reaches its absolute maximum and minimum values in $\left[ 0, |\hat{u}_{\psi,dtj}|_{(n_{dt})}\right]$. If the minimum is reached at $c_\phi=0$, no bias correction is needed. If the minimum is reached at $c_\phi=|\hat{u}_{\psi,dtj}|_{(n_{dt})}$, is also reached in $[|\hat{u}_{\psi,dtj}|_{(n_{dt})},\infty)$, and the bias correction is maximal.
\medskip

\textit{Uniqueness.}
By definition, $A_{dt}(c_\phi)$ is an infinitely differentiable function in $c_\phi\in[0,\infty)-\mathcal U$.
The first and second order derivatives of $A_{dt}(c_\phi)$ in $c_\phi\in I_{\ell}-\{|\hat{u}_{\psi,dtj}|_{(\ell)}\}$, $\ell=0,\dots, n_{dt}-1$, are
\begin{eqnarray*}
	\frac{\partial A_{dt}(c_\phi)}{\partial c_\phi}&=&2\bigg(\hat B_{dt}+\Big( 1-\frac{n_{dt}}{N_{dt}}\Big)\frac{\sigma_{\theta_dt}}{n_{dt}}\Big(c_\phi(n_{dt,\ell}^+-n_{dt,\ell}^-)
	+\underset{j\in \Lambda_{dt,\ell}}\sum \hat{u}_{\psi,dtj}\Big)\bigg)
	\\
	&\cdot&\Big(1-\frac{n_{dt}}{N_{dt}}\Big)\frac{\sigma_{\theta_dt}}{n_{dt}}(n_{dt,\ell}^+-n_{dt,\ell}^-)
	+
	2c_\phi\Big(1-\frac{n_{dt}}{N_{dt}}\Big)^2\Big( \frac{\sigma_{\theta_dt}}{n_{dt}}\Big)^2n_{dt,\ell},
	\\
	\frac{\partial^2A_{dt}(c_\phi)}{\partial^2 c_\phi}&=&2\Big(1-\frac{n_{dt}}{N_{dt}}\Big)^2\Big( \frac{\sigma_{\theta_dt}}{n_{dt}}\Big)^2(n_{dt,\ell}^+-n_{dt,\ell}^-)^2+2\Big(1-\frac{n_{dt}}{N_{dt}}\Big)^2\Big( \frac{\sigma_{\theta_dt}}{n_{dt}}\Big)^2n_{dt,\ell}
	\\
	&=&2\Big(1-\frac{n_{dt}}{N_{dt}}\Big)^2\Big( \frac{\sigma_{\theta_dt}}{n_{dt}}\Big)^2\Big( (n_{dt,\ell}^+-n_{dt,\ell}^-)^2+n_{dt,\ell}\Big).
\end{eqnarray*}
For $c_\phi\in I_{\ell}-\{|\hat{u}_{\psi,dtj}|_{(\ell)}\}$, $\ell=0,\dots, n_{dt}-1$, it follows that
$\frac{\partial^2A_{dt}(c_\phi)}{\partial^2 c_\phi}>0$ if and only if $n_{dt,\ell}>0$,
so $A_{dt}(c_\phi)$ is strictly convex in $\overset{n_{dt}}{\underset{l=0}\bigcup}  (I_{\ell}-\{|\hat{u}_{\psi,dtj}|_{(\ell)}\})=[0,|\hat{u}_{\psi,dtj}|_{(n_{dt})}]-\mathcal U$. Given that a strictly convex continuous function in an open set is strictly convex at its closure, $A_{dt}(c_\phi)$ is strictly convex in the compact interval $[0,|\hat{u}_{\psi,dtj}|_{(n_{dt})}]$.
Therefore, the uniqueness of the global minimum of $A_{dt}(c_\phi)$ in $[0,|\hat{u}_{\psi,dtj}|_{(n_{dt})}]$ is guaranteed.
\medskip

\textit{Explicit expression of $\hat c_{\phi, dt}$.}
We distinguish two cases.
Case 1 is the extreme solution $\hat c_{\phi, dt}=|\hat{u}_{\psi,dtj}|_{(n_{dt})}$.
Case 2 is a solution $\hat c_{\phi, dt}\in I_{\ell}$ for some $\ell\in\{0,\dots, n_{dt}-1\}$,
so either $\hat c_{\phi, dt}=|\hat u_{\psi,dtj}|_{(\ell)}$, or $\hat c_{\phi, dt}$ fulfills that
$\frac{\partial A_{dt}(c_\phi)}{\partial c_\phi}\Big|_{c_{\phi}=\hat c_{\phi, dt}}=0$, i.e.
\begin{align*}
	&\bigg( \hat B_{dt}+\Big( 1-\frac{n_{dt}}{N_{dt}}\Big)\frac{\sigma_{\theta_dt}}{n_{dt}}\Big(\hat c_{\phi, dt}(n_{dt,\ell}^+-n_{dt,\ell}^-)
	+\underset{j\in \Lambda_{dt,\ell}}\sum \hat{u}_{\psi,dtj}\Big)\bigg)(n_{dt,\ell}^+-n_{dt,\ell}^-)\\&
	+\hat c_{\phi, dt}\Big(1-\frac{n_{dt}}{N_{dt}}\Big)\Big( \frac{\sigma_{\theta_dt}}{n_{dt}}\Big)n_{dt,\ell}=0\Longleftrightarrow \left(\hat B_{dt}+\left( 1-\frac{n_{dt}}{N_{dt}}\right)\frac{\sigma_{\theta_dt}}{n_{dt}}\underset{j\in \Lambda_{dt,\ell}}\sum \hat u_{\psi,dtj}\right)(n_{dt,\ell}^+-n_{dt,\ell}^-)\\&
	+\left( 1-\frac{n_{dt}}{N_{dt}}\right)\frac{\sigma_{\theta_dt}}{n_{dt}}\left((n_{dt,\ell}^+-n_{dt,\ell}^-)^2+n_{dt,\ell} \right)\hat c_{\phi, dt}=0.
\end{align*}
Solving for $\hat c_{\phi, dt}$ from this equation, we get
\begin{align}\label{c.opt}
	\hat c_{\phi, dt}=\frac{\bigg( \hat B_{dt}+\big( 1-\frac{n_{dt}}{N_{dt}}\big)\frac{\sigma_{\theta_dt}}{n_{dt}}\underset{j\in \Lambda_{dt,\ell}}\sum \hat u_{\psi,dtj}\bigg)(n_{dt,\ell}^--n_{dt,\ell}^+)}
	{\Big( 1-\frac{n_{dt}}{N_{dt}}\Big)\frac{\sigma_{\theta_dt}}{n_{dt}}\Big((n_{dt,\ell}^+-n_{dt,\ell}^-)^2+n_{dt,\ell} \Big)}.
\end{align}
Since $\hat c_{\phi, dt}\in (|\hat u_{\psi,dtj}|_{(\ell)}, |\hat u_{\psi, dtj}|_{(\ell+1)})$, $\hat c_{\phi, dt} >0$ and
the numerator of \eqref{c.opt} is strictly positive.
Further, there are two possibilities: (i) $n_{dt,\ell}^+< n_{dt,\ell}^-$ and $\hat B_{dt}+\left( 1-\frac{n_{dt}}{N_{dt}}\right)\frac{\sigma_{\theta_dt}}{n_{dt}}\underset{j\in \Lambda_{dt,\ell}}\sum \hat u_{\psi,dtj}>0$; (ii) $n_{dt,\ell}^+> n_{dt,\ell}^-$ and $\hat B_{dt}+\left( 1-\frac{n_{dt}}{N_{dt}}\right)\frac{\sigma_{\theta_dt}}{n_{dt}}\underset{j\in \Lambda_{dt,\ell}}\sum \hat u_{\psi,dtj}<0$.
Therefore, if the number of negative outliers, $ n_{dt,\ell}^-$, is greater than the number of positive outliers, $n_{dt,\ell}^+$, then the bias is positive.
Otherwise, the opposite applies.
\end{proof}

Theorem \ref{exis.uni} provides a local, area-time specific approach, that calculates the robustness parameter that best bounds the outlier observations in each subdomain to reduce the predictive bias, without detriment to the MSE. Finally, it allows us to intuit the relationship between the sign of the bias and the number of large positive and negative residuals. Selecting the value of $\hat c_{\phi,dt}$ not only avoids subjective choices, but also reveals the atypical condition of a subdomain.

\section{Model-based simulations}\label{suppl.sim}

\setcounter{equation}{0}
\renewcommand{\theequation}{\thesection.\arabic{equation}}

\setcounter{table}{0}
\renewcommand{\thetable}{\thesection.\arabic{table}}

\setcounter{figure}{0}
\renewcommand{\thefigure}{\thesection.\arabic{figure}}

The target of this section is to report additional results from the simulation experiments presented in Section \ref{sec.sim} of the manuscript. The outline of the experimental design is described below, including the scenarios for the incorporation of unit-level and area-level outliers, the cases for the time dependency random effects and the number of iterations, $S=500$. Both the MQ3 and TWMQ models are fitted using the iterative re-weighted least squares (IRLS) algorithm. In addition, the projective influence function $\psi$ is the Huber function \eqref{psiHuber} with tuning constant $c_\psi = 1.345$, the same as the function $\phi$ of the robust bias corrected MQ predictors BMQ and BTMQ, but their tuning constant has been selected area-time specific.
So as to calculate $\hat c^{(s)}_{\phi, dt}$, $s=1,\dots, S$, we use a fine grid from 0 to 10, with evenly spaced breaks of 0.001 width.
For the fitting of the MQ models, the prediction of small area linear indicators, the estimation of the MSE and the selection of the robustness parameters, we have used a code developed by the authors. In addition, LMMs are fitted using residual maximum likelihood (REML) so as to them calculate empirical best linear unbiased predictors (EBLUP) of the population means. The \texttt{R} library \texttt{nlme} has been used for this purpose and, in particular, the function \texttt{lme}.

Simulations 1--5 have the following steps:
\begin{enumerate}
\item Define $\beta_1=100$ and $\beta_2=5$. Vary $q$ on a fine grid $G\subset [0,1]$.

\item[$\rightarrow$] Choose Scenario [0,0], [e,0] or [e,u], where
\item[] [0,0] -- absence of outliers, $u_{1,d}\backsim N(0,3)$ and $e_{dtj}\backsim N(0,6)$;
\item[] [e,0] -- only individual level outliers, $u_{1,d}\backsim N(0,3)$ and $e_{dtj}\backsim\delta N(0,6)+(1-\delta)N(20,150)$, where $\delta$    is an independently generated Bernoulli random variable with $P(\delta=1)=0.97$;
\item[] [e,u] -- outliers affect both area and individual effects, $u_{1,d}\backsim N(0,3)$ for areas $1\leq d\leq 36$, $u_{1,d}\backsim N(9,20)$ for areas $37\leq d\leq 40$, and $e_{dtj}\backsim\delta N(0,6)+(1-\delta)N(20,150)$.

\item[$\rightarrow$] Choose Case 1.1, 1.2 or 2, where
\item[] Case 1. $\uu_2=\underset{1\leq t\leq T}{\mbox{col}}(u_{2,t})\backsim N_T(\bm 0, \Sigma_u)$, where $\Sigma_u=\sigma_u^2\Omega_T(\rho)$ and the correlation matrix is
\begin{align}\label{rho}
	\Omega_T(\rho)=\frac{1}{1-\rho^2}
	\renewcommand\arraystretch{0.6}\arraycolsep=1.2pt
	\begin{pmatrix}
		1 & \rho & \cdots&\rho^{T-1}\\
		\rho&1&\ddots&\rho^{T-2}\\
		\vdots&\ddots&\ddots&\vdots\\
		\rho^{T-1}&\rho^{T-2}&\rho&1
	\end{pmatrix}\in\mathcal M_{T\times T}, \quad \rho\in(-1,1).
\end{align}
Case 1.1: $\sigma_u=1$, \ $\rho=0.2$; \ and Case 1.2: $\sigma_u=1$, \ $\rho=0.8$.
\item[] Case 2. Each $u_{2,t}$ is independently generated according to a stationary $AR(3)$ model with coefficients $\phi_1=0.4$, $\phi_2=0.3$, $\phi_3=0.25$ and white noise variance $\sigma=1$.

\item Repeat $S=500$ times $(s=1,\dots,S)$:
\begin{enumerate}
	\item For $d=1\dots, D$, $t=1\dots, T$, $j\in U_{dt}$, generate $x_{dtj}^{(s)}\backsim \text{LogN}(1,0.5)$, $u_{1,d}^{(s)}$ and $e_{dtj}^{(s)}$ depending on the chosen scenario, and $u_{2,t}^{(s)}$ depending on the chosen case.
	
	\item For $d=1\dots, D$, $t=1\dots, T$, $j\in U_{dt}$, calculate
	\[y_{dtj}^{(s)}=\beta_1+x_{dtj}^{(s)}\beta_2+u_{1,d}^{(s)}+u_{2,t}^{(s)}+e_{dtj}^{(s)}, \quad \overline Y_{dt}^{(s)}=\frac{1}{N_{dt}}\underset{j=1}{\overset{N_{dt}}\sum}
	y_{dtj}^{(s)}. \]
	
	\item  Fit the MQ3 models using the population data. Calculate $q_{dtj}^{(s)}$ and then $\theta_d^{(s)}$,  $d=1\dots, D$, $t=1\dots, T$, $j\in U_{dt}$. Use the IRLS algorithm.
	
	\item Randomly generating $n_{dt}$ different positions between 1 and $N_{dt}$, draw a sample $s_{dt}^{(s)}$ of size $n_{dt}$, $d=1\dots, D$, $t=1\dots, T$. In what follows, only sample data are used.
	
	\item Calculate the HÃ¡jek estimator with equal weights, i.e. the arithmetic mean:
	\[\hat{\overline{Y}}_{dt}^{hajek}=\frac{1}{n_{dt}}\sum_{j=1}^{n_{dt}} y_{dtj}.\]
	
	\item Using REML, fit the area-level $\text{LMM}_1$ model
	\begin{align*}
		\nonumber
		y_{dtj}=\beta_{1}+x_{dtj}\beta_{2}+u_{1,d}+e_{dtj}, \ u_{1,d}\backsim N(0,\sigma_{u_1}^2), \ e_{dtj}\backsim N(0, \sigma^2_{e}), \ \ \sigma_{u_1}^2, \sigma^2_{e}>0,
	\end{align*}
	and the area-level and time-level $\text{LMM}_2$ model
	\begin{align*}
		\nonumber
		y_{dtj}=\beta_{1}+x_{dtj}\beta_{2}+u_{1,d}+u_{2,t}+e_{dtj}, \ &u_{1,d}\backsim N(0,\sigma_{u_{1}}^2), \ u_{2,t}\backsim N(0,\sigma_{u_{2}}^2), \quad\quad\quad \\ &e_{dtj}\backsim N(0, \sigma^2_{e}),
		\ \sigma_{u_{1}}^2, \sigma_{u_{2}}^2, \sigma^2_{e}>0,
	\end{align*}
	where $\beta_{1}$ and $\beta_{2}$ are the corresponding regression coefficients; $u_{1,d}$ are the area-level random intercepts of $\text{LMM}_1$; $u_{1,d}$ and $u_{2,t}$ are the area-level and time-level random intercepts of $\text{LMM}_2$, respectively; and $e_{dtj}$ are the corresponding model errors.
	
	\item Calculate the predictors $\hat{\overline{Y}}_{dt}^{eblup_1}$ and $\hat{\overline{Y}}_{dt}^{eblup_2}$, $d=1,\dots, D$, $t=1,\dots, T$, given by
	\begin{align*}
		& \hat{\overline{Y}}_{dt}^{eblup_1}=\frac{1}{N_{dt}}\Big\{\sum_{j\in s_{dt}^{(s)}}y_{dtj}+\sum_{j\in r_{dt}^{(s)}}(\hat\beta_{1}+x_{dtj}\hat\beta_{2}+\hat u_{1,d})\Big\}\\&
		\hat{\overline{Y}}_{dt}^{eblup_2}=\frac{1}{N_{dt}}\Big\{\sum_{j\in s_{dt}^{(s)}}y_{dtj}+\sum_{j\in r_{dt}^{(s)}}(\hat\beta_{1}+x_{dtj}\hat\beta_{2}+\hat u_{2,d}+\hat u_{2,t})\Big\},
	\end{align*}
	where $r_{dt}^{(s)}=U_{dt}-s_{dt}^{(s)}$ is the non sampled subset of $U_{dt}$; $\hat u_{1,d}$ is the EBLUP of the random intercept $u_{1,d}$ for $\text{LMM}_1$; and $\hat u_{1,d}$ and $\hat u_{2,t}$ are the EBLUPs of the random intercepts $u_{1,d}$ and $u_{2,t}$ for $\text{LMM}_2$, respectively.
	
	\item Fit the MQ3 model, i.e. estimate $\hat\bbeta^{(s)}_\psi(q)$ and $\hat\sigma_{q}^{(s)}=\hat\sigma_{\psi}(\hat\bbeta^{(s)}_\psi(q))$, with $q\in G$.
	
	\item For $d=1\dots, D$, $t=1\dots, T$, $j\in s_{dt}^{(s)}$, estimate $\hat q_{dtj}^{(s)}$ and then $\hat\theta_d^{(s)}$.
	
	\item Calculate the predictors $\hat{\overline{Y}}_{dt}^{mq}$ and $\hat{\overline{Y}}_{dt}^{bmq}$, $d=1,\dots, D$, $t=1,\dots, T$.
	
	\item Calculate the inter-period weights $\ww_t^{(s)}=(w_{t1}^{(s)},\dots, w_{tT}^{(s)})$, $t=1,\dots, T$.
	
	\item For $t=1,\dots, T$, $d=1,\dots, D$, fit the TWMQ models with $q=\hat\theta_d^{(s)}$, i.e. estimate $\hat\bbeta^{(s)}_\psi\big(\hat\theta_d^{(s)},\ww_t^{(s)}\big)$ and $\hat\sigma_{\hat\theta_d t}^{(s)}=\hat\sigma_{\psi}(\hat\bbeta^{(s)}_\psi\big(\hat\theta_d^{(s)},\ww_t^{(s)}\big))$.
	
	\item Define  $\overline{\hat\bbeta}^{(s)}_\psi=(\overline{\hat\beta}^{(s)}_{\psi 1}, \overline{\hat\beta}^{(s)}_{\psi 2})=\frac{1}{DT}\underset{d=1}{\overset{D}\sum}\underset{t=1}{\overset{T}\sum}\hat\beta^{(s)}_\psi\big(\hat\theta_d^{(s)},\ww_t^{(s)}\big)$.
	
	\item Calculate the predictors $\hat{\overline{Y}}_{dt}^{tmq}$ and $\hat{\overline{Y}}_{dt}^{btmq}$, $d=1,\dots, D$, $t=1,\dots, T$.
	
	\item For the TMQ predictor, calculate $rmse^{tmq}_{dt}\in\{rmse_{11,dt}^{tmq}, rmse_{12,dt}^{tmq}, rmse_{21,dt}^{tmq},rmse_{22,dt}^{tmq}\}$ and for the BTMQ predictor, calculate  $rmse^{btmq}_{dt}\in\{rmse_{1,dt}^{btmq}, rmse_{2,dt}^{btmq}, rmse_{3,dt}^{btmq}\}$, $d=1,\dots, D$, $t=1,\dots, T$.
\end{enumerate}

\item For $\hat{\overline{Y}}_{dt}\in\big\{\hat{\overline{Y}}_{dt}^{hajek},\hat{\overline{Y}}_{dt}^{eblup_1},\hat{\overline{Y}}_{dt}^{eblup_2}, \hat{\overline{Y}}_{dt}^{mq},\hat{\overline{Y}}_{dt}^{bmq}, \hat{\overline{Y}}_{dt}^{tmq},\hat{\overline{Y}}_{dt}^{btmq}\big\}$, $d=1,\dots, D$, $t=1,\dots, T$, $\hat\tau_1\in\big\{\hat\beta_{1}, \overline{\hat\beta}_{\psi 1}\big\}$,
$\hat\tau_2\in\big\{\hat\beta_{2},  \overline{\hat\beta}_{\psi 2}\big\}$ and $\hat\theta_d$, calculate
\begin{align}
	\nonumber
	&\text{BIAS}(\hat \tau_l)=\frac{1}{S}\sum_{s=1}^S(\hat \tau_l^{(s)}-\beta_l), \quad \text{RMSE}(\hat \tau_l)=\Big( \frac{1}{S}\sum_{s=1}^S(\hat \tau_l^{(s)}-\beta_l)^2\Big)^{1/2}, \quad l=1,2,\\&
	\nonumber
	\text{BIAS}_d=\frac{1}{S}\sum_{s=1}^S(\hat \theta_d^{(s)}-\theta_d^{(s)}), \quad \text{RMSE}_d=\Big( \frac{1}{S}\sum_{s=1}^S(\hat \theta_d^{(s)}-\theta_d^{(s)})^2\Big)^{1/2},
	\\&
	\label{RBIAS.RMSE}
	\text{BIAS}_{1,dt}=\frac{1}{S}\underset{s=1}{\overset{S}{\sum}}\big(\hat {\overline Y}_{dt}^{(s)}- {\overline Y}_{dt}^{(s)} \big),\quad
	\text{RMSE}_{1,dt}=\Big(\frac{1}{S}\underset{s=1}{\overset{S}{\sum}}\big(\hat {\overline Y}_{dt}^{(s)}- {\overline Y}_{dt}^{(s)} \big)^2\Big)^{1/2}.
\end{align}

For $rmse_{dt}\in\{rmse^{tmq}_{dt},rmse^{btmq}_{dt}\}$, $d=1,\dots, D$, $t=1,\dots, T$, calculate
\begin{align*}
	\text{BIAS}_{2,dt}=\frac{1}{S}\overset{S}{\underset{s=1}{\sum}}(rmse_{dt}^{(s)}-\text{RMSE}_{1,dt}),\
	\text{RMSE}_{2,dt}= \Big(\frac{1}{S}\overset{S}{\underset{s=1}{\sum}}\big(rmse_{dt}^{(s)}-\text{RMSE}_{1,dt}\big) ^2\Big)^{1/2},
\end{align*}
where $\text{RMSE}_{1,dt}$ is taken from \eqref{RBIAS.RMSE} for $\hat{\overline{Y}}_{dt}\in\big\{\hat{\overline{Y}}_{dt}^{tmq},\hat{\overline{Y}}_{dt}^{btmq}\big\}$.

Consistent with the later notation, write $\text{RMSE}_{1,dt}\in\big\{\text{RMSE}_{dt}^{tmq}, \text{RMSE}_{dt}^{btmq}\big\}$, $d=1,\dots, D$, $t=1,\dots, T$, and $\text{RMSE}_{1}=\frac{1}{DT}\underset{d=1}{\overset{D}{\sum}}\underset{t=1}{\overset{T}{\sum}} \text{RMSE}_{1,dt}\in\big\{\text{RMSE}^{tmq}, \text{RMSE}^{btmq}\big\}$.

\item For $d=1,\dots, D$, $t=1,\dots, T$, calculate the relative performance measures
\begin{align*}
	&\text{RBIAS}(\hat \tau_l)=\frac{100\cdot\text{BIAS}(\hat \tau_l)}{\beta_l}, \quad \text{RRMSE}(\hat \tau_l)=\frac{100\cdot\text{RMSE}(\hat \tau_l)}{\beta_l}, \quad l=1,2,\\&
	\nonumber
	\text{RBIAS}_d=\frac{100\cdot\text{BIAS}_d}{\theta_d^*}, \quad \text{RRMSE}_d=\frac{100\cdot\text{RMSE}_d}{\theta_d^*}, \quad \theta_d^*=\frac{1}{S}\underset{s=1}{\overset{S}{\sum}} \theta_d^{(s)},\\&
	\text{RBIAS}_{1,dt}=\frac{100\cdot\text{BIAS}_{1,dt}}{\overline{Y}_{dt}^*},\quad
	\text{RRMSE}_{1,dt}=\frac{100\cdot\text{RMSE}_{1,dt}}{\overline Y_{dt}^*}, \quad \overline{Y}_{dt}^*=\frac{1}{S}\underset{s=1}{\overset{S}{\sum}}\overline Y_{dt}^{(s)},\\&
	\text{RBIAS}_{2,dt}=\frac{100\cdot\text{BIAS}_{2,dt}}{\text{RMSE}_{1,dt}}, \quad \text{RRMSE}_{2,dt}=\frac{100\cdot\text{RMSE}_{2,dt}}{\text{RMSE}_{1,dt}},
\end{align*}
and the average relative performance measures
\begin{align*}
	& \text{ARBIAS}=\frac{1}{D}\sum_{d=1}^D|\text{RBIAS}_d|, \quad \text{RRMSE}=\frac{1}{D}\sum_{d=1}^D \text{RRMSE}_d,
	\label{PerMea}
	\\&\text{ARBIAS}_l=\frac{1}{DT}\sum_{d=1}^D\sum_{t=1}^T|\text{RBIAS}_{l,dt}|, \ \text{RRMSE}_l=\frac{1}{DT}\sum_{d=1}^D\sum_{t=1}^T\text{RRMSE}_{l,dt}, \ l=1,2.
\end{align*}
\end{enumerate}

Finally, to measure overestimates (underestimates) of the several methods of MSE estimation, we define the proportion of subdomains in which the proposed estimates are higher (lower) than the empirical values which, we recall again, are obtained as the output of Simulation in \eqref{RBIAS.RMSE}. In accordance with the above, let be
\[P_+=\frac{1}{DT}\sum_{d=1}^D\sum_{t=1}I(\text{BIAS}_{2,dt}\geq0),\quad P_-=\frac{1}{DT}\sum_{d=1}^D\sum_{t=1}I(\text{BIAS}_{2,dt}<0)=DT-P_+. \]

\subsection{Fitting algorithm and estimation of domain means of MQ coefficients}\label{sec.sim1suple}

Here we show the performance of the fitting algorithm and the estimation of domain means of unit-level MQ coefficients for MQ3 models. 

Table \ref{ModelSimbeta0} presents the relative performance measures (in \%) for the REML estimators of the model parameters for the $\text{LMM}_2$, which is an LMM with independent area-level and time-level random intercepts, and the mean IRLS estimators for the TWMQ models. In general terms, it can be observed that both fitting algorithms perform quite well. In the classical setting, with no atypical data, the results are slightly better for the $\text{LMM}_2$. Otherwise, the estimation is less biased and more accurate by fitting the TWMQ models. The differences between cases for the generation of time-level random effects are of minor importance.

\begin{table}[!h]
\centering
\renewcommand{\arraystretch}{0.85}
\begin{tabular}{|l|rr|rr|rr|}
	\cline{2-7}
	\multicolumn{1}{c|}{}&  \multicolumn{2}{c}{ [0,0] \quad 1--40} &  \multicolumn{2}{|c}{ [e,0] \quad 1--40} & \multicolumn{2}{|c|}{[e,u] \quad 1--40} \\
	\multicolumn{1}{c|}{}&RBIAS&RRMSE&RBIAS&RRMSE&RBIAS&RRMSE\\
	\cline{2-7}
	\multicolumn{7}{l}{Case 1.1  $\quad\quad\quad  \uu_2\backsim N_T(\bm 0, \Sigma_u)$: $\quad \Sigma_u=\sigma^2_u\Omega_T(\rho)$,  $\quad\sigma_u=1$, $\quad \rho=0.2$} \\[2pt]
	\hline
	$\hat\beta_1: \ $ REML&{ -0.006}&{ 0.477}&0.600&0.798&1.507&1.609\\
	$\quad\quad $ IRLS&-0.015&0.491&{ 0.083}&{ 0.497}&{ 1.003}&{\ 1.144}\\
	$\hat\beta_2: \ $ REML&{-0.005}&{ 0.710}&-0.130&1.322&-0.128&1.320\\
	$\quad\quad $ IRLS&-0.027&0.861&{ -0.065}&{ 0.886}&{ -0.081}&{ 1.193}\\
	\hline
	\multicolumn{7}{l}{Case 1.2  $\quad\quad\quad  \uu_2\backsim N_T(\bm 0, \Sigma_u)$: $\quad \Sigma_u=\sigma^2_u\Omega_T(\rho)$,  $\quad\sigma_u=1$, $\quad \rho=0.8$} \\[2pt]
	\hline
	$\hat\beta_1: \ $ REML&{ 0.002}&{1.257}&0.618&1.415&1.525&2.002\\
	$\quad\quad $ IRLS&0.012&1.260&{ 0.102}&{ 1.264}&{ 1.019}&{ 1.648}\\
	$\hat\beta_2: \ $ REML&{ -0.008}&{ 0.714}&-0.135&1.329&-0.133&1.327\\
	$\quad\quad $ IRLS&-0.028&0.852&{ -0.068}&{ 0.887}&{ -0.077}&{ 1.195}\\
	\hline
	\multicolumn{7}{l}{Case 2  $\quad\quad\quad \ \ u_{2,t}\backsim AR(3)$: $\ \ \phi_1=0.4$, $\ \  \phi_2=0.3$, $\ \ \phi_3=0.25$, $\quad  \sigma=1$} \\[2pt]
	\hline
	$\hat\beta_1: \ $ REML&{-0.139}&{2.290}&0.464&2.326&1.373&2.662\\
	$\quad\quad $ IRLS&-0.144&2.288&-0.053&2.280&0.868&2.447\\
	$\hat\beta_2: \ $ REML&-0.037&{0.694}&0.044&1.281&0.045&1.279\\
	$\quad\quad $ IRLS&{ -0.026}&0.867&0.010&0.910&0.019&1.221\\
	\hline
\end{tabular}
\caption{RBIAS and RRMSE values (in \%) for the estimation of $\beta=(\beta_1,\beta_2)=(100,5)$.}\label{ModelSimbeta0}
\end{table}


\begin{table}[!h]
\centering
\renewcommand{\arraystretch}{0.875}
\begin{tabular}{|>{\small}l|rr|rr|rr|}
	\cline{2-7}
	\multicolumn{1}{c|}{}&  \multicolumn{2}{c}{ [0,0] \quad 1--40} &  \multicolumn{2}{|c}{ [e,0] \quad 1--40} & \multicolumn{2}{|c|}{[e,u] \quad 1--40} \\
	\multicolumn{1}{c|}{}&RBIAS&RRMSE&RBIAS&RRMSE&RBIAS&RRMSE\\
	\cline{2-7}
	\hline
	Case 1.1&0.231&7.506&0.281&7.603&0.277&7.394\\
	Case 1.2&0.240&7.401&0.271&7.514&0.274& 7.315\\
	Case 2&0.252&7.492&0.245&7.569&0.237&7.345\\
	\hline
\end{tabular}
\caption{RBIAS and RRMSE values (in \%) for the estimation of $\theta_d$.}\label{ModelSimthetad}
\end{table}

Table \ref{ModelSimthetad} presents the relative performance measures (in \%) for the estimation of domain means of unit-level MQ coefficients for MQ3 models. The most important point to note is the similarity of the results between the different scenarios. Consequently, we can deduce that the estimation of $\theta_d$ is not highly affected by the presence of atypical data, neither at individual nor at area level. As for the difference between cases for the generation of time-level random effects, the same applies. In terms of numbers, the ARBIAS is around 0.2--0.3\% and the RRMSE rises to 7.0--8.0\%, giving more than acceptable results for both relative error and bias.

\subsection{MSE estimation}\label{sec.sim2}

In this section we investigate the performance of several methods of MSE estimation for the TMQ and BTMQ predictors. Table \ref{ModelSim2} presents the performance measures for Case 1.1, Case 1.2 and Case 2, and the different scenarios for the generation of atypical data. We have included a third column with the proportion of subdomains in which the bias is positive and, therefore, the RMSE overestimated.

\begin{table}[H]
\centering
\renewcommand{\arraystretch}{0.65}
\setlength{\tabcolsep}{1.5pt}
\begin{tabular}{|l|ccc|ccc|ccc|}
	\cline{2-10}
	\multicolumn{1}{c|}{}&  \multicolumn{3}{c}{ [0,0] \quad 1--40} &  \multicolumn{3}{|c}{ [e,0] \quad 1--40} & \multicolumn{3}{|c|}{[e,u] \quad 1--40} \\
	\multicolumn{1}{c|}{}&ARBIAS&RRMSE& $P_+$&ARBIAS&RRMSE&$P_+$&ARBIAS&RRMSE&$P_+$\\
	\cline{2-10}
	\multicolumn{10}{l}{$\quad$ Case 1.1  $\quad\quad\quad \uu_2\backsim N_T(\bm 0, \Sigma_u)$: $\quad \Sigma_u=\sigma^2_u\Omega_T(\rho)$,  $\quad\sigma_u=1$, $\quad \rho=0.2$} \\[2pt]
	\hline
	\makecell[l]{$\text{RMSE}^{tmq}$}&  \multicolumn{3}{c|}{0.811}&\multicolumn{3}{c|}{1.002}&\multicolumn{3}{c|}{1.040}\\
	\hline
	$rmse_{11,dt}^{tmq}$&6.624&57.506&0.97&8.448&47.668&0.00&18.129&64.743&0.21\\
	$rmse_{12,dt}^{tmq}$&5.281& 57.301&0.89&9.714&47.928&0.00&19.499&64.362&0.13\\
	$rmse_{21,dt}^{tmq}$&6.977&55.923&0.98&8.036&46.331&0.00&18.064&64.639&0.25\\
	$rmse_{22,dt}^{tmq}$&7.006&55.881&0.98&7.745&45.820&0.00&17.870&64.224&0.27\\
	\hline
	\makecell[l]{$\text{RMSE}^{btmq}$}&\multicolumn{3}{c|}{0.638}&\multicolumn{3}{c|}{0.871}&\multicolumn{3}{c|}{0.933}\\
	\hline
	$rmse_{1,dt}^{btmq}$&4.119&54.558& 0.70&18.203&44.658&0.00&35.200&68.035&0.10\\
	$rmse_{2,dt}^{btmq}$&4.145& 54.462&0.72&17.465&43.312&0.00&34.504&66.770&0.10\\
	$rmse_{3,dt}^{btmq}$&6.467&59.178&0.14&15.335&61.613&0.00&22.095&91.702&0.10\\
	\hline
	\multicolumn{10}{l}{$\quad$ Case 1.2  $\quad\quad\quad\uu_2\backsim N_T(\bm 0, \Sigma_u)$: $\quad \Sigma_u=\sigma^2_u\Omega_T(\rho)$,  $\quad\sigma_u=1$, $\quad \rho=0.8$} \\[2pt]
	\hline
	\makecell[l]{$\text{RMSE}^{tmq}$}&  \multicolumn{3}{c|}{0.800}&\multicolumn{3}{c|}{1.017}&\multicolumn{3}{c|}{1.056}\\
	\hline
	$rmse_{11,dt}^{tmq}$&6.007&57.181&0.97&10.660&46.763&0.00&19.548&63.270&0.12\\
	$rmse_{12,dt}^{tmq}$&4.648&57.016&0.92&11.825&47.067&0.00&20.994&63.044&0.11\\
	$rmse_{21,dt}^{tmq}$&6.406&55.626&0.98&10.254&45.521&0.00&19.383&63.263&0.14\\
	$rmse_{22,dt}^{tmq}$&6.438&55.579&0.98&9.950&44.980&0.00&19.128&62.821&0.14\\
	\hline
	\makecell[l]{$\text{RMSE}^{btmq}$}&\multicolumn{3}{c|}{0.620}&\multicolumn{3}{c|}{0.880}&\multicolumn{3}{c|}{0.940}\\
	\hline
	$rmse_{1,dt}^{btmq}$&3.862&54.707&0.72&19.912&44.300&0.00&36.336&67.695&0.10\\
	$rmse_{2,dt}^{btmq}$&3.892&54.599&0.73&19.161&42.918&0.00&35.631&66.408&0.10\\
	$rmse_{3,dt}^{btmq}$&5.620&68.705&0.17&16.238&60.570&0.00&24.483&91.384&0.10\\
	\hline
	\multicolumn{10}{l}{$\quad$ Case 2  $\quad\quad\quad\  u_{2,t}\backsim AR(3)$: $\quad  \phi_1=0.4$, $\quad  \phi_2=0.3$, $\quad  \phi_3=0.25$, $\quad  \sigma=1$} \\[2pt]
	\hline
	\makecell[l]{$\text{RMSE}^{tmq}$}&  \multicolumn{3}{c|}{0.802 }&\multicolumn{3}{c|}{1.012}&\multicolumn{3}{c|}{ 1.049}\\
	\hline
	$rmse_{11,dt}^{tmq}$&7.335&58.195&0.97&9.572&46.851&0.00&19.109&64.197&0.18\\
	$rmse_{12,dt}^{tmq}$&5.982&57.982&0.94&10.782&47.142&0.00&20.474&63.888&0.11\\
	$rmse_{21,dt}^{tmq}$&7.696&56.635&0.97&9.170&45.454&0.00&18.993&63.946&0.21\\
	$rmse_{22,dt}^{tmq}$&7.726&56.592&0.97&8.867&44.918&0.00&18.786&63.510&0.22\\
	\hline
	\makecell[l]{$\text{RMSE}^{btmq}$}&\multicolumn{3}{c|}{0.630 }&\multicolumn{3}{c|}{0.878 }&\multicolumn{3}{c|}{0.938 }\\
	\hline
	$rmse_{1,dt}^{btmq}$&4.668&55.316&0.80&18.947&44.205&0.00&35.512&67.396&0.10\\
	$rmse_{2,dt}^{btmq}$&4.715&55.216&0.81&18.185&42.809&0.00&34.795&66.093&0.10\\
	$rmse_{3,dt}^{btmq}$&6.213&69.447&0.15&15.769&58.991&0.00&22.402&72.775&0.10\\
	\hline
\end{tabular}
\caption{By rows, empirical RMSE for each case and scenario. By columns, ARBIAS and RRMSE (in \%) for the corresponding RMSE estimators.}\label{ModelSim2}
\end{table}

As a first general comment, estimating RMSEs is much more difficult than estimating small area linear indicators, such as population means. Therefore, the magnitude of the results in Table \ref{ModelSim2} should be assessed with caution. Although it is suggested that the RMSE estimators for the TMQ predictors offer the most balanced performance in terms of ARBIAS and RRMSE, the empirical RMSEs of the BTMQ predictors are smaller. It can be observed that 
the difference between cases for the generation of time effects is of small importance.
In terms of RRMSE, the presence of area-level outliers greatly worsens the average values in Table \ref{ModelSim2}.

\subsection{Role of the time component variability and temporal correlation}\label{sec.sim3}
Here we analyze how the values of $\sigma_u>0$ and $\rho\in (-1,1)$ in Case 1 determine the behaviour of the RMSE estimators and, not least importantly, that of the TMQ and BTMQ predictors themselves.
Recall that in Case 1 the time dependency effects are generated as
$\uu_2=\underset{1\leq t\leq T}{\mbox{col}}(u_{2,t})\backsim N_T(\bm 0, \Sigma_u)$, where $\Sigma_u=\sigma_u^2\Omega_T(\rho)$. Our aim is: (1) to study the role of the time component variability, and (2) to investigate how temporal correlation affects. In order to do so, we set $\sigma_u=1$ and vary the correlation $\rho\in\{0, 0.2, 0.35, 0.5, 0.75, 0.9, 0.95\}$; and we set $\rho=0.5$ and vary the standard deviation $\sigma_u\in\{0.25, 0.5, 0.75, 1, 1.5, 2, 2.5\}$. Given the performance of the RMSE estimators in Table \ref{ModelSim2}, we focus the analysis on $rmse_{22,dt}^{tmq}$ and $rmse_{2,dt}^{btmq}$, respectively. The reason is that they seem to be slightly better in Case 1 or, at least, not worse than the other alternatives. In addition, we have limited ourselves to Scenario [0,0], without the presence of outliers, because we are interested in understanding what happens in a conventional setting.

Table \ref{ModelSim3} reports the results for the TMQ predictor and Table \ref{ModelSim3.Supple} for the BTMQ predictor. The empirical RMSE is included because the relative measures, although comparable, refer to quantities with slightly different scale. Namely, in all our settings $\text{RMSE}^{tmq}>\text{RMSE}^{btmq}$.

\begin{table}[H]
\centering
\renewcommand{\arraystretch}{1}
\setlength{\tabcolsep}{1.5pt}
\begin{tabular}{|l|ccccccc|}
\multicolumn{8}{c}{Case 1. Scenario [0,0]}\\
\cline{2-8}
\multicolumn{1}{c}{{$(\sigma_u=1,\ \rho)$}}&   \multicolumn{1}{|c}{(1, 0.0)}& \multicolumn{1}{c}{(1, 0.2)}&
\multicolumn{1}{c}{(1, 0.35)}& \multicolumn{1}{c}{(1, 0.5)}&
\multicolumn{1}{c}{(1, 0.75)}& \multicolumn{1}{c}{(1, 0.9)}&\multicolumn{1}{c|}{(1, 0.95)}\\
\hline
\makecell[l]{$\text{RMSE}^{tmq}$} &0.807&0.811&0.811&0.804&0.796&0.806&0.794\\
\hline
ARBIAS&7.649&7.006&6.755&7.028&7.142&5.459&8.330\\
RRMSE&56.082&55.881&55.822&56.066&56.047&54.734&53.159\\
$P_+$&0.99&0.98&0.97&0.99&0.99&0.96&0.99\\
\hline
\multicolumn{8}{c}{ }\\
\cline{2-8}
\multicolumn{1}{c}{{$(\sigma_u,\ \rho=0.5)$}}&
\multicolumn{1}{|c}{(0.25, 0.5)} &   \multicolumn{1}{c}{(0.5, 0.5)}& \multicolumn{1}{c}{(0.75, 0.5)} & \multicolumn{1}{c}{(1, 0.5)}& \multicolumn{1}{c}{(1.5, 0.5)}& \multicolumn{1}{c}{(2, 0.5)}& \multicolumn{1}{c|}{(2.5, 0.5)}\\
\hline
\makecell[l]{$\text{RMSE}^{tmq}$} &0.691&0.721&0.758&0.804&0.935&1.102&1.263\\
\hline
ARBIAS&22.652&18.306&13.133&7.028&7.573&19.054&27.629\\
RRMSE&71.134&66.577&61.416&56.066&46.829&42.775&40.607\\
$P_+$&1&1&1&0.99&0.11&0.00&0.00\\
\hline
\end{tabular}
\caption{Study of the effect of $\sigma_u$ and $\rho$ in $rmse_{22,dt}^{tmq}$ (in \%). }\label{ModelSim3}
\end{table}
\begin{table}[H]
\centering
\renewcommand{\arraystretch}{1}
\setlength{\tabcolsep}{1.5pt}
\begin{tabular}{|l|ccccccc|}
\multicolumn{8}{c}{Case 1. Scenario [0,0]}\\
\cline{2-8}
\multicolumn{1}{c}{{$(\sigma_u=1,\ \rho)$}}&   \multicolumn{1}{|c}{(1, 0.0)}& \multicolumn{1}{c}{(1, 0.2)}&
\multicolumn{1}{c}{(1, 0.35)}& \multicolumn{1}{c}{(1, 0.5)}&
\multicolumn{1}{c}{(1, 0.75)}& \multicolumn{1}{c}{(1, 0.9)}&\multicolumn{1}{c|}{(1, 0.95)}\\
\hline
\makecell[l]{$\text{RMSE}^{btmq}$} &0.634&0.638&0.638&0.630&0.618&0.623&0.622\\
\hline
ARBIAS&4.234&4.145&4.047&4.127&4.246&3.493&4.946\\
RRMSE&54.650&54.462&54.437&54.740&55.004&53.819&55.818\\
$P_+$&0.79&0.72&0.68&0.72&0.81&0.67&0.83\\
\hline
\multicolumn{8}{c}{ }\\
\cline{2-8}
\multicolumn{1}{c}{{$(\sigma_u,\ \rho=0.5)$}}&
\multicolumn{1}{|c}{(0.25, 0.5)} &   \multicolumn{1}{c}{(0.5, 0.5)}& \multicolumn{1}{c}{(0.75, 0.5)} & \multicolumn{1}{c}{(1, 0.5)}& \multicolumn{1}{c}{(1.5, 0.5)}& \multicolumn{1}{c}{(2, 0.5)}& \multicolumn{1}{c|}{(2.5, 0.5)}\\
\hline
\makecell[l]{$\text{RMSE}^{btmq}$} &0.546&0.568&0.596&0.630&0.738&0.888&1.038\\
\hline
ARBIAS&19.230&14.551&9.199&4.127&11.364&23.917&32.844\\
RRMSE&69.091&64.719&59.837&54.740&46.134&43.395&44.477\\
$P_+$&1&1&0.99&0.72&0.08&0.00&0.00\\
\hline
\end{tabular}
\caption{Study of the effect of $\sigma_u$ and $\rho$ in $rmse_{2,dt}^{btmq}$ (in \%).}\label{ModelSim3.Supple}
\end{table}

It can be seen that, once the variance of the time component is set to $\sigma_u=1$, the higher the correlation, the better the results. By the same token, the findings are convincing: with moderate variability, increasing correlation makes the use of the TWMQ models worthwhile. Secondly, it is ensured that the variability of the time component is key in all aspects. As long as the correlation is moderate, say $\rho=0.5$, the ARBIAS results that are set out in Table \ref{ModelSim3} show that, if a large part of the randomness lies in the time component, the error is underestimated. Nevertheless, we should bear in mind that the target of the TWMQ models is not to capture random relationships, but rather well-founded time dependency structures, where correlation is somewhat large and randomness is moderate. What is more, if the time component is purely random, there is no point in borrowing strength from time.

\subsection{Consistency of $\hat\theta_d$}\label{sec.sim4}
In this section we aim to analytically investigate the consistency of the estimators $\hat\theta_d$ of the domain population means of unit-level MQ coefficients $\theta_d$ for MQ3 models, $d=1,\dots,D$. Here the population is only generated once for computational reasons and, for each iteration, the sample set is drawn by simple random sampling without replacement within each subdomain. The study is performed for Case 1.1 and Scenario [0,0], i.e. with uncontaminated data.

We employ the following steps:
\begin{enumerate}
\item For $d=1\dots, D$, $t=1\dots, T$, $j\in U_{dt}$, generate $x_{dtj}\backsim \text{LogN}(1,0.5)$, $u_{1,d}\backsim N(0,3)$, $e_{dtj}\backsim N(0,6)$, $\uu_2=\underset{1\leq t\leq T}{\mbox{col}}\big(u_{2,t}\big)\backsim N_T(\bm 0, \Omega_T(0.2))$, where  $\Omega_T(\rho)$ has been defined in \eqref{rho}, and then calculate $y_{dtj}=\beta_1+x_{dtj}\beta_2+u_{1,d}+u_{2,t}+e_{dtj}$, with $\beta_1=100$ and $\beta_2=5$.

\item  Fit the MQ3 models using the population data. Calculate $q_{dtj}$ and then $\theta_d$, $d=1\dots, D$, $t=1\dots, T$, $j\in U_{dt}$. Use the IRLS algorithm.

\item Repeat $S=500$ times $(s=1,\dots,S)$:
\begin{enumerate}
\item Randomly generating $n_{dt}$ different positions between 1 and $N_{dt}$, draw a sample $s_{dt}^{(s)}$ of size $n_{dt}$, $d=1\dots, D$, $t=1\dots, T$. In what follows, only sample data are used.

\item Fit the MQ3 model using the sample data and the IRLS algorithm.

\item For $d=1\dots, D$, $t=1\dots, T$, $j\in s_{dt}^{(s)}$, estimate $\hat q_{dtj}^{(s)}$ and then $\hat\theta_d^{(s)}$.
\end{enumerate}

\item For $d=1,\dots, D$, $t=1,\dots, T$, calculate the absolute and relative performance measures
\begin{align*}
&\text{BIAS}_d=\frac{1}{S}\sum_{s=1}^S(\hat \theta_d^{(s)}-\theta_d), \ \text{RBIAS}_d=\frac{100\cdot\text{BIAS}_d}{\theta_d}, \ \text{ARBIAS}=\frac{1}{D}\sum_{d=1}^D|\text{RBIAS}_d|,\\&
\text{RMSE}_d=\Big( \frac{1}{S}\sum_{s=1}^S(\hat \theta_d^{(s)}-\theta_d)^2\Big)^{1/2}, \ \text{RRMSE}_d=\frac{100\cdot\text{RMSE}_d}{\theta_d}, \ \text{RRMSE}=\frac{1}{D}\sum_{d=1}^D \text{RRMSE}_d.
\end{align*}
\end{enumerate}

Let $d=1,\dots, D$, $t=1,\dots, T$. The population size of each $U_{dt}$ is set equal to $N_{dt}=100$, so that of $U_d$ is $N_{d}=10\cdot 100=1000$. The sample sizes are set equal to $n_{dt}\in\{1,2,5,10,25,50\}$, and therefore, in the same way, it holds that $n_{d}\in\{10,20,50,100,250,500\}$ and the sample fractions are $\frac{n_{dt}}{N_{dt}}=\frac{n_{d}}{N_{d}}\in\{0.01, 0.02, 0.05, 0.10, 0.25, 0.50\}$. Figure \ref{Simu4} prints boxplots of RBIAS (left) and RRMSEs (right), both in \%, for the previous sample sizes.

\begin{figure}[H]
\centering
\begin{subfigure}{.7\textwidth}
\includegraphics[width=0.46\textwidth]{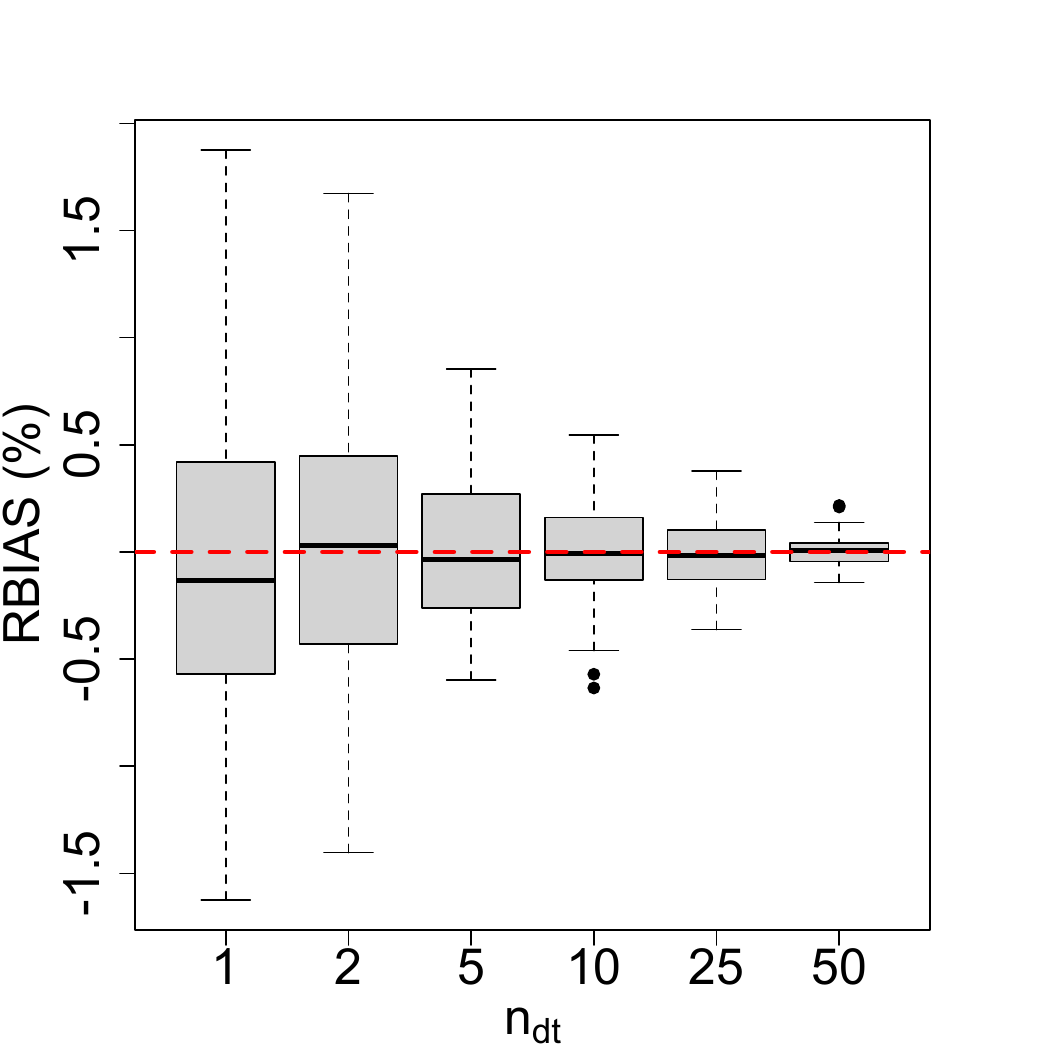}
\includegraphics[width=0.46\textwidth]{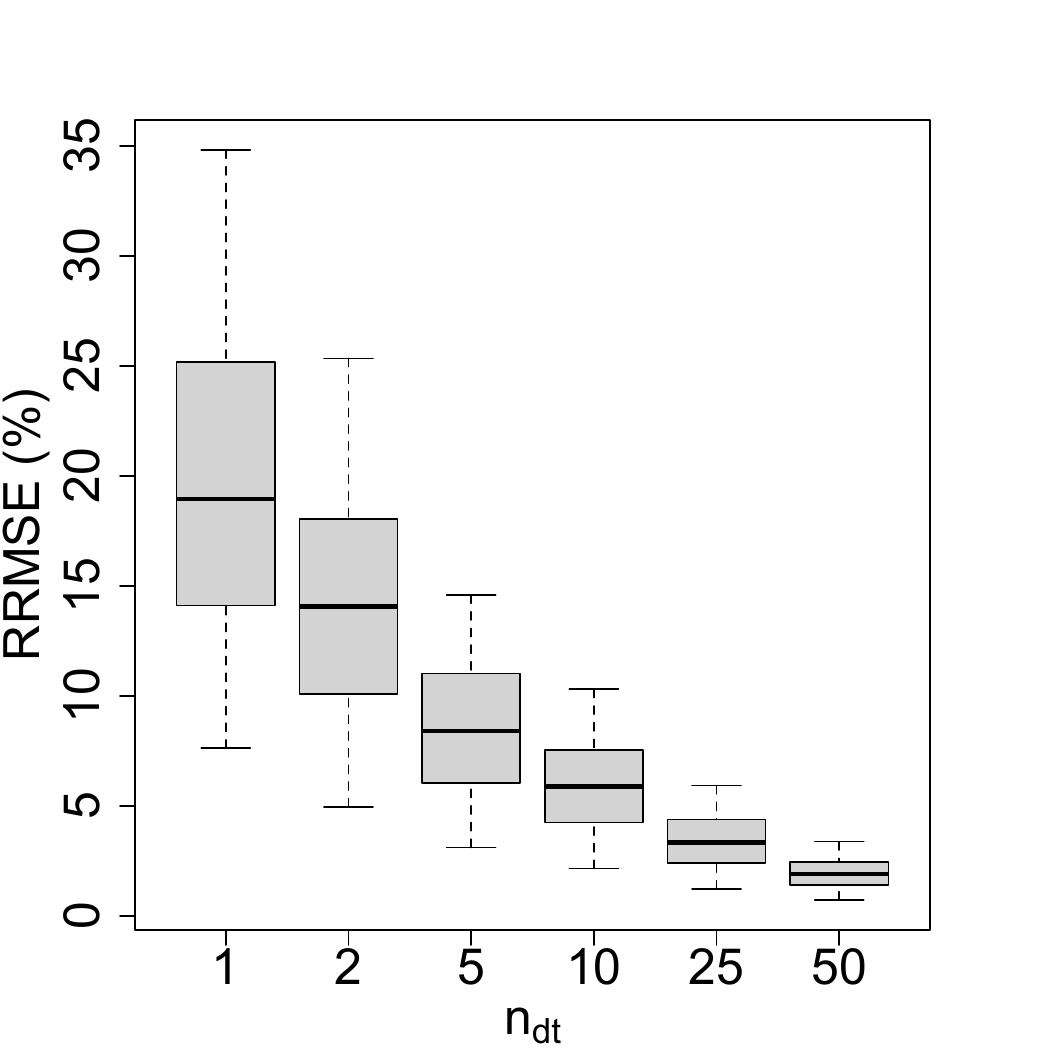}
\end{subfigure}
\caption{Boxplots of RBIAS (left) and RRMSE (right) in \% for the estimation of $\theta_d$.}
\label{Simu4}
\end{figure}

As expected, the RBIAS decreases slightly with increasing sample size, showing an origin-centric behaviour. Nonetheless, what is more remarkable is that the variability is clearly reduced with increasing sample size, leading to a conical pattern in Figure \ref{Simu4} (left). On the other hand, the RRMSE shows a monotonic pattern of decreasing with increasing sample size (Figure \ref{Simu4}, right), with a more than acceptable performance for $n_{dt}\geq 10$ ($n_{d}\geq 100$). Not only does the RBIAS decrease slightly with increasing sample size, but also the variability is clearly reduced. Table \ref{theta} confirms this behavior, with percentage ARBIAS and RRMSEs decreasing as $n_{dt}$  increases, but suggesting some stabilization for $n_{dt}\geq 10$. From this empirical evidence, we conclude that the results for the estimation of $\theta_d$ are satisfactory and the asymptotic consistency is already evident at moderate sample sizes.

\begin{table}[H]
\centering
\renewcommand{\arraystretch}{1.2}
\begin{tabular}{|l|cccccc|}
\multicolumn{1}{c}{}&\multicolumn{6}{c}{$n_{dt}$}\\
\cline{2-7}
\multicolumn{1}{c|}{}&1&2&5&10&25&50\\
\hline
ARBIAS (\%)& 0.609 &0.471 &0.284& 0.195 &0.147 &0.052\\
RRMSE (\%)&19.621 &13.980 & 8.523 & 5.925 & 3.407 & 1.948\\
\hline
\end{tabular}
\caption{ARBIAS and RRMSE (\%) for the estimation of $\theta_d$.}
\label{theta}
\end{table}

\section{Additional analysis for the application to real data}\label{sec.aplic.suple.top}

\setcounter{equation}{0}
\renewcommand{\theequation}{\thesection.\arabic{equation}}

\setcounter{table}{0}
\renewcommand{\thetable}{\thesection.\arabic{table}}

\setcounter{figure}{0}
\renewcommand{\thefigure}{\thesection.\arabic{figure}}

\subsection{Model validation}\label{sec.val.suppe}

Residual analysis is widely used to assess the adequacy of a model by examining the differences between observed and predicted values. Let $d=1\dots, D$, $t=1,\dots, T$. For $j=1,\ldots,n_{dt}$ and $q=\hat\theta_d$, the model residuals are $$\hat{e}_{\psi,dtj}\triangleq y_{dtj}-\xx_{dtj}' \hat\bbeta_\psi(\hat\theta_d,w_t).$$ 
We define the subdomain sample means of model residuals as $$\overline{\hat{e}}_{\psi,dt.}=\frac{1}{n_{d_t}}\sum_{j=1}^{n_{dt}} \hat{e}_{\psi,dtj},$$
the subdomain raw residuals as $$\overline{\hat{e}}_{\psi,dt.}-\overline{\hat{e}}_{\psi,..},$$ where $\overline{\hat{e}}_{\psi,..}=\frac{1}{DT}\sum_{d=1}^D\sum_{t=1}^T \overline{\hat{e}}_{\psi,dt.}$,
and the subdomain standardized residuals (SSR) by dividing by the standard deviation, given by $\nu=\Big(\frac{1}{DT}\sum_{d=1}^D\sum_{t=1}^T(\overline{\hat{e}}_{\psi,dt.}-\overline{\hat{e}}_{\psi,..})^2\Big)^{1/2}$.

Figure \ref{ASR} includes boxplots of the SSRs by province (left) and year (right).
As a result, most of them oscillate around $y=0$ and lie in the interval $(-3,3)$. Not surprisingly, the provincial variability is greater than the annual one, but neither province seems to be particularly poorly modelled.
Outlier detection, based on the selection of area-time specific robustness parameters, is presented in Section \ref{sec.out} of the main document.

\begin{figure}[ht!]
\centering
\begin{subfigure}{.7\textwidth}
\includegraphics[width=0.47\textwidth]{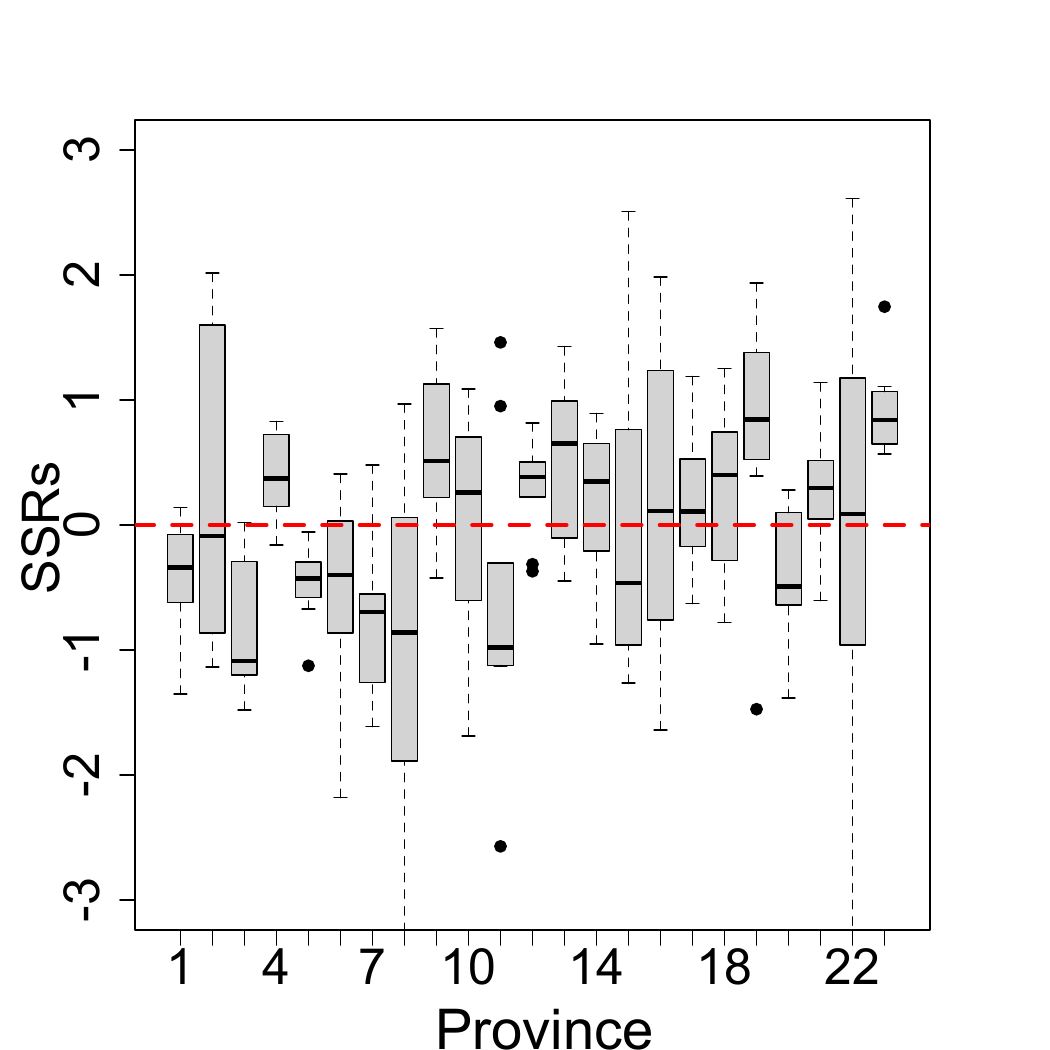}
\includegraphics[width=0.47\textwidth]{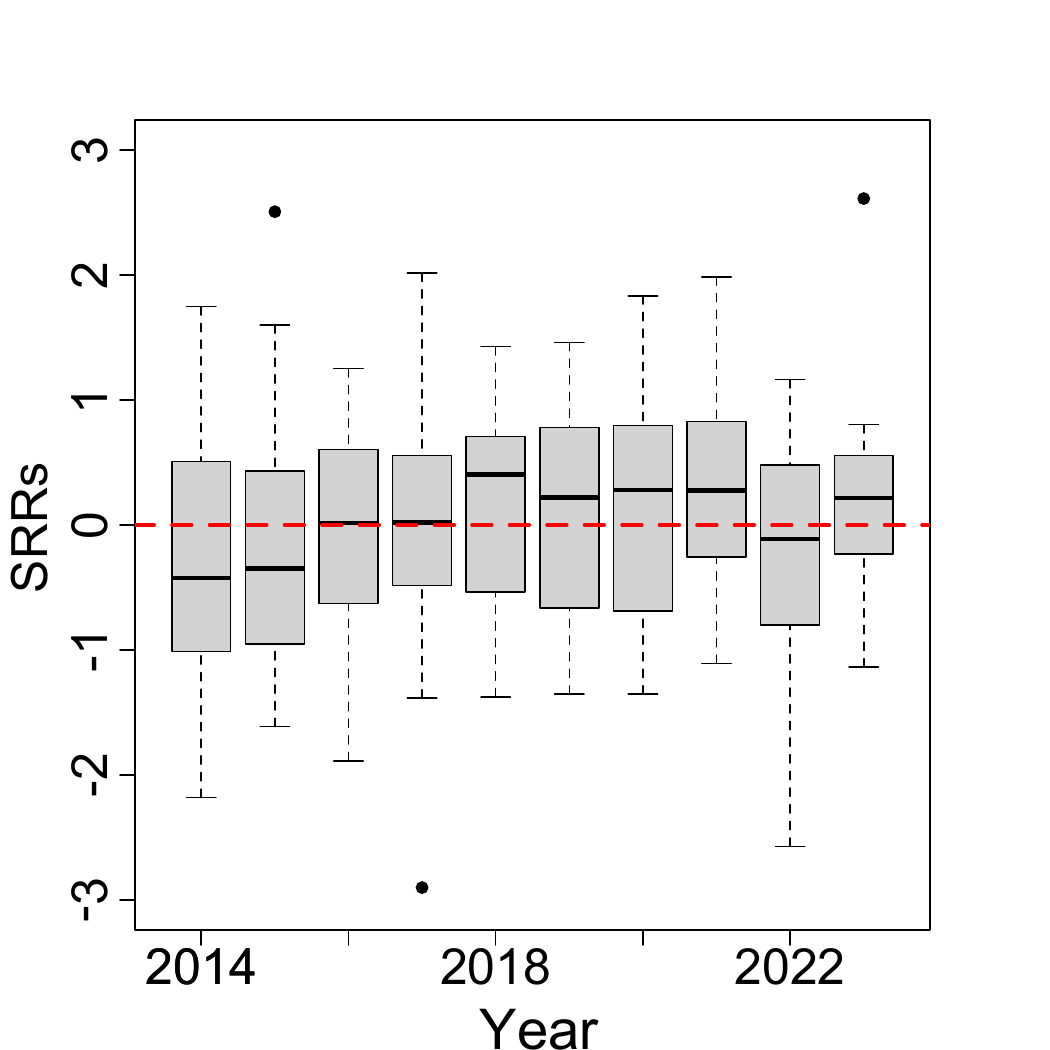}
\end{subfigure}
\caption{Boxplot of the SSRs by province (left) and year (right).}
\label{ASR}
\end{figure}

\subsection{Additional mapping results}\label{sec.aplic.suple}

This section includes the results for the TMQ predictor in the application to real data. Figure \ref{map.Income.SM} maps the equivalized disposable income for Empty Spain in 2013 (left), 2018 (center) and 2022 (right). The interpretation is in line with that discussed in Section \ref{pred.error}.
Nevertheless, a comparison of Figures \ref{map.Income} and \ref{map.Income.SM} reveals valuable insights because it is helpful in order to identify the differences in magnitude between the two small area predictions in some provinces and years. The latter evidences the unacceptable prediction bias of the TMQ predictor, as has been discussed repeatedly throughout the paper.

\begin{figure}[ht!]
\centering
\hspace{-13mm}
\includegraphics[width=0.39\linewidth]{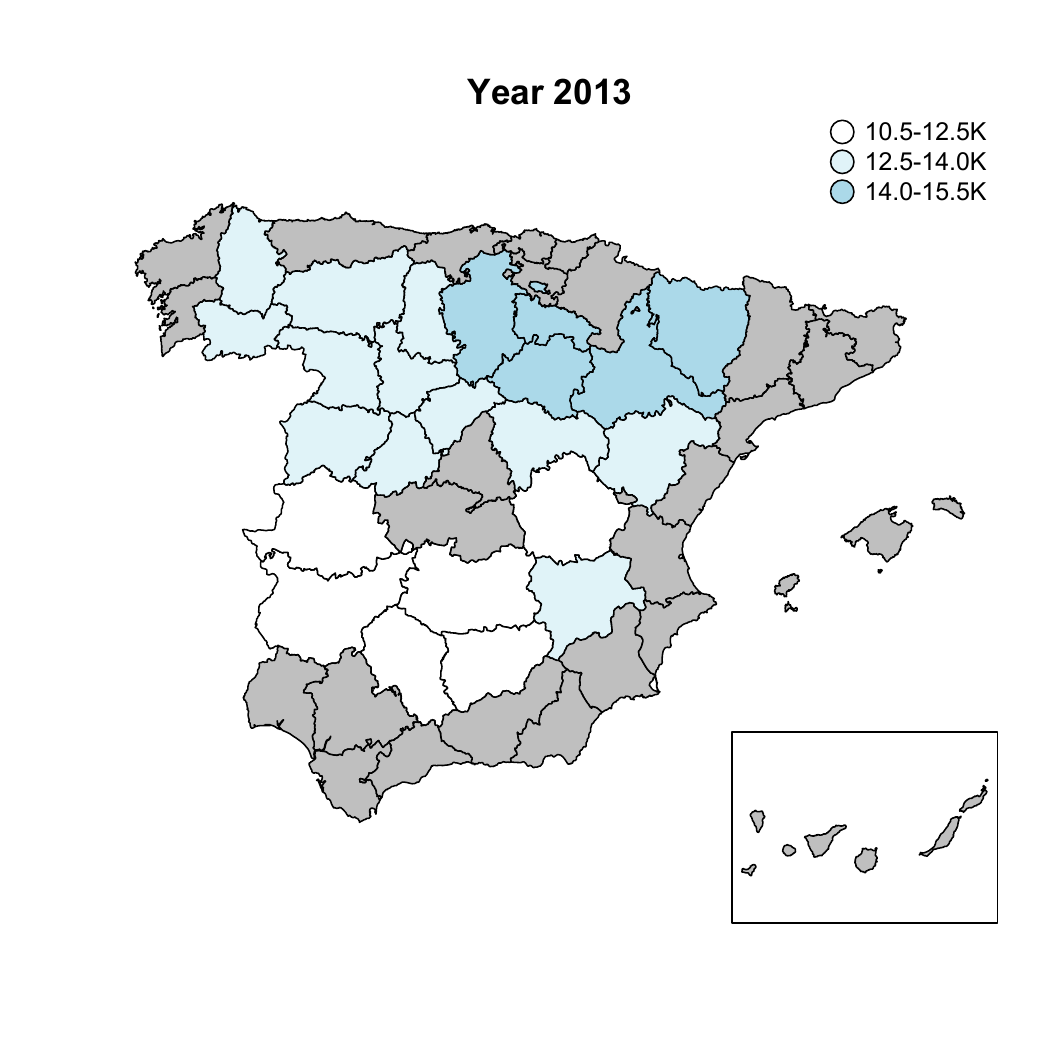}
\hspace{-11.5mm}
\includegraphics[width=0.39\linewidth]{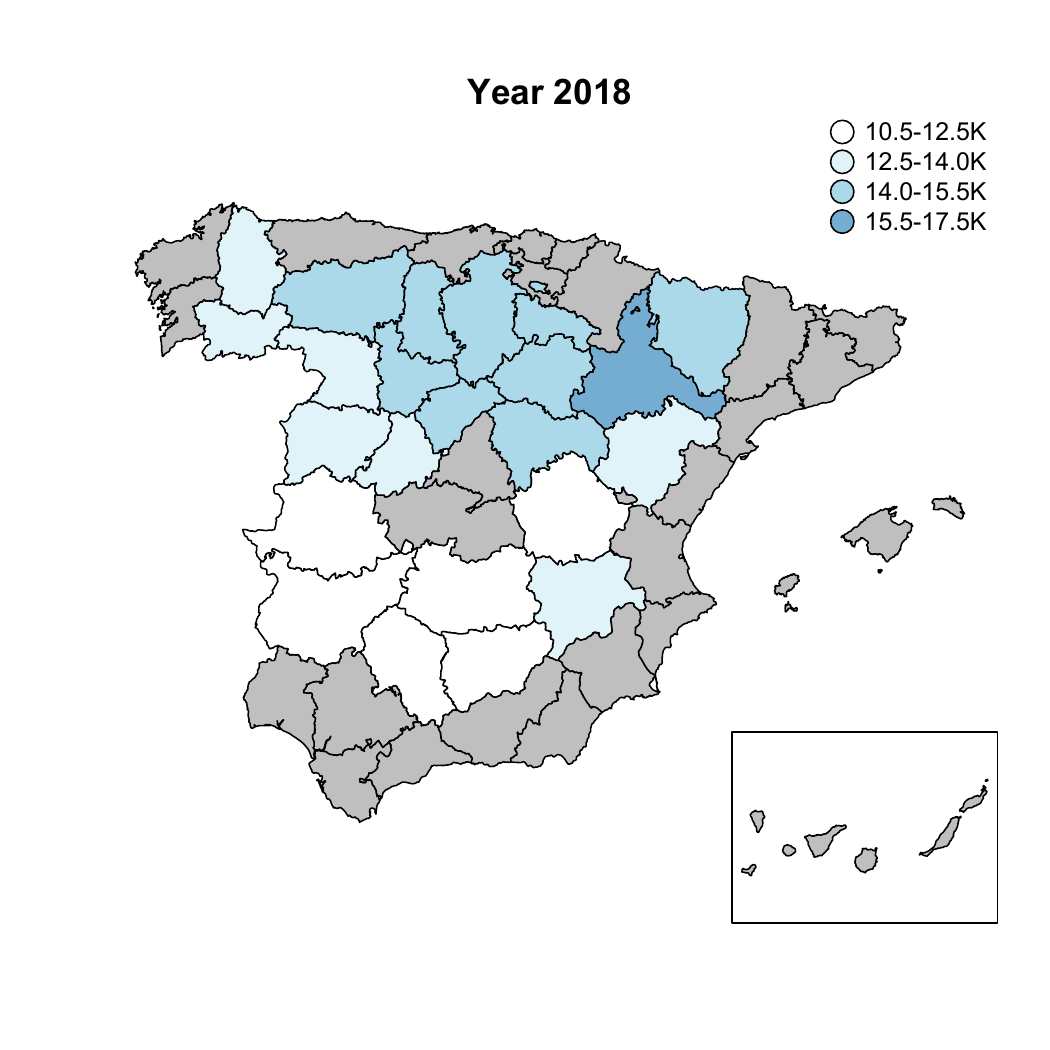}
\hspace{-11.5mm}
\includegraphics[width=0.39\linewidth]{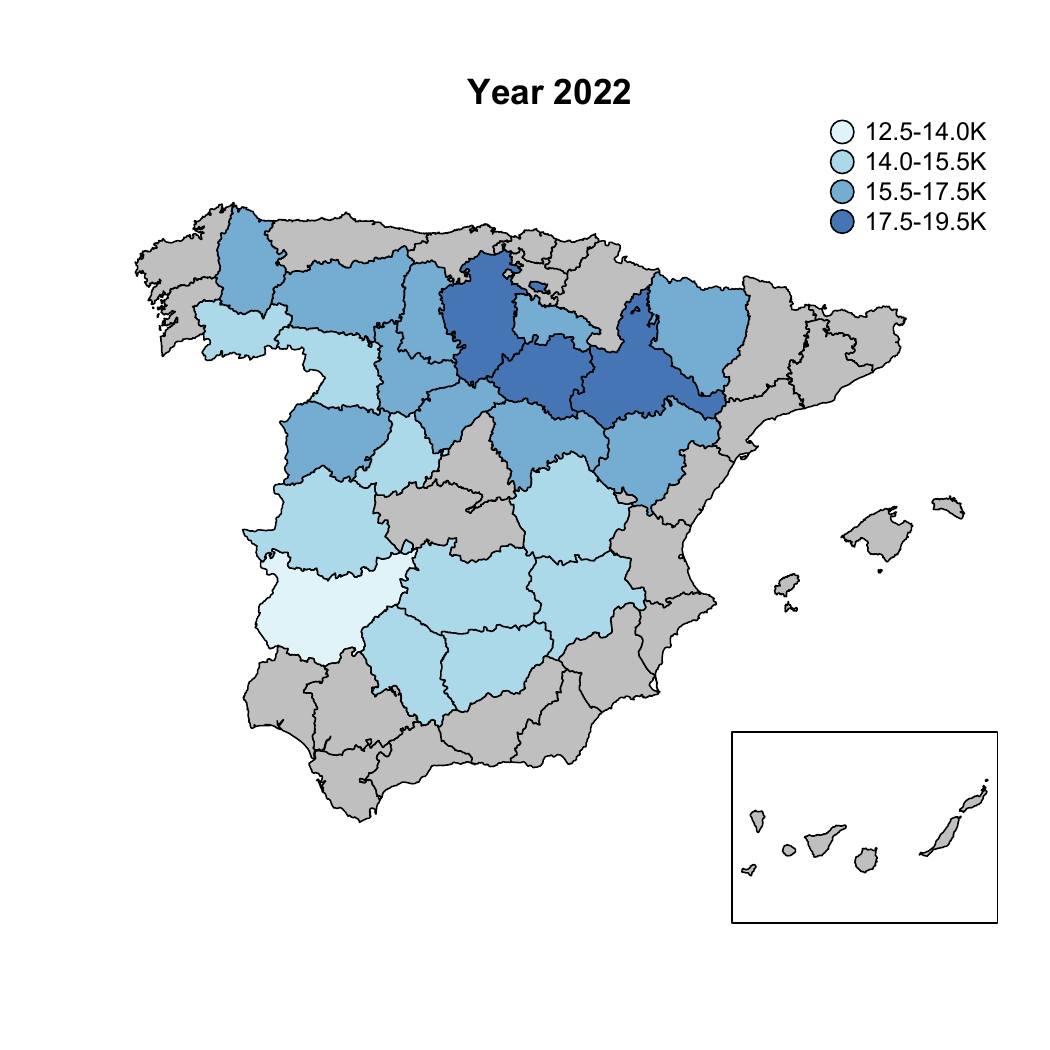}
\hspace{-10mm}
\caption{Estimates of equivalized disposable income for Empty Spain in 2013 (left), 2018 (center) and 2022 (right). Results obtained using the TMQ predictor.}
\label{map.Income.SM}
\end{figure}

For the sake of completeness, Figure \ref{map.RRMSE.SM} maps the coefficients of variation for the predictions in Figure \ref{map.Income.SM}, using the $rmse_{22,dt}^{tmq}$ estimator proposed in Section \ref{MSE.TMQ}. More significantly than the differences in the predictions themselves are the differences in their coefficients of variation. Because of the bias, the error is much higher than that shown in Figure \ref{map.RRMSE} for the BTMQ predictor.\\

\begin{figure}[ht!]
\centering
\hspace{-13mm}
\includegraphics[width=0.39\linewidth]{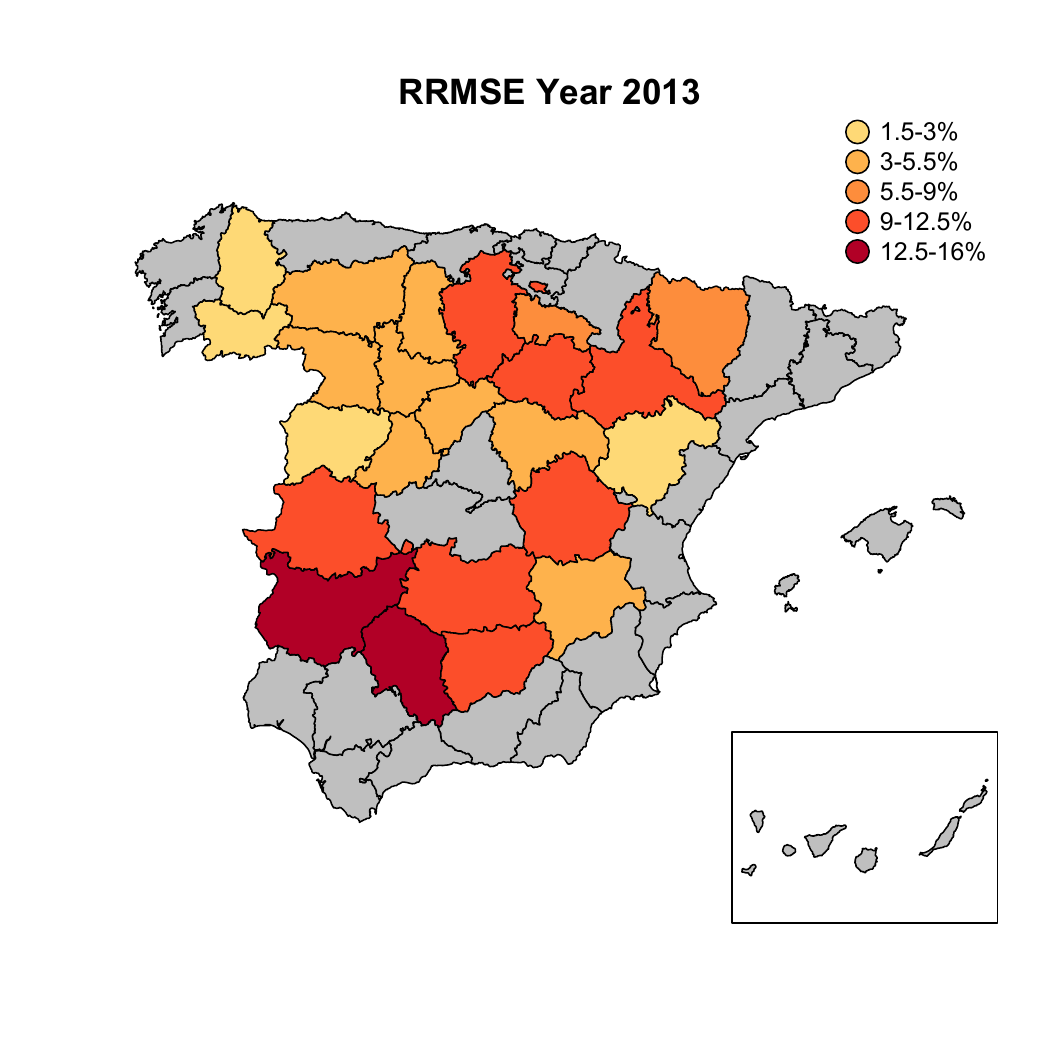}
\hspace{-11.5mm}
\includegraphics[width=0.39\linewidth]{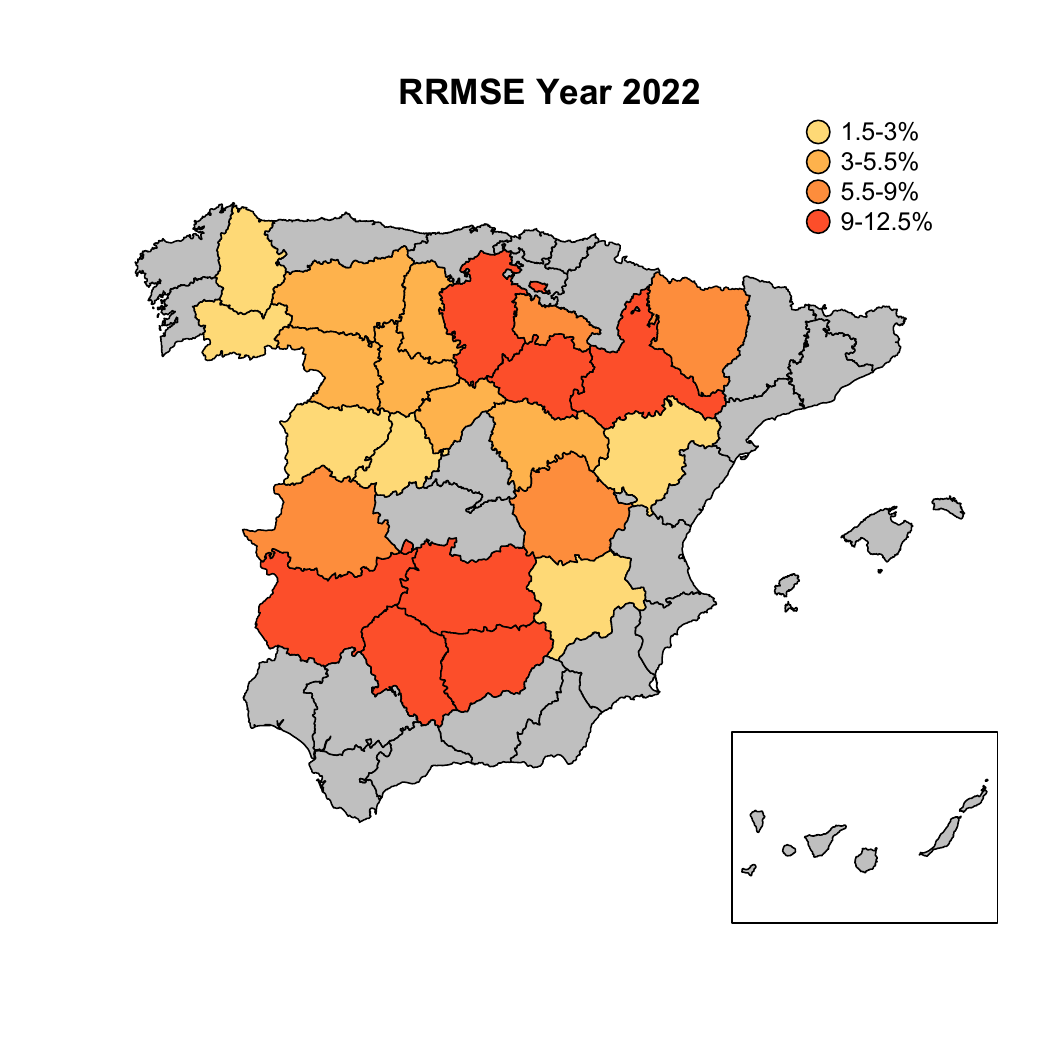}
\hspace{-11.5mm}
\includegraphics[width=0.39\linewidth]{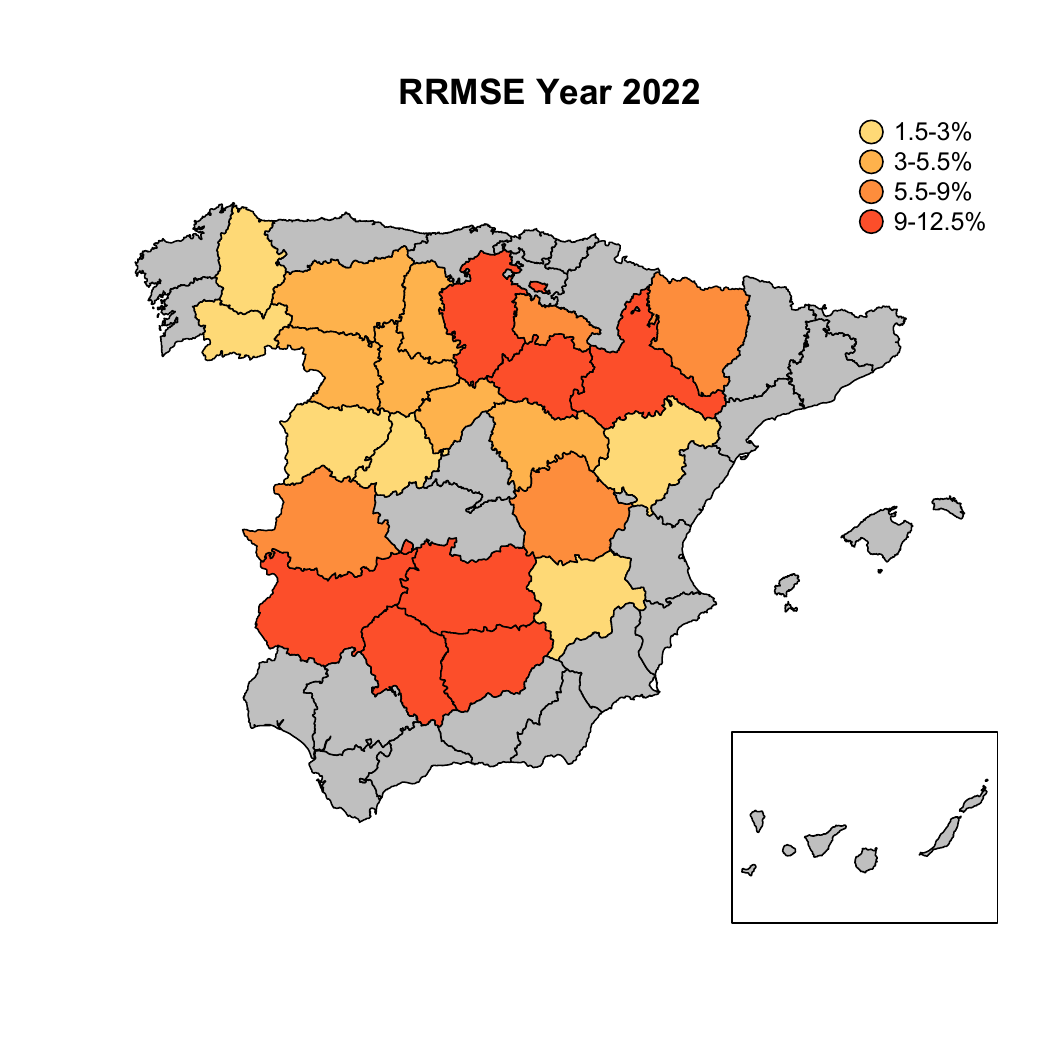}
\hspace{-10mm}
\caption{Coefficients of variation of the equivalized disposable income for  Empty Spain in 2013 (left), 2018 (center) and 2022 (right). Results for the TMQ predictor.}
\label{map.RRMSE.SM}
\end{figure}

\newpage

\section*{References}

\begin{itemize}
\setlength{\textwidth}{31pc}

\item[]
Bianchi A., Salvati N. (2015). Asymptotic properties and variance estimators of the M-quantile regression coefficients estimators. {\it Communications in Statistics 44}, 2416-2429.

\item[] Friedman, M. (1937). The use of ranks to avoid the assumption of normality implicit in the analysis of variance. {\it Journal of the American Statistical Association 32}, 200, 675-701.

\item[] Huber, P. (1981). {\it Robust Statistics}. John Wiley  and Sons.

\item[] Tukey, J. (1949). Comparing individual means in the analysis of variance. {\it Biometrics 5}, 99-114.

\item[]
Welsh, A.H. (1986). Bahadur representations or robust scale estimators based on regression residuals.
{\it Annals of Statistics 14}, 1246-1251.

\item[] Zewotir, T., Galpin, J. (2007). A unified approach on residuals, leverages and outliers in the linear mixed model. {\it TEST 16}, 1, 58-75.

\end{itemize}

\end{document}